\documentclass[prx,twocolumn,showpacs,superscriptaddress,preprintnumbers,amssymb]{revtex4-2} 
\usepackage{graphicx}
\usepackage{latexsym}
\usepackage{amssymb}
\usepackage{amsmath}
\usepackage{mathtools}
\usepackage{amsmath} 
\usepackage{amsfonts} 
\usepackage{upgreek}
\usepackage{bm}
\usepackage{booktabs}
\usepackage{microtype}
\usepackage{lmodern}
\usepackage[percent]{overpic}
\usepackage{makecell}

\usepackage{multirow}
\usepackage[shortlabels]{enumitem}
\usepackage{color}
\usepackage[colorlinks, citecolor=blue, linkcolor=blue]{hyperref}
\usepackage{float}
\usepackage{xcolor}
\usepackage{physics}
\usepackage{manfnt}
\usepackage{lipsum}
\usepackage[ruled,linesnumbered]{algorithm2e}
\usepackage{tabularx}
\usepackage{nicematrix}
\usepackage{braket}
\usepackage{bm}
\usepackage[normalem]{ulem}
\usepackage{tikz}
\usepackage[most]{tcolorbox}
\tcbuselibrary{breakable}

\newcommand{\stagefigplaceholder}[2][]{%
  \IfFileExists{#2}{\includegraphics[#1]{#2}}{%
    \fbox{\parbox{0.88\linewidth}{\centering\vspace{0.6em}
    Figure placeholder for \texttt{\detokenize{#2}}.\\
    The corresponding PDF is expected in the Overleaf/project figure folder.\\
    \vspace{0.6em}}}}%
}

\DeclareRobustCommand{\cylindericon}[1][1.65]{%
  \mathord{\vcenter{\hbox{%
    \begin{tikzpicture}[
      x=1em,
      y=1em,
      scale=#1,
      line cap=round,
      line join=round,
      line width=0.055em
    ]
      \def\cylrx{0.34}%
      \def\cylry{0.11}%
      \def\cylht{0.72}%

      \path[use as bounding box] (-0.43,-0.15) rectangle (0.43,0.87);

      \draw[dash pattern=on 0.55pt off 1.50pt]
        (-\cylrx,0)
        arc[start angle=180,end angle=0,
            x radius=\cylrx,y radius=\cylry];

      \draw (-\cylrx,\cylht) -- (-\cylrx,0);
      \draw ( \cylrx,\cylht) -- ( \cylrx,0);

      \draw
        (-\cylrx,0)
        arc[start angle=180,end angle=360,
            x radius=\cylrx,y radius=\cylry];

      \draw (0,\cylht)
        ellipse [x radius=\cylrx,y radius=\cylry];
    \end{tikzpicture}%
  }}}%
}

\usepackage{xurl}

 \newtcolorbox{axiombox}[1][]{
  breakable,
  colback=gray!15,    
  boxrule=0pt,     
  arc=3mm,           
  left=2mm,          
  right=2mm,         
  top=2mm,           
  bottom=2mm,        
  auto outer arc,
  #1
}

\usepackage{pifont} 

\usepackage{xcolor}
\usepackage{physics}
\usepackage{tcolorbox}
\usepackage{manfnt}
\usepackage{lipsum}
\usepackage[ruled,linesnumbered]{algorithm2e}
\usepackage{tabularx}
\usepackage{nicematrix}
\usepackage{braket}
\usepackage{bm}
\usepackage[normalem]{ulem}
\usepackage{tikz}
\usepackage{here}
\usepackage{amsthm}

\usepackage{amsthm}

\usepackage{hyperref}
\usepackage{cleveref}

\theoremstyle{plain} 
\newtheorem{theorem}{Theorem}

\theoremstyle{definition} 
\newtheorem{definition}[theorem]{Definition}
 
\newtheorem*{remark}{Remark}
\newtheorem{exmp}{Example}
\makeatother

\theoremstyle{plain}

\newtheorem{lemma}[theorem]{Lemma}
 
\newtheorem{Proposition}[theorem]{Proposition}

\newtheorem{assumption}[theorem]{Assumption}

\definecolor{YSblue}{RGB}{65,105,225}

\definecolor{GPTcolor}{RGB}{0,50,140}

\definecolor{BSorange}{RGB}{140,50,0}

\definecolor{ZWcolor}{RGB}{140,0,50}

\definecolor{ZLcolor}{RGB}{140,0,90}

\newcommand{\calA}{{\mathcal A}}

\newcommand{\calE}{{\mathcal E}}

\newcommand{\calH}{{\mathcal H}}

\newcommand{\calN}{{\mathcal N}}

\newcommand{\calR}{{\mathcal R}}

\newcommand{\mfa}{{\mathfrak a}}
\newcommand{\mfb}{{\mathfrak b}}
\newcommand{\mfc}{{\mathfrak c}}

 \hypersetup{colorlinks=true, 
  breaklinks=true,
 linkcolor=green!50!black,
 citecolor=red!30!orange!90!black,
 filecolor=OliveGreen, 
 urlcolor=blue!60!cyan!80!green!96!black,
 filebordercolor={.8 .8 1}, 
 urlbordercolor={.8 .8 0}}

\definecolor{prxblue}{RGB}{10, 36, 106}

\begin{document} 

\title{\textbf{ Approximate Quantum Error Correction at Chiral Topological Edges}}

\def\UIUCPhysics{Department of Physics, Grainger College of Engineering, University of Illinois at Urbana-Champaign, Urbana, Illinois 61801, USA}
\def\UIUCicmt{The Anthony J. Leggett Institute for Condensed Matter Theory, University of Illinois at Urbana-Champaign, Urbana, Illinois 61801, USA}
\def\KIAS{Korea Institute for Advanced Study, Seoul 02455, South Korea}
\def\UIUCiquist{Illinois Quantum Information Science and Technology Center, University of Illinois at Urbana-Champaign, Urbana, Illinois 61801, USA}


\author{Yuntai Song}
\affiliation{\UIUCicmt}
\affiliation{\UIUCPhysics}
\author{Zejun Liu}
\affiliation{\UIUCicmt}
\affiliation{\UIUCPhysics}
\author{Zhencheng Wang}
\affiliation{\UIUCPhysics}
\author{Jong Yeon Lee}
\affiliation{\UIUCicmt}
\affiliation{\UIUCPhysics}
\affiliation{\UIUCiquist}
\affiliation{\KIAS}
\author{Bowen Shi}
\affiliation{\UIUCicmt}
\affiliation{\UIUCPhysics}
\affiliation{\UIUCiquist}

\date{\today} 

\begin{abstract}

Topologically ordered phases naturally realize quantum error correction through nonlocal encoding of quantum information. More recently, conformal field theories have been shown to realize approximate quantum error-correcting codes, but such constructions generally require fine tuning to criticality. Here we introduce a family of approximate quantum error-correcting codes realized by the chiral edges of two-dimensional topologically ordered phases. The proposed encoding combines the robustness of a gapped topological bulk with the flexibility of gapless edge conformal field theories.
To characterize its robustness, we study coherent-information loss under local erasure. We derive an exact expression relating coherent-information loss to relative entropy, reducing the recoverability problem to universal properties of the edge theory. This leads to power-law scaling of coherent-information loss with the size of the erased region. 
We further show that, for geometrically local erasures near one edge, the two-dimensional chiral edge code is at least as robust as the dimensionally reduced CFT code, and is strictly more robust in several representative examples.
For Abelian code subspaces, we further construct a power-law-range recovery map supported on the erased region together with a power-law-range buffer; this recovery map depends only on the code subspace, not on the unknown encoded state. We provide numerical calculations for lattice realizations of compact free boson and Ising CFT examples that support the theoretical predictions of the power-law exponents.
\end{abstract}

\maketitle
\begingroup
\makeatletter
\renewcommand*\l@subsubsection[2]{}
\makeatother
\endgroup

\section{Introduction}


Quantum error-correcting (QEC) codes protect logical information by making it inaccessible to sufficiently local physical degrees of freedom. A robust quantum memory should protect the encoded state against errors affecting a small number of qubits. For an exact QEC code, this robustness is quantified by \emph{code distance}: all erasures below a certain size are perfectly correctable~\cite{PhysRevA.52.R2493,1997PhRvA..55..900K}. Topologically ordered phases, such as a toric code, naturally realize such robustness, as local observables cannot distinguish different topological sectors up to corrections exponentially small in system sizes~\cite{Wen1990,PhysRevB.41.9377,Kitaev_2003,Dennis_2002,RevModPhys.87.307}. From a condensed-matter perspective, such protection is tied to the robustness of gapped topological phases under any local perturbations~\cite{HastingsWen_2005,ChenGuWen_2010,BravyiHastingsMichalakis_2010,PhysRevX.2.021004,PhysRevB.104.235151,placke2024topologicalquantumspinglass,lavasani2024stabilityklocalquantumphases,deroeck2024ldpcstabilizercodesgapped}.

\begin{figure*}[!t]
    \centering
    \includegraphics[width=0.92\linewidth]{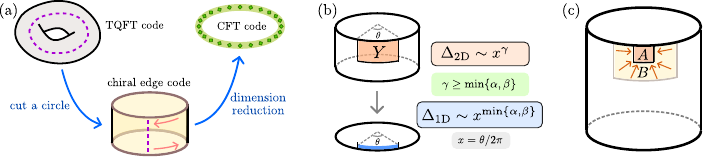}
    \caption{{\bf Summary of the main results.} (a) Relation between three many-body routes to quantum memory. 
    The chiral edge code proposed here (Def.~\ref{def:chiral-edge-code}) interpolates between a TQFT code and a CFT code: it inherits topological protection from the gapped bulk and CFT-like local distinguishability from the gapless edge. 
     (b) Coherent information loss (Def.~\ref{def:Ic-loss}) for 2D chiral edge code and 1D CFT code counterpart. As a main result, we show that the 2D chiral edge code is at least as robust as the dimensionally reduced counterpart under local erasure, with strict enhancement in representative examples (Prop.~\ref{prop:powers} and Sec.~\ref{sec:power_law_numerics}). (c) A local erasure on region $A$ near the edge can be recovered using a channel supported on $A$ together with a power-law-range buffer $B$ when the anyon sectors in the code subspace are Abelian (Thm.~\ref{thm:quasi-local-recovery}). 
    }
    \label{fig:summary}
\end{figure*}

In fact, the notion of exact robustness mostly applies in the limit of zero correlation length. In order to properly discuss such robustness in the \emph{phases of matter}, the idea of exact QEC should be extended beyond, namely approximate quantum error-correcting (AQEC) codes: a small amount of logical information is allowed to be probed locally as long as the amount vanishes in an appropriate limit~\cite{Leung_1997,Schumacher2002,BarnumKnill2002,Kretschmann2008,NgMandayam2010,PhysRevLett.104.120501,2017Quant...1....4F,Yi_2024}. This viewpoint is natural in quantum many-body physics as well, where low-energy states are rarely exactly locally indistinguishable. A useful way to characterize the robustness of an AQEC code is through local distinguishability: how much information about the encoded state can be extracted from observables supported on a small physical region. Several diagnostics of local distinguishability have been introduced and analyzed~\cite{2017Quant...1....4F,Bentsen_2024,Yi_2024,2025arXiv251004453Y}. These quantities often exhibit characteristic scaling with the size of the erased region. When such scaling can be derived or bounded, it yields quantitative constraints on code performance, including the accuracy of approximate recovery~\cite{2020PhRvX..10d1018F,Zhou2020,Bentsen_2024,2025arXiv251004453Y}, the spatial support required for a recovery channel~\cite{2017Quant...1....4F}, and the circuit complexity of the code states~\cite{Yi_2024}. While AQEC codes from gapped states in the same phase as the toric code provide an example with exponential decay of local distinguishability diagnostics, critical systems with vanishing energy gap may provide examples with power law decay diagnostics.


A canonical example of AQEC in a gapless system is provided by conformal field theory (CFT). Early connections between critical systems and quantum error correction appeared in both holographic settings~\cite{Almheiri_2014,Yoshida2015} and studies of the critical Ising model, where nontrivial recovery from local erasure was demonstrated by interpreting its Gibbs state as a quantum source-channel code~\cite{Pastawski_2017}. More recently, low-energy subspaces of $1+1$D CFTs were shown to form AQEC codes whose recovery properties are controlled by universal CFT data~\cite{Sang2024AQECC}. Related approximate error-correcting structures have also been identified in chaotic quantum systems~\cite{BaoCheng2019,2019PhRvL.123k0502B,Bentsen_2024,JWKim2024}.

The gapless nature of a CFT leads to a form of protection qualitatively different from that of a gapped topological phase. Because a CFT has no finite correlation length, a finite interval is generally not exactly blind to the encoded logical state. Instead, its local distinguishability typically vanishes as a power of the ratio between the interval size and the total system size, rather than exponentially. Consequently, the Knill--Laflamme conditions~\cite{1997PhRvA..55..900K} are satisfied only asymptotically, with an accuracy governed by universal CFT data~\cite{Zou2019,Yi_2024,Sang2024AQECC}.

Despite their universal error-correcting properties, CFT codes face an important limitation as physical quantum memories: the underlying gapless theory is realized only by tuning a lattice model to a critical point. Exact symmetries may forbid some relevant perturbations and reduce the amount of fine-tuning required~\cite{1987PhRvB..36.5291A,2017PhRvL.118b1601F,2019PhRvL.123r0201Y,2023PhRvD.107l5025L}, but the gaplessness is generally not robust to arbitrary local perturbations. This motivates the question addressed in this work: can conformal degrees of freedom furnish approximate quantum error correction when their existence is guaranteed by a stable phase of matter?

Chiral topological order provides a natural setting. Its bulk is gapped and supports anyonic superselection sectors, while its physical edge is gapless (ungappable) because of the non-vanishing chiral central charge~\cite{Wen1995-review,Kitaev2005}. Fractional quantum Hall states and chiral spin liquids are canonical examples~\cite{Laughlin1983,Kalmeyer-Laughlin1987,Moore1991,Nielsen2012,Eck2024chiral}. For the completely chiral edges considered here, the low-energy edge theory is described by a chiral CFT. Unlike a critical one-dimensional lattice model, the existence of the edge mode is enforced by the bulk topology rather than by tuning a Hamiltonian to a critical point. Thus a chiral topological phase combines two ingredients that are usually separated: a stable gapped bulk and universal gapless conformal edge degrees of freedom.

We use this bulk-edge combination to define a class of AQEC codes, which we call \emph{chiral edge codes} (Def.~\ref{def:chiral-edge-code}). Place a chiral topological order on a cylinder, the two physical edges carry counter-propagating chiral edge modes separated by a gapped bulk. Among the low-energy states without bulk excitations, we focus on \emph{cylinder primary states} (see Eq.~\eqref{eq:primaryLR}) labeled by a bulk anyon sector and by a choice of primary state on each physical edge, as explained in Sec.~\ref{sec:theory:chiral}. The choice of code subspace is given by selecting a finite set of anyon sectors, choosing one cylinder primary state in each selected sector, and taking their span. 

The construction combines the topological-sector structure of a topological code with the power-law local distinguishability of a 1D CFT code, as illustrated in Fig.~\ref{fig:summary}(a). A noncontractible bulk annulus can resolve the encoded anyon sector (see Thm.~\ref{thm:Delta_as_avg_relent_relent}), while a
contractible region detached from the edge is locally indistinguishable, up to a finite-size correction. However, a local region on physical edge can distinguish different cylinder primary states at finite-size, with the corresponding information leakage vanishing only in power-law. Chiral edge codes therefore exhibit a hybrid form of protection: algebraic suppression of local distinguishability along the gapless edges, together with the stronger exponential suppression in the gapped bulk.



The central questions for chiral edge codes discussed in this work are how to quantify its robustness against local noise and whether the resulting recovery channel can itself be supported near the noisy region. In the next section, we summarize our answers to these two questions. We first characterize local information leakage by \emph{coherent information loss} as in Def.~\ref{def:Ic-loss} and relate it to UV-finite relative entropies that can be studied through dimensional reduction (Assumption~\ref{asmp:UVfinite_entanglement_map_columns}) for chiral edge codes. We then establish a
power-law-range recovery map for Abelian code subspaces; see Sec.~\ref{sec:quasi-local-R}.


\section{Summary of Main Results}

We now summarize the main results about chiral edge codes and the assumptions entering them. Our primary local diagnostic is the \emph{coherent-information loss} \cite{Schumacher2002} under erasure. Given a code subspace \(\mathbb{V}\), we introduce a reference system \(R\) to purify the maximally mixed code state on $\mathbb{V}$. For a physical region \(A \subset Q\) (see Fig.~\ref{fig:Ic-A}), we define the coherent information loss $\Delta(A;\mathbb{V})$ as the mutual information between $A$ and $R$ (Def.~\ref{def:Ic-loss}). It quantifies the logical information leaked to the erased region. Moreover, it upper bounds the coherent-information deficit
generated by any noise channel supported on $A$, with equality for
complete erasure, and therefore gives a channel-independent diagnostic for noise on a given support. It vanishes for exactly correctable erasures, while for an AQEC code it may remain nonzero at
finite size. The relevant question is then how it scales with the
geometry of the erased region and with the overall system size.

\vspace{5pt} \noindent {\bf Relative entropy expression.}
The key simplification for chiral edge codes is provided by
topological-sector orthogonality. If the complement of the erased region contains a noncontractible bulk annulus, that annulus resolves the anyon sector up to the finite-correlation-length correction~\cite{Shi2019fusion}. Consequently, this diagnostic can be converted into a single average of relative entropies on the erased region (Thm.~\ref{thm:Delta_as_avg_relent_relent}). 
This expresses coherent information loss in terms of a UV-finite quantity. For full-column regions, our dimensional-reduction assumption (Assumption~\ref{asmp:UVfinite_entanglement_map_columns}) identifies this relative-entropy average with the corresponding quantity in the associated one-dimensional CFT, allowing CFT methods to determine its scaling.

\vspace{5pt} \noindent {\bf Power-law exponents.}
Our analysis of the coherent information loss leads to three power law exponents $\alpha, \beta, \gamma$ which we define below; see Fig.~\ref{fig:summary}(b) for a summary.
For an erasure $Y$ near a single edge of the two-dimensional system, the coherent-information loss scales as $ \Delta_{\rm 2D} \sim x^{\gamma}$. By contrast, the dimensionally reduced one-dimensional CFT code has $\Delta_{\rm 1D} \sim x^{\min\{\alpha,\beta\}}$, for an erased interval of the same angular size $x \ll 1$. The exponents \(\alpha\) and \(\beta\) are defined by relative entropies of small full-column regions $X$ and their complements on the cylinder, and therefore admit a purely one-dimensional CFT interpretation under dimensional-reduction assumption (Assumption~\ref{asmp:UVfinite_entanglement_map_columns}), whereas \(\gamma\) is intrinsic to a genuinely two-dimensional region localized near a single edge.

\vspace{5pt} \noindent {\bf Enhanced robustness.}  
We show that the chiral edge code has an enhanced robustness compared with the 1D CFT code obtained by its dimensional reduction.
Our main result (Prop.~\ref{prop:powers}) establishes the hierarchy ($\gamma \ge \alpha\geq \min\{\alpha,\beta\}$), which implies that the two-dimensional encoding is no worse against geometrically local erasures than its dimensionally reduced 1D counterpart. The hierarchy is strict in several of the examples studied
in Sec.~\ref{sec:power_law_numerics}. The physical reason is that dimensional reduction treats a full column of the cylinder as a local interval, while a local two-dimensional error near one edge cannot simultaneously access the opposite edge through the gapped bulk.

\vspace{5pt} \noindent {\bf Locality of Recovery.}  
Small coherent information loss implies the existence of an approximate recovery channel~\cite{Schumacher2002,PhysRevLett.104.120501}, but does not
by itself constrain its spatial support. For code subspaces containing only Abelian anyon sectors, we combine full-boundary entanglement bootstrap~\cite{Chiral-vira2024} 
conditions with the decay of mutual information between separated intervals in the dimensionally reduced CFT. Under these conditions, an $O(1)$-sized erasure adjacent to one edge can be approximately recovered by a channel supported on the erased region $A$ together with a buffer $B$ whose length grows subextensively with the cylinder's circumference; see Fig.~\ref{fig:summary}(c) and Sec.~\ref{sec:quasi-local-R}.
The recovery map can be chosen as a universal twirled Petz map constructed from the maximally mixed state on the code subspace~\cite{Fawzi2015,JRSWW-universal-recovery}; it therefore depends on the code subspace and the erased region, but not on the unknown encoded state or on the particular local noise channel supported on $A$.


\vspace{5pt} \noindent {\bf Chiral semion and Ising codes.} We examine the predicted power-law behavior in two representative examples: the chiral semion code, realized by bosonic Laughlin-type states and dimensionally reduced to the compactified free-boson CFT, and the chiral Ising code, whose dimensional reduction gives the Ising CFT. For the corresponding 1D CFT codes, we extract the exponents $\alpha$ and $\beta$ numerically and compare them with analytical CFT predictions for the relative entropies. For the chiral semion code, we additionally estimate the genuinely two-dimensional exponent $\gamma$ directly using bosonic Laughlin-state wave functions on finite cylinders. The results support the predicted hierarchy $\gamma \geq \alpha \geq \min \{\alpha, \beta\}$ with strict separation
for several code subspaces. Finally, for the Abelian code subspaces, we
estimate the CFT exponents that control the buffer range and the
asymptotic error in the power-law-range recovery theorem.

\vspace{5pt} \noindent {\bf Organization.} The remainder of the paper is organized as follows. Section~\ref{sec:background} develops the information-theoretic
framework, introducing coherent-information loss and establishing the
bound for arbitrary noise channels with a local support. Section~\ref{sec:Chiral edge of 2D topological order: theory and examples} defines chiral edge codes and explains their bulk-edge structure and dimensional-reduction correspondence. Section~\ref{sec:coherent_loss_vs_relent} derives the relative-entropy formula and the resulting robustness hierarchy of power laws, while Section~\ref{sec:power_law_numerics} examines the chiral semion and chiral Ising examples. Section~\ref{sec:quasi-local-R} establishes the power-law-range recovery theorem, and Section~\ref{sec:discussion} concludes with implications and open directions.

\section{Background}\label{sec:background}

In this section, we provide necessary background, including the coherent information in the study of error correcting properties. We further provide a general bound on the coherent information loss under local error channels, which will be useful in later sections.

\subsection{Coherent information}
\label{subsec:coherent_information}
 
Consider a quantum code on a physical system $Q$, and information stored in a code subspace $\mathbb{V}=\text{span}\{|\psi_Q^i\rangle\}_{i=1}^D$, where $\langle \psi_{Q}^i|\psi_Q^j\rangle = \delta_{i,j}$, and $D$ is the number of logical states in the code subspace.
Suppose we want to study the quantum error correcting property, which is the ability to preserve quantum information under decoherence. It is important to understand if there is a decoding channel that recovers the original state. Let
\begin{equation}
    \left|\psi_{R Q}\right\rangle:=\frac{1}{\sqrt{D}} \sum_{i=1}^D|\psi_Q^i\rangle\otimes\left|i_R\right\rangle
    \label{eq:maximal_entangled}
\end{equation}
be the maximally entangled state between the code words and the reference qubit(s) $R$, where $|i_R\rangle$ is a computational basis of $R$. Let $\calN_Q$ be the error channel, and $\calR$ be the recovery channel. Then the quantity of interest will be the entanglement fidelity~\cite{PhysRevA.54.2629}  
\begin{equation}
    F_{e}:=\max_{\calR} \bra{\psi_{RQ}}\operatorname{id}_R\otimes(\mathcal{R} \circ \calN_{Q})(\ket{\psi_{RQ}}\bra{\psi_{RQ}})\ket{\psi_{RQ}}
    \label{eq:ent_fidelity}
\end{equation}  
Consider the (von Neumann) coherent information 
\begin{equation}
    I_c(\rho_{QR}):=S(\rho_Q)-S(\rho_{QR}),
\end{equation}
where $S(\rho_A)\equiv -\Tr(\rho_A\log \rho_A)$ is the von Neumann entropy. A lower bound of $F_e$ can be obtained as~\cite{Schumacher2002} $F_e \geq 1-2\sqrt{\varepsilon}$, in which $\varepsilon\rightarrow 0$ is the \emph{deficit} of coherent information near its maximum $I_c=\log D$ with code subspace $\mathbb{V}$,  
\begin{equation}
\begin{aligned}\label{eq:Ic-drop-def}
    \varepsilon (Q;\mathbb{V}) &:= I_c (|\psi_{QR}\rangle) - I_c(\rho'_{QR})  
\end{aligned}  
\end{equation} 
where $\rho'_{QR}= (\operatorname{id}_R\otimes \calN_{Q})(\ket{\psi_{RQ}}\bra{\psi_{RQ}})$. Therefore, one way to characterize the loss of quantum information under the noisy channel is to bound the coherent information deficit $\varepsilon$.

Note that $\varepsilon \ge 0$ as one can rewrite
\begin{equation}
    \varepsilon (Q;\mathbb{V})=I(Q:R)_{|\psi_{QR}\rangle} -  I(Q:R)_{\rho'_{QR}},
\end{equation}
where $I(A:B)_\rho:= S(\rho_A) + S(\rho_B) - S(\rho_{AB})$ is the mutual information.
The non-negativity of $\varepsilon$ then follows from the monotonicity of mutual information under a quantum channel acting on one party.

\subsection{Coherent information loss on local regions} 

We shall be interested in the decoherence caused by local channels and the related coherent information loss. By a local channel we mean a quantum channel supported on a strict subset of the physical systems $A\subset Q$, where $A$ is geometrically local. The environment typically interacts with the quantum memory in a geometrically local way, and that is the main reason we care about local channels. 

First, we observe a general property of the coherent information loss for channels supported on a subsystem $A \subset Q$, applicable to any local channels and code subspace $\mathbb{V}$.

\begin{Proposition}[Bound on coherent information deficit]\label{prop:bound-1}
Let $A\subset Q$ be a subsystem, and suppose the decoherence is described by a local channel $\calN_{A}$. The deficit of coherent information (defined in Eq.~\eqref{eq:Ic-drop-def}), which we call $\varepsilon(A;\mathbb{V})$, is upper bounded by the mutual information as
\begin{equation}
    \varepsilon (A;\mathbb{V}) \le I(A:R)_{|\psi_{QR}\rangle}.
\end{equation} 
\end{Proposition} 
\begin{proof}
    Let $B = Q \setminus A$. We see
\begin{equation}
\begin{aligned}
    \varepsilon (A;\mathbb{V}) &= I(A:R|B)_{|\psi_{QR}\rangle} - I(A:R|B)_{\rho'}  \\
    &\le I(A:R|B)_{|\psi_{QR}\rangle} \\
    &= I(A:R)_{|\psi_{QR}\rangle}.
    \end{aligned}
\end{equation}
The equality in the first line follows from the fact that $|\psi_{QR}\rangle$ and $\calN_{A}(|\psi_{QR}\rangle \langle \psi_{QR}|)$ have identical reduced density matrices on $BR$.
The inequality in the second line follows from strong subadditivity. The last line follows from the purity of $|\psi_{AB R}\rangle$.
\end{proof}

The nice thing about the bound is that only the original state $|\psi_{QR}\rangle$ is needed. We avoid talking about the dependence of the local channel (i.e., error type). The bound is saturated for the strongest local channel, which means a complete corruption of certain qubit(s)~\footnote{One realization of the strongest channel is such that it replaces the local density by a product, e.g. $\rho'_{QR} = \lambda_{A} \otimes \rho_{BR}$, where $\lambda_{A}$ is an arbitrary density matrix and $\rho_{B R} = \Tr_{A} |\psi_{QR}\rangle \langle \psi_{QR}|$.}.  

\begin{definition}[Coherent information loss]\label{def:Ic-loss}
    We introduce the ``coherent information loss'' as the change of coherent information under an erasure channel on $A \subset Q$, as in Fig.~\ref{fig:Ic-A}:  
\begin{equation}\label{eq:upper_bound_CI_drop}
\begin{aligned}
     \Delta(A;\mathbb{V}) & := I_c \large(|\psi_{QR}\rangle) - I_c(\Tr_A(|\psi_{QR}\rangle\langle \psi_{QR}|) \large) \\
     & = I(A:R)_{|\psi_{QR}\rangle}.    
\end{aligned}
\end{equation}

\end{definition}

The bound (Prop.~\ref{prop:bound-1}) is simply $\Delta(A;\mathbb{V}) \ge \varepsilon(A;\mathbb{V})$ for any quantum channel on $A$.

\begin{figure}
    \centering
    \includegraphics[width=0.56\linewidth]{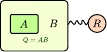}
    \caption{Regions related to the definition of coherent information loss. $Q=AB$ is the physical system of the quantum code. $R$ is the reference for purification purposes.}
    \label{fig:Ic-A}
\end{figure}

\begin{lemma}\label{lemma:delta} 
The coherent information loss for any $A \subset Q$ and code subspace $\mathbb{V}=\text{span}\{|\psi_Q^i\rangle\}_{i=1}^D$ has an alternative expression 
\begin{equation}
     \Delta(A;\mathbb{V})
     =
     \bigl(S_{A}+S_Q-S_{B}\bigr)_{\bar{\rho}_{Q}},
\label{eq:drop_from_rel_ent_relent}
\end{equation} 
where $\bar{\rho}_Q=\frac{1}{D}\sum_i \ket{\psi^i_Q}\bra{\psi^i_Q}$ is the maximum entropy state of the code subspace.
\end{lemma}
\begin{proof}
    It follows from Eq.~\eqref{eq:upper_bound_CI_drop} and the purity of $|\psi_{QR}\rangle$ that $\Delta(A;\mathbb{V}) =(S_{A}+S_Q-S_{B})_{|\psi_{QR}\rangle}$. Then, we notice that the reduced density matrix of $|\psi_{QR}\rangle$ on the physical system $Q$ is $\bar{\rho}_Q$.
\end{proof}
 
The basis-invariant nature of $\Delta(A;\mathbb{V}) $ reflects the fact that local robustness is ultimately a code subspace question: one must control how well a small physical region can distinguish an \emph{arbitrary} encoded state from the maximally mixed code state, not merely how it distinguishes a preferred basis of codewords.

Below we describe a few examples of the coherent information loss $\Delta(A;\mathbb{V}) $ associated with the erasure of local region $A$ and a code subspace $\mathbb{V}$ for different kinds of quantum memory. While some examples are for illustration purposes, we provide Example~\ref{exmp:AQECC} specifically to explain why such quantities are well motivated for AQEC codes.

\begin{exmp}[Stabilizer code (QEC codes)]
\label{exmp:stabilizer_code}
Consider a $[[n,k,d]]$ qubit stabilizer code. This means we have a code subspace spanned by an orthonormal set of states
$\mathbb{V}^{\rm stabilizer} = \text{span}\{|\psi^i\rangle\}_{i=1}^D$, where $D= 2^k$.
Let $A\subset Q$ be a set of physical qubits
with $|A|\le d-1$. By the Knill--Laflamme condition, every operator $O_A$ supported on $A$ obeys~\cite{1997PhRvA..55..900K,1997PhDT.......232G}
\begin{equation}
    \langle \psi^i|O_A|\psi^j\rangle=c(O_A)\,\delta_{i,j},
\end{equation}
in which $c(O_A)$ is a constant that depends on the operator $O_A$. With simple algebra, one can show that the erased region $A$ is uncorrelated with the reference system $R$. Thus, an exact relation holds
\begin{equation}
\Delta(A;\mathbb{V}^{\rm stabilizer})=I(A:R)_{|\psi_{QR}\rangle}=0 .
\end{equation}

\end{exmp}

\begin{exmp}[TQFT code (QEC codes)]
\label{exmp:TQFT_Code}
Consider a topologically ordered system on the torus. Let the code subspace be a set of degenerate ground states
\begin{equation}
    \mathbb{V}^{\text{TQFT}}(\mathbb{T}^2) =\text{span}\{ |\psi^\mfa\rangle\},
\end{equation}
where $|\psi^\mfa\rangle$ is a ground state on the torus corresponding to an anyon, and $\langle \psi^\mfa | \psi^\mfb\rangle = \delta_{\mfa,\mfb}$. We shall refer to this code as \emph{TQFT code}.
This code is a quantum error correcting code~\cite{Kitaev_2003,Dennis_2002}. We can verify this by showing that coherent information loss vanishes for any local errors.

Let $A$ be a local disk, as shown in Fig.~\ref{fig:T2-Q1}, and let $B= Q\setminus A$. By the property of the topological order ground state
\begin{equation}
\begin{aligned}
    \Delta(A; \mathbb{V}^{\text{TQFT}}(\mathbb{T}^2)) 
    &= ( S_{A} + S_Q - S_{B})_{\bar{\rho}} \\
    & \le ( S_{A} + S_{A B_-} - S_{B_-})_{\bar{\rho}} \\
    & = ( S_{A} + S_{A B_-} - S_{B_-})_{|\psi^i\rangle}\\
    & =0,
\end{aligned}
\end{equation}
where $\bar{\rho}_Q=\frac{1}{D}\sum_\mfa \ket{\psi^\mfa}\bra{\psi^\mfa}$. The first line is from Lemma~\ref{lemma:delta}; the 2nd line follows from the strong subadditivity, using $B_- \subset B$ as in Fig.~\ref{fig:T2-Q1}(b). The third line uses the fact that the ground states of topologically ordered systems are locally indistinguishable. The last line follows from the entanglement area law of gapped ground states; precisely speaking, the condition $(S_{A} + S_{A B_-} - S_{B_-})_{|\psi^i\rangle} =0$ is the entanglement bootstrap axiom {\bf A0}~\cite{Shi2019fusion}. 
For chiral topological order, the area law and the local indistinguishability should both be approximate on the ground states. Such errors are expected to decay exponentially with the size of regions. Thus, rigorously speaking, a TQFT code is an AQEC code whose local distinguishability decays exponentially. We neglect such errors in the rest of the work.

Consequently, the coherent information cannot decrease under any local decoherence. This means that local erasure errors of the TQFT code can be recovered by a certain quantum channel acting on the system $Q$.
Importantly, the topology of $A$ is important, and the erasure on it can be perfectly correctable even if its size $|A|$ is larger than the code distance, as long as it is a disk.

In fact, following from the same entropy analysis, one can show the existence of a \emph{local} recovery channel $\calR_{AB_-}$ that can recover any local channel $\calN_A$, that is
\begin{equation}\label{eq:local-recovery}
    \rho_Q = \calR_{AB_-}\circ \calN_A(\rho_Q)
\end{equation} 
for any state $\rho_Q$ in the code subspace.

\begin{figure}[H]
    \centering
    \includegraphics[width=0.89\linewidth]{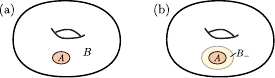}
    \caption{Local decoherence of TQFT code can be fully recovered. The physical system $Q$ is a torus. $A$ is a local disk $B= Q\setminus A$. The region $B_- \subset B$ is an annulus surrounding $A$, and thus $AB_-$ is a local disk that thickens $A$ by a few correlation lengths.}
    \label{fig:T2-Q1}
\end{figure}

The toric code~\cite{Kitaev_2003,Dennis_2002} is one of the lattice model examples of TQFT construction. On a torus, its four-dimensional ground-state subspace admits a minimally entangled basis
\begin{equation}
\mathbb{V}^{\rm TC}=\operatorname{span}\left\{\left|\psi^{\mathfrak{a}}\right\rangle: \mathfrak{a} \in\{1, e, m, \epsilon\}\right\}
\end{equation}
where $\mathfrak{a}$ labels the anyon flux through a chosen noncontractible cycle.

More generally, for a local Hamiltonian in the same gapped topological phase, connected to the fixed-point model by a gapped path and hence by quasi-adiabatic continuation~\cite{HastingsWen_2005,ChenGuWen_2010}, the exact local indistinguishability is replaced by an exponentially accurate one, as guaranteed by the stability of topological quantum order under local perturbations~\cite{BravyiHastingsMichalakis_2010}.
\end{exmp}

\begin{exmp}[AQEC codes: local diagnostics and local recovery]\label{exmp:AQECC}
Approximate quantum error correction can be quantified by several related criteria, including entanglement fidelity, coherent-information loss, information--disturbance tradeoffs, and recovery-map conditions~\cite{Leung_1997,Schumacher2002,Kretschmann2008,PhysRevLett.104.120501,BarnumKnill2002,NgMandayam2010}. For a certain quantum memory, one is often interested in a more refined question: how much logical information can be learned from a small physical region? 

For exact QEC codes, this question is answered sharply by the code distance $d$: every erased region $A$ for which $|A|<d$ is perfectly correctable. For an AQEC code, the corresponding statement is no longer sharp; instead, one obtains a tradeoff between the size and geometry of the region and the allowed recovery error. Several previous works have developed local diagnostics for quantifying this kind of tradeoff. For example, local approximate correctability and approximate code distance were formulated for lattice AQEC codes using Bures-distance recovery after erasure of a region~\cite{2017Quant...1....4F}; robustness is characterized by a local information-theoretic quantity called ``subsystem variance'' in Ref.~\cite{Yi_2024,2025arXiv251004453Y}; approximate code distance is defined by mutual information between a small physical region and the reference system in~\cite{Bentsen_2024}.

The common argument is that a good approximate memory should make sufficiently small subsystems nearly decoupled from the logical reference. This is the role played in this paper by the coherent-information loss $\Delta(A;\mathbb{V})$ under local erasure on $A$, which measures how much information about the reference is visible to the erased region. For exact QEC codes, \(\Delta(A;\mathbb V)=0\) for every correctable erasure region. 
For an AQEC code, \(\Delta(A;\mathbb V)\) may be nonzero at finite size, and the meaningful quantity is its scaling with the size and geometry of \(A\). 
Moreover, by Prop.~\ref{prop:bound-1}, \(\Delta(A;\mathbb V)\) upper bounds the coherent-information deficit of any quantum channel supported on \(A\), so it gives a channel-independent local diagnostic for local noise.

There is a second, stronger question: even if small information leakage guarantees the existence of some recovery map, must that recovery be geometrically local? 

General recovery theorems, including Petz-type and rotated-Petz recovery maps, usually provide an abstract recovery operation and do not by themselves bound its spatial support~\cite{PhysRevLett.104.120501,Fawzi2015,JRSWW-universal-recovery}. The local-recovery refinement asks whether erasure of \(A\) can be corrected by a channel supported only near \(A\), for example within an \(\ell\)-neighborhood~\cite{2017Quant...1....4F}. This locality of the recovery channel is nontrivial and is the question addressed in Sec.~\ref{sec:quasi-local-R}. 
For Abelian chiral edge code subspaces, we construct a power-law-range recovery map for local disk erasures at one of the physical edges, supported on the erased region together with a subextensive buffer.

\end{exmp}

\subsection{Effective code distance}
 
For an exact quantum error-correcting code, the code distance $d$ provides a sharp characterization of robustness: every erasure acting on fewer than $d$ physical qubits is perfectly correctable, while larger erasures need not be. Approximate quantum error-correcting codes generally do not possess such a threshold. The coherent-information loss may already be nonzero for arbitrarily small erased regions, even though it vanishes in the thermodynamic limit. Consequently, robustness is better described by how rapidly the coherent-information loss grows with the size of the erased region rather than by a single integer-valued distance. 

Motivated by this observation, we consider an effective code distance based on the coherent-information loss. Let $\Delta(A;\mathbb{V})$ denote the coherent-information loss of a code subspace $\mathbb{V}$ under the erasure of a region $A$. Given a tolerance $\delta>0$, we define the effective code distance~\cite{Bentsen_2024}
\begin{equation}
    d_{\rm eff}(\delta):=\min \{|A|:\Delta(A; \mathbb{V} ) \ge \delta\},
\end{equation}
where $|A|$ denotes the size of the erased region. Equivalently, $d_{\rm eff}(\delta)$ is the smallest erasure size for which the coherent information loss remains above the tolerance.  

This notion is closely related to recent proposals for extending the concept of code distance to approximate quantum error-correcting codes, where recoverability is characterized by information-theoretic quantities rather than exact correctability; see Ref.~\cite{2017Quant...1....4F,Bentsen_2024,Yi_2024,2025arXiv251004453Y}. 
In this paper, we use coherent information loss as our primary diagnostic for local recoverability.

A closely related effective-distance viewpoint will be useful in Sec.~\ref{subsec:theoretical-summary}, where the relevant notion of size is the angular length of an erased region along the cylinder rather than simply the number of erased microscopic degrees of freedom. There, the scaling of coherent-information loss with this geometric size will provide a quantitative way to compare the robustness of chiral edge codes with their dimensionally reduced CFT counterparts.

\section{Chiral edge code} 
\label{sec:Chiral edge of 2D topological order: theory and examples}

In this section, we put forward our proposal for AQEC codes utilizing the chiral edges of topological order. In order to do so, we first review the relevant notations in the modern theory of topological order and its gapless edges (Sec.~\ref{sec:theory:chiral}). We then formally introduce the chiral edge code in Sec.~\ref{sec:chiral-edge-code}; see Def.~\ref{def:chiral-edge-code}.
For instance, we can make a code subspace using a few low lying cylinder primary states $|\mfa, \phi^I_\mfa(L), \phi^J_{\bar{\mfa}}(R)\rangle$ corresponding to anyon $\mfa$ and the associated edge Virasoro primary fields $\phi^I_\mfa(L)$ and $\phi^J_{\bar{\mfa}}(R)$. Explicit examples of chiral topological order and their edge states can be found in Sec.~\ref{sec:examples-chiral-states}, which can be used to build a variety of explicit code subspaces.

\begin{figure}[h]
    \centering
\includegraphics[width=0.92\linewidth]{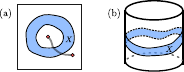}
    \caption{ Anyon charges detected by an annulus surrounding it, using the information convex set. (a) An annulus embedded in a disk. One should picture an anyon, or, more generally, a localized cluster of excitations with total topological charge $\mfa\in\calA$ sitting inside the disk enclosed by annulus of $X$. (b) The same intuition works for a noncontractible annulus on the cylinder. The annulus $X$ itself contains no anyon core.}  
    \label{fig:ICS-annulus}
\end{figure}

 \subsection{Chiral edge of topological order: the theory}\label{sec:theory:chiral}

We review the modern theory for topological order in 2D and its gapless chiral edges~\cite{Wen1995-review,Kitaev2005} to set up the relevant background and notation. 
We note that Fractional quantum Hall (FQH) systems are special types of symmetry enriched topological orders with $U(1)$ charge conservation symmetries.

\emph{The bulk:} Anyons in 2D gapped bosonic systems are expected to be classified by the unitary modular tensor category, and the chiral central charge $c_-$~\cite{Kitaev2005}. Let the set of anyons be a finite set
\begin{equation}
    \calA = \{ 1, \mfa, \mfb, \mfc, \cdots\}
\end{equation} 
with the fusion rules $\mfa\times \mfb = \sum_\mfc N_{\mfa\mfb}^\mfc \mfc$. $\{N_{\mfa\mfb}^\mfc\}$ are non-negative integers called the fusion multiplicities, and the fusion with the vacuum is trivial: $\mfa \times 1 = \mfa, \forall \mfa \in \calA$. 

The labels $\mfa\in\calA$ should be understood as \emph{anyon charge} labels, i.e.\ labels for the superselection sectors of bulk anyons. Quantum dimensions of anyons $\{d_\mfa\}$ is a set of positive numbers ($d_\mfa\ge 1$) satisfying
\begin{equation}
    d_\mfa \times d_\mfb = \sum_\mfc N_{\mfa\mfb}^\mfc d_\mfc,
\end{equation}
and are determined uniquely by this formula. An anyon is Abelian if $d_\mfa = 1$, and it is non-Abelian if $d_\mfa > 1$.

We also need an information-theoretic way to recognize these topological charges directly from local reduced density matrices based on entanglement bootstrap~\cite{Shi2019fusion}. 
Consider an annulus $X$ within the bulk of the topological order. Such an annulus is wider than the bulk correlation length and is away from the edge. Consider a convex set of states on the annulus locally indistinguishable from the ground state, known as the information convex set~\cite{Shi2019fusion,ShiLu-information-convex,Shi2018ICS}. Such a convex set must form a simplex, due to the area law constraints, according to the entanglement bootstrap~\cite{Shi2019fusion}. The extreme points are in one-to-one correspondence with the anyon types. Explicitly, 
\begin{equation}
    \Sigma(X) = \{ \rho_X = \sum_\mfa  p_\mfa\rho^\mfa_X| \rho^\mfa_X \text{ are extreme points} \}
\end{equation}
with a set of mutually orthogonal extreme states $\{\rho^\mfa\}_{\mfa \in \calA}$. By the orthogonality, we mean the vanishing of fidelity ($F(\rho,\sigma):= \|\sqrt{\rho}\sqrt{\sigma}\|_1^2$) 
\begin{equation}\label{eq:sector_orthogonality_annulus}  F(\rho^\mfa_X,\rho^\mfb_X)=\delta_{\mfa,\mfb}.
\end{equation}
Importantly, any state on the cylinder with no bulk excitations reduces to the bulk annulus $X$, as Fig.~\ref{fig:ICS-annulus}(b) would prepare a state in $\Sigma(X)$. This establishes a way to detect the anyon charge purely
from the bulk, without referring to the edge. For lattice systems realizing chiral topological orders, we expect the orthogonality is approximate $F(\rho^\mfa_X,\rho^\mfb_X) \approx \delta_{\mfa,\mfb}$, where the error is small when the thickness is larger compared with the bulk correlation length. We shall neglect such errors for our applications.


The chiral central charge $c_-$ is the coefficient that determines the thermal Hall conductance of a chiral gapped system~\cite{Kane1997,Kitaev2005}.
The value of $c_-$ is determined by the bulk topological order. When $c_-$ is nonzero, there is an unequal number of left movers and right movers on an edge. In particular, when $c_- \neq 0$, the edge must be gapless.

\emph{The chiral edge:} Consider a chiral topological order on a cylinder, as in Fig.~\ref{fig:psi-a}. Low-energy states without bulk excitations are described by the gapless modes propagating along the two physical edges~\cite{Wen1995-review}. The bulk topological sector is labeled by an anyon type $\mfa$, while each edge is characterized by a primary state of the corresponding edge conformal field theory. Suppose that the lower edge only has right movers, and the upper edge only has left movers; this necessarily means that the chiral central charge $c_-$ is nonzero~\cite{Wen1995-review,Kitaev2005}. We denote such a cylinder state by
\begin{equation}\label{eq:primaryLR}
    |\mfa, \phi^I_\mfa(L), \phi^J_{\bar{\mfa}}(R)\rangle, \quad \mfa \in \calA.
\end{equation}
Here $L$ and $R$ label the two edges of the cylinder, $\mfa$ specifies the bulk anyon sector, and $I,J$ distinguish different primary states associated with that sector. The symbols $\{\phi^I_\mfa\}$ should be explained as primary fields of CFT from a dimensional reduction point of view. Importantly, these state vectors are mutually orthonormal
\begin{equation}
   \langle \mfa, \phi^I_\mfa(L), \phi^J_{\bar{\mfa}}(R) |\mfa', \phi^{I'}_{\mfa'}(L), \phi^{J'}_{\bar{\mfa}'}(R)\rangle = \delta_{\mfa,\mfa'}  \delta_{I,I'} \delta_{J,J'}. \nonumber
\end{equation} 
Throughout this work, these cylinder primary states will serve as the building blocks of our code subspaces.

\begin{figure}[H]
    \centering
    \includegraphics[width=0.85\linewidth]{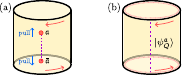}
    \caption{Prepare a state in the code subspace of the chiral edge code.
    $|\psi^\mfa_Q\rangle$ is the short-hand notation for $|\mfa,\phi^I_\mfa(L),\phi^I_{\bar{\mfa}}(R)\rangle$ of Eq.~\eqref{eq:primaryLR}. The cylinder should be sufficiently wide that its height is much greater than the bulk correlation length. (a) Create an anyon pair $(\mfa, \bar{\mfa})$ in the bulk and pull the anyons towards the opposite edges. (b) Spread the edge excitation uniformly on the edge and extremize the energy to obtain a certain primary state.  }
    \label{fig:psi-a}
\end{figure}

Here are a few more words on the edge states and the nature of chiral topological order. Throughout the manuscript, we assume that the properties of the edge states are matched to certain 1D CFTs through a dimensional reduction, detailed in Sec.~\ref{subsec:dimension-reduction} below. In particular, the assumption that we have a unique ground state excludes the more subtle types of edges obtained by stacking a topological order with a gapped edge on top of a chiral topological order. We further assume that the upper and lower edges are of the same type for simplicity.

While we do not need more details of the edge theory, we summarize a few recent developments for interested readers. The edges have only right movers but no left movers; such edges are \emph{completely chiral} \cite{Kong2017edge,Kong2019:part1,Kong2019:part2}; see also \cite{Chiral-vira2024,cross-ratio2024}. They are among the simplest edges of topological orders. Local perturbation should not gap out any edge mode. Topologically ordered systems can have ungappable edges even when $c_-=0$; this can happen due to some extra obstruction related to higher central charges~\cite{PhysRevX.3.021009,Ng2020-higher,Kaidi2021}. While 1D CFT can be either rational or irrational, the 1D CFT obtained by the dimensional reduction of a purely chiral edge is expected to be a rational CFT (RCFT). A single edge is closely related to the holomorphic (or anti-holomorphic) part of CFT, whose mathematical backbone is the concept of vertex operator algebra~\cite{1993hep.th....1009H,2005PNAS..102.5352H,Kong2017edge,Sopenko2023}.

\subsection{Quantum memory with the chiral edges}\label{sec:chiral-edge-code}

We utilize the states $|\mfa, \phi^I_\mfa(L), \phi^J_{\bar{\mfa}}(R)\rangle$ on a cylinder geometry of a topological order with chiral gapless edges to encode information. We formally introduce the chiral edge code. 

\begin{definition}[Chiral edge code]\label{def:chiral-edge-code}
For a chiral topological order on a cylinder with a pair of chiral edges, we choose a subset of anyon types \(\mathcal A_{\rm code}\subseteq \mathcal A\). For each anyon type, \(\mfa\in\mathcal A_{\rm code}\) we fix one cylinder primary state $\ket{\psi_Q^\mfa}
  :=|\mfa, \phi^I_\mfa(L), \phi^J_{\bar{\mfa}}(R)\rangle$.
Namely, we take one primary state for each anyon sector, and let 
\begin{equation}
\mathbb{V}^{\chi} \bigl(\cylindericon[1.2]\bigr)
:=
  \mathrm{span} \left\{
  \ket{\psi_Q^\mfa}:\mfa\in\mathcal A_{\rm code}
  \right\},
  \label{eq:chiCFT_codespace_primary_relent}
\end{equation}
with $D:=\dim  \mathbb{V}^{\chi} \bigl(\cylindericon[1.2]\bigr)=|\mathcal A_{\rm code}|.$ 
 We call this choice of code subspace a chiral edge code.
\end{definition}

We shall analyze the robust nature of such a quantum memory in later sections. As we shall see, the requirement that we have at most one state for each anyon sector is crucial to some of our analysis. The idea easily generalizes to broader setups, e.g. by replacing some of the primaries with descendants.

\subsection{Examples of chiral edge codes}\label{sec:examples-chiral-states}

We provide two examples of chiral edge codes. They make use of two different chiral topological orders.

\begin{exmp}[Chiral Ising code]
\label{exmp:chiral_Ising}

We use the Ising anyon topological order~\cite{Kitaev2005}. There are three anyon types $\calA_{\rm Ising} =\{ 1, \sigma, \epsilon \}$, with quantum dimensions $d_1=1, d_\sigma = \sqrt{2}$, $d_\epsilon=1$, and fusion rule $\sigma \times \sigma = 1 + \epsilon$, $\epsilon \times \epsilon =1$ and $\epsilon \times \sigma = \sigma$. The chiral central charge is $c=1/2$. The chiral edge of it should be described by the holomorphic (anti-holomorphic) part of an Ising CFT. On a cylinder, the entire set of primary states is
\begin{equation}
    |\mfa, \phi_\mfa(L),\phi_{\bar{\mfa}}(R)\rangle, \quad
    \mfa \in \calA_{\rm Ising}.
\end{equation}
In total, there are three of them.
Because there is a unique $\phi_\mfa$ for any $\mfa$ for this example, we shall simplify the notation 
as
\begin{equation}
    |\psi^{\mfa}\rangle := |\mfa, \phi_\mfa(L),\phi_{\bar{\mfa}}(R)\rangle,  \quad
    \mfa \in \calA_{\rm Ising}.
\end{equation}
The chiral Ising code has the following choices of code subspaces
\begin{equation}
    \begin{aligned}
        &\text{span}\{ |\psi^1\rangle,|\psi^\sigma\rangle  \}, \\
         &\text{span}\{ |\psi^1\rangle,|\psi^\epsilon\rangle  \}, \\
          &\text{span}\{ |\psi^\sigma\rangle,|\psi^\epsilon\rangle  \}, \\
          &\text{span}\{ |\psi^1\rangle,|\psi^\sigma\rangle,|\psi^\epsilon\rangle  \}. 
    \end{aligned}
\end{equation}
Three of the choices are 2-dimensional, and one is 3-dimensional.
\end{exmp}

\begin{exmp}[Chiral semion code]
\label{exmp:chiral_Semion}

The chiral semion topological order has two anyon types $\calA_{\rm semion} =\{1, s\}$. It is associated with the smallest nontrivial modular tensor category~\cite{Rowell2007}.
Here, $s$ is the semion, with Abelian fusion rule $s \times s =1$, $d_s=1$ and its topological spin $\theta_s = i$. The chiral central charge is $c_-=1$. The set of edge primary states is
\begin{itemize}[leftmargin=15pt]
    \item for the sector of trivial anyon $1\in \calA_{\rm semion}$,
    \begin{equation}
    |\psi_1^{I,J}\rangle:=|1, \phi^I_1(L), \phi_1^J(R)\rangle,
\end{equation}
where $I, J$ take values on integer lattices $\mathbb{Z}$. 
    \item for the sector of semion $s\in \calA_{\rm semion}$,
    \begin{equation}
     |\psi_s^{I,J}\rangle:= |s, \phi^I_s(L), \phi_{{s}}^J(R)\rangle,
\end{equation}
where $I, J$ take values on a shifted integer lattices $\mathbb{Z}+ \frac{1}{2}$.
\end{itemize}
Note that there are multiple (in fact, an infinite number) of cylinder primary states that correspond to any chosen anyon type. When there is $U(1)$ charge conservation symmetry, it can be realized by the $\nu = 1/2$ bosonic Laughlin state. In that context, the values $I$ and $J$ represent the net $U(1)$ charge of the two edges $(q_L,q_R) = (I,J)$.

To obtain a code space, it is necessary to have a state in the vacuum sector ($\mfa=1$) and a state in the semion sector ($\mfa=s$): 
\begin{equation}
    \mathbb{V}^{\chi} \bigl(\cylindericon[1.2]\bigr)
:= \text{span}\{   |\psi_1^{I,J}\rangle,   |\psi_s^{I',J'}\rangle \}.
\end{equation}
Note that we still have the freedom to choose $I,J \in \mathbb{Z}$ and $I',J' \in \mathbb{Z}+1/2$. For each such choice, we obtain a code subspace.

For later convenience, we introduce another shorthand notation for some of the frequently used code states as
\begin{equation}
\begin{aligned}
        |0, 0\rangle &= |\psi_{1}^{0,0}\rangle \\
        |s, s\rangle &= |\psi_{s}^{\frac{1}{2},\frac{1}{2}}\rangle\\
        |3s,s\rangle &= |\psi_{s}^{\frac{3}{2},\frac{1}{2}}\rangle\\
        |s, t\rangle &= |\psi_{s}^{\frac{1}{2},-\frac{1}{2}}\rangle.
\end{aligned}
\end{equation}
 \end{exmp}

\subsection{Dimensional reduction}\label{subsec:dimension-reduction} 

We now describe the dimensional-reduction correspondence between chiral topological order on the cylinder and one-dimensional rational conformal field theory (RCFT). It provides a way to compute certain cutoff-independent entanglement quantities.

The correspondence has two ingredients: First, there is a geometric projection
from a 2D cylinder to a 1D circle as illustrated in Fig.~\ref{fig:Dimension-Reduction}. Second, cylinder primary states are matched with primary states of the associated RCFT. We only assume that this correspondence preserves UV-finite entanglement quantities, such as the relative entropy and the mutual information between separated regions. We do not assume equality of microscopic reduced density matrices, entanglement spectra, or cutoff-dependent entropies. The assumption is used only for relative entropies and mutual information between separated full-column regions.

\begin{figure}[H]
\centering
\includegraphics[width=0.85\linewidth]{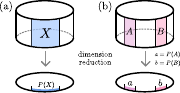}
\caption{{\bf Dimensional reduction.} It relates the chiral topological order on a cylinder and a certain RCFT on a circle. (a) A full column $X$ is projected to an interval $P(X)$. (b) Two full columns, $A$ and $B$, are projected onto disjoint intervals $a$ and $b$. We shall encounter regions of the cylinder that cannot be associated with the dimensional reduction; see Fig.~\ref{fig:Srel}.}
\label{fig:Dimension-Reduction}
\end{figure}

\emph{(i) Geometric projection.} Let the cylinder be $
M = S^1\times [0,L_y]$, with coordinates \((u,y)\), where \(u\in[0,2\pi)\) is the angular coordinate around the cylinder and \(y\in[0,L_y]\) is the coordinate across the cylinder. For an arc \(\mathcal W_{j}\in (u_j-\frac{\theta_j}{2},u_j+\frac{\theta_j}{2}]\subset S^1\) with angle width $\theta_j$, we define the corresponding \emph{full column}
in the cylinder by
\begin{equation}
  X_{j}
  :=
  \{(u,y)\in M:\ u\in\mathcal W_{j},\ 0\le y\le L_y\}.
  \label{eq:Qi_full_column}
\end{equation}
The dimensional-reduction map projects this full column to the interval with the same angular support in the 1D circle:
\begin{equation}
    P\left(X_{j}\right):= \mathcal{W}_{j}\subset S^1.
\end{equation}
as shown in Fig.~\ref{fig:Dimension-Reduction} (a). For a union of several disjoint arcs, this map is still applicable (see Fig.~\ref{fig:Dimension-Reduction} (b)).

\emph{(ii) Primary-state and UV-finite entanglement quantities correspondence.} 
We shall be interested in cylinder primary states
\begin{equation}
    \ket{\psi^{\mfa,I,J}}
    :=
    \ket{\mfa,\phi_\mfa^{I}(L),\phi_{\bar \mfa}^{J}(R)},
    \quad \mfa\in\mathcal A.
    \label{eq:DR_cylinder_primary}
\end{equation}
for the purpose of studying chiral edge code as shown in~\eqref{eq:chiCFT_codespace_primary_relent}. Here \(\mfa\) labels the bulk anyon sector, while \(I,J\) specify the chosen cylinder primary states on the two physical edges. 
By the dimensional reduction, we mean a linear map from cylinder primary state to a primary state of the 1D RCFT:
\begin{equation}
   \Gamma:\quad  \ket{\psi^{\mfa,I,J}}
    \quad\to\quad
    \ket{\varphi^{\mfa,I,J}},
    \label{eq:DR_primary_correspondence}
\end{equation}
which preserves the inner product and the following matching assumption (Assumption~\ref{asmp:UVfinite_entanglement_map_columns}).
Here, the labels $\mfa,I,J$ collectively determine a Virasoro primary $\ket{\varphi^{\mfa,I,J}}$ of the 1D RCFT state inherited from the cylinder primary state.   


We note that our notion of dimensional reduction is closely related to ideas in existing literature about the correspondence between 1+1D CFT and chiral edges, such as~\cite{Qi2012,Chiral-vira2024,cross-ratio2024}. For the purpose of studying entanglement quantities, we shall need the following:


\begin{assumption}[UV-finite entanglement quantities matching]
\label{asmp:UVfinite_entanglement_map_columns}
Consider states $\{ |\Psi_k\rangle \}_{k=1}^m$ on the cylinder of the form $\sum_{\mfa,I,J} c_{\mfa,I,J}|\psi^{\mfa,I,J}\rangle$. Any UV-finite entanglement quantity computed with this set of states and a set of columns $\{X_j\}_{j=1}^{n}$ must be identical to the same quantity computed for states $\{ \Gamma|\Psi_k\rangle \}_{k=1}^m$ on intervals $\{P(X_j)\}_{j=1}^{n}$.
\end{assumption}
 
We emphasize that we do not assume the matching of cutoff-dependent quantities, such as the von Neumann entropy of a single region or the mutual information of adjacent intervals. The following are examples of the identification of UV-finite entanglement quantities suggested by Assumption~\ref{asmp:UVfinite_entanglement_map_columns}. We shall use them explicitly in the study of chiral edge code.

\emph{Relative entropy.} The relative entropy is defined given two states $\{ |\Phi_1\rangle, |\Phi_2\rangle\}$ and a region $Z =X_1 \cup X_2  \cup \cdots$ formed by a union of full columns. It is UV finite, and according to Assumption~\ref{asmp:UVfinite_entanglement_map_columns}, 
\begin{equation}
S\left(\rho_Z^{\ket{\Psi_1}}\middle\|\rho_Z^{\ket{\Psi_2}}\right)=S\left(\rho_{P(Z)}^{\Gamma\ket{\Psi_1}}\middle\|\rho_{P(Z)}^{\Gamma\ket{\Psi_2}}\right),
\end{equation}

\emph{Mutual information of disjoint full columns.}
For two disjoint full columns $A,B$ (as in Fig.~\ref{fig:Dimension-Reduction}(b)), the mutual information is defined given a single state $|\Psi\rangle$, and is UV finite. According to Assumption~\ref{asmp:UVfinite_entanglement_map_columns}, we have
\begin{equation}
  I\left(A:B\right)_{\ket{\Psi}}
  \;=\;
  I\left(P(A):P(B)\right)_{\Gamma\ket{\Psi}}.
  \label{eq:MI_column_match}
\end{equation}  
No such identification is assumed when $A$ and $B$ are adjacent, where mutual information is UV-sensitive.

\subsubsection{Examples of dimensional reduction}

Applying the idea of dimensional reduction to code subspaces of a chiral edge code, we obtain corresponding code subspaces of a 1D CFT code.
\begin{equation}
\mathbb V^{\rm CFT} = \Gamma(\mathbb{V}^\chi \bigl(\cylindericon[1.2]\bigr)).
\end{equation}
We give two concrete examples of dimensional reduction.
Later, we shall compare the information robustness of the chiral edge code with that of its 1D counterpart.

\begin{exmp}[Ising anyon versus Ising chain]\label{exmp:1D-Ising}
The dimensional reduction of the Ising anyon topological order (Example~\ref{exmp:chiral_Ising}) corresponds to a 1D RCFT with central charge $c=1/2$. Such a CFT is necessarily the Ising minimal model. On the lattice, such a CFT can be realized by a transverse-field Ising chain at its critical point. The Ising CFT has three primary states $|\varphi^1\rangle$, $|\varphi^\sigma\rangle$ and $|\varphi^\epsilon\rangle$. They correspond to the lowest-energy states of the cylinder and correspond to the three anyon sectors: vacuum, $\sigma$ and $\epsilon$ of the Ising topological order.
\end{exmp}

\begin{exmp}[Chiral semion versus compactified free boson RCFT]\label{exmp:1D-semion}
The dimensional reduction of chiral semion topological order (Example~\ref{exmp:chiral_Semion}) is the compactified free boson CFT with $c=1$ at compactification radius $R=\sqrt{2}$. The lowest primary state of this CFT has $(h,\bar{h})$ being $(0,0)$ and $ (1/4,1/4)$, where $h$ and $\bar{h}$ are the holomorphic and anti-holomorphic conformal weights.
These CFT primary states are the dimensional reduction of cylinder primary states of the trivial anyon $|\varphi_{1}^{0,0}\rangle$ and 4 states with the semion superselection sectors
\begin{equation}
    |\varphi_{s}^{\frac{1}{2},\frac{1}{2}}\rangle, \quad |\varphi_{s}^{\frac{1}{2},-\frac{1}{2}}\rangle, \quad |\varphi_{s}^{-\frac{1}{2},\frac{1}{2}}\rangle, \quad |\varphi_{s}^{-\frac{1}{2},-\frac{1}{2}}\rangle,
\end{equation} 
for which $(h,\bar{h}) = (1/4, 1/4)$, 
and with changes on each edge being $\pm 1/2$.

\end{exmp}

\section{Power laws and the robustness of chiral edge code} 
\label{sec:coherent_loss_vs_relent}

We have introduced the chiral edge code in Def.~\ref{def:chiral-edge-code}. In this section, we present the theoretical analysis of its coherent-information loss due to the erasure of a local region. We contrast this behavior of the 2D chiral edge code with the 1D CFT code obtained by dimensional reduction. The crucial finding can be stated in terms of a set of power-law exponents $\alpha$, $\beta$, and $\gamma$ that we define in terms of the relative entropy between the 1D CFT and the 2D chiral edge states.  

This section is organized as follows. First, we introduce three power-law exponents of relative entropy in Sec.~\ref{sec:S-rel-power}. In Sec.~\ref{subsec:chi_CFT_code_relent}, we explain an exact relation between relative entropy and the coherent information loss. In Sec.~\ref{subsec:CFT-code_relent}, we explain how these exponents are related to the coherent information loss in 2D chiral edge code and its dimensional reduction. In Sec.~\ref{subsec:theoretical-summary}, we provide a physical summary of why such powers capture the robustness of memory.

\begin{figure*}[t]
      \centering
    \includegraphics[width=0.99\linewidth]{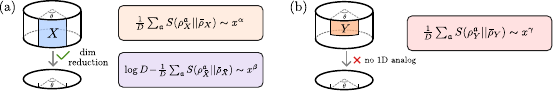}
    \caption{Three power law exponents defined by relative entropy behavior at small $x$, where $x = \theta/2\pi$ vanishes at small region sizes. (a) The exponents $\alpha$ and $\beta$ have a 1D analog by dimensional reduction. Here $X$ is an interval times an arc that touches both edges and has angular size $x$. (b) The exponent $\gamma$ has no analog of dimensional reduction, where the region $Y$ touches only one edge of the cylinder. }
    \label{fig:Srel}
\end{figure*} 

\subsection{Relative entropy and power law}\label{sec:S-rel-power}

The relative entropy is defined given two density matrices \(\rho\) and \(\lambda\) as
\begin{equation}
    S(\rho\|\lambda):=\Tr(\rho\log\rho-\rho\log\lambda).
\end{equation}
This quantity is non-negative, vanishes if and only if \(\rho=\lambda\), and is UV-finite in quantum field theories. 
We consider three kinds of relative entropy behavior for a code subspace of chiral edge code. Two of them ($\alpha,\beta$) have a dimensional reduction picture of a 1D CFT code with a certain code subspace, and one of them ($\gamma$) has no 1D analog, as summarized in Fig.~\ref{fig:Srel}.

To state the precise definition of the powers,
we recall that, in a chiral edge code, each codeword in \(\mathbb{V}^{\chi} \bigl(\cylindericon[1.2]\bigr)\) is a cylinder primary state
\begin{equation}
    \ket{\psi_Q^\mfa}:=
    \ket{\mfa,\phi_\mfa^{I_\mfa}(L),\phi_{\bar \mfa}^{J_\mfa}(R)},
\end{equation}
and corresponds to a distinct anyon type \(\mfa\in\mathcal A_{\rm code}\). For any region \(Y\subseteq Q\), let $\rho_Y^\mfa:=\Tr_{Q\setminus Y}\ket{\psi_Q^\mfa} \bra{\psi_Q^\mfa}$ and $\bar\rho_Q:=\frac1D\sum_{\mfa\in\mathcal A_{\rm code}}
    \ket{\psi_Q^\mfa} \bra{\psi_Q^\mfa}$, with $D := |\calA_{\rm code}|$.  
The defining properties of the exponents are
\begin{eqnarray}
   x^\alpha &\sim& \frac{1}{D}\sum_{\mfa \in \calA_{\rm code}} S(\rho^\mfa_X|| \bar{\rho}_X), \label{eq:alpha-def}\\
   x^\beta &\sim& \log D - \frac{1}{D}\sum_{\mfa \in \calA_{\rm code}} S(\rho^\mfa_{\bar{X}}|| \bar{\rho}_{\bar{X}}), \label{eq:beta-def}\\
   x^\gamma &\sim&  \frac{1}{D}\sum_{\mfa \in \calA_{\rm code}} S(\rho^\mfa_{{Y}}|| \bar{\rho}_{Y}), \label{eq:gamma-def}
\end{eqnarray}
where the regions $X$ and $Y$ are shown in Fig.~\ref{fig:Srel}, ``$\sim$'' refers to the leading power law behavior at small $x$, and the \emph{angular size} $x:=\theta/2\pi$ where $\theta \in [0,2\pi]$ is angle associated with the subsystem. 
We shall relate them to the coherent information loss.

\begin{remark} We provide a few remarks about the exponents:
\begin{enumerate}[leftmargin=13pt]
    \item Quantities in \eqref{eq:alpha-def} and \eqref{eq:beta-def} have dimensional reduction interpretations, and they correspond to UV-finite quantities in 1D CFT, as suggested in Fig.~\ref{fig:Srel}. The reason is that relative entropies are UV finite, and thus the dimensional reduction idea explained in Sec.~\ref{subsec:dimension-reduction} applies. 
    
    \item The exponent $\alpha$ can be computed using 1D CFT techniques.  One method is to do a direct computation using the replica trick~\cite{Lashkari2014,Lashkari2015,Sarosi2016,Sarosi2017,Chowdhury_2022,Lashkari2026} for small intervals. Another more naive way is to derive a certain joint convexity upper bound on $\alpha$. As we show, the value $\alpha$ in the replica computation generally saturates the upper bound. See Appendix~\ref{app:A-alpha}.
 
    \item The justification of the power law related to $\beta$ is suggested by our numerical results (Sec.~\ref{sec:power_law_numerics}), and it is also backed up by the most recent progress on replica computation in conformal field theory, Ref.~\cite{Lashkari2026}, with the relevant finding we summarize in Appendix~\ref{app:A-beta}.  
    \item All powers satisfy $\alpha,\beta,\gamma >0$. This is due to the monotonicity of relative entropy. There are examples $\alpha > \beta$ and other examples with $\alpha < \beta$; see Sec.~\ref{sec:power_law_numerics}. We do not know a theoretical way to compute $\gamma$ exactly, but in Prop.~\ref{prop:powers}, we show $\gamma \ge \alpha$.

    \item We will define closely related quantities $\gamma_{\mfa \mfb}$ ($\alpha_{\mfa \mfb}$) for the purpose of understanding power-law-range recovery map in Sec.~\ref{sec:quasi-local-R} and Appendix~\ref{app:gamma-computation}. 
\end{enumerate}
    
\end{remark}

\subsection{Coherent-information loss versus relative entropy}
\label{subsec:chi_CFT_code_relent}

We now establish a concrete relation between the coherent information loss and the relative entropy in the context of chiral edge code. We first present a lemma.

\begin{lemma}[Coherent-information loss as relative entropy]
\label{lemma:Delta_Srel-general}
Consider a quantum code with code subspace \(\mathbb{V}=\mathrm{span}\{\ket{\psi_Q^i}\}_{i=1}^D\). Let \(A\subset Q\) be a physical region and \(B=Q\setminus A\). Then
\begin{equation}
\Delta(A;\mathbb{V})
=
\log D+
\frac1D\sum_{i=1}^D
\left[
S \left(\rho_A^i\middle\|\bar\rho_A\right)
-
S \left(\rho_B^i\middle\|\bar\rho_B\right)
\right]
\label{eq:loss_on_chain_relent}
\end{equation}
where \(\bar\rho_Q=\frac1D\sum_i\ket{\psi_Q^i} \bra{\psi_Q^i}\).
\end{lemma} 

\begin{proof}
Using Lemma~\ref{lemma:delta},
\begin{equation}
\begin{aligned}
\Delta(A;\mathbb{V})
&=(S_A+S_Q-S_B)_{\bar\rho} \\
&=S(\bar\rho_A)+\log D-S(\bar\rho_B).
\end{aligned}
\end{equation}
For any ensemble \(\{p_j,\lambda^j\}\),
\begin{equation}
    \sum_j p_j S \left(\lambda^j\middle\|\sum_k p_k\lambda^k\right)
    =
    S \left(\sum_jp_j\lambda^j\right)-\sum_jp_jS(\lambda^j).
\end{equation}
Applying this identity to the two reduced ensembles on \(A\) and \(B\), and using \(S(\rho_A^i)=S(\rho_B^i)\) for each pure codeword \(\ket{\psi_Q^i}\), gives Eq.~\eqref{eq:loss_on_chain_relent}.
\end{proof}

\begin{tcolorbox}[breakable, enhanced, colback=yellow!5!white,colframe=orange!35!white]
\begin{theorem}
\label{thm:Delta_as_avg_relent_relent}
Consider a chiral edge code on a cylinder \(Q\). Let \(A\) be a local region such that \(B:=Q\setminus A\) contains a noncontractible bulk annulus of the cylinder; e.g., those in Fig.~\ref{fig:cylinder-Q1-relent}. Then
\begin{equation}
\Delta(A;\mathbb{V}^{\chi} \bigl(\cylindericon[1.2]\bigr))
  =
  \frac1D\sum_{\mfa\in\mathcal A_{\rm code}}
  S \left(\rho_A^\mfa\middle\|\bar\rho_A\right).
  \label{eq:Delta_as_avg_relent_relent}
\end{equation}
\end{theorem}
\end{tcolorbox}
Note that, the choice of region $A$ in the theorem is quite general and is not limited to a single disk. Some allowed choices are illustrated in Fig.~\ref{fig:cylinder-Q1-relent}.

\begin{figure}[h]
    \centering
    \includegraphics[width=0.99\linewidth]{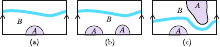}
    \caption{A cylinder $Q$ partition into an erased region $A$ and its complement $B= Q \setminus A$. The complement $B$, in any of the three cases, contains a noncontractible bulk annulus, illustrated as blue strips. The anyon charge of the code words can be measured on $B$.} 
    \label{fig:cylinder-Q1-relent}
\end{figure}

\begin{proof}
By Lemma~\ref{lemma:Delta_Srel-general},
\begin{equation}\label{eq:copy}
\begin{aligned}
&\Delta(A;\mathbb{V}^{\chi} \bigl(\cylindericon[1.2]\bigr))
\\&=
\log D+
\frac1D\sum_{\mfa\in\mathcal A_{\rm code}}
\left[
S \left(\rho_A^\mfa\middle\|\bar\rho_A\right)
-
S \left(\rho_B^\mfa\middle\|\bar\rho_B\right)
\right].
\end{aligned}
\end{equation}
It is known that $B$ contains a noncontractible bulk annulus $X$ shown in Fig.~\ref{fig:ICS-annulus}. By Eq.~\eqref{eq:sector_orthogonality_annulus}, the bulk annulus can detect the anyon sectors of each code word perfectly. Thus, by the monotonicity of fidelity
\begin{equation}
 F(\rho_B^\mfa,\rho_B^\mfb)  \le  F(\rho_X^\mfa,\rho_X^\mfb)=0,
    \quad \mfa \ne \mfb  .
    \label{eq:Fid-orthogonal}
\end{equation}
Since states \(\{\rho_B^\mfa \} \) have mutually orthogonal support, \(\bar\rho_B\) is block diagonal with equal weights. Therefore
\begin{equation}
S \left(\rho_B^\mfa\middle\|\bar\rho_B\right)=\log D, \quad \forall \mfa\in\mathcal A_{\rm code}.
\end{equation}
Plugging this into Eq.~\eqref{eq:copy}, and Eq.~\eqref{eq:Delta_as_avg_relent_relent} follows.
\end{proof}

\subsection{Comparing memory robustness of 1D and 2D with power law exponents}
\label{subsec:CFT-code_relent}

\begin{tcolorbox}[breakable, enhanced, colback=yellow!5!white,colframe=orange!35!white]
\begin{Proposition}\label{prop:powers}
    Consider a 2D chiral edge code with code space $\mathbb{V}^{\chi} \bigl(\cylindericon[1.2]\bigr))$. Let $Y$ be a connected region near an edge as in Fig.~\ref{fig:power-law-region}, for which the angular size $x$ is small. Then,  
    \begin{equation}
        \Delta(Y;\mathbb{V}^{\chi}) \sim  x^\gamma.  
    \end{equation}
    Furthermore, the 1D CFT code subspace $\mathbb{V}^{\rm CFT}$ obtained from the dimension reduction of the chiral edge code satisfies
    \begin{equation}
        \Delta(P(X);\mathbb{V}^{\rm CFT}) \sim  x^{\min\{\alpha,\beta\}},
    \end{equation}
    where $P(X)$ is an interval of the circle that has the same angular size as $Y$ as in Fig.~\ref{fig:power-law-region}. Moreover, $\gamma \ge \alpha$. Consequently,
    \begin{eqnarray}
        \gamma \ge \alpha\geq \min\{\alpha,\beta\}.
    \end{eqnarray}
\end{Proposition}
\end{tcolorbox}

\begin{figure}[h]
    \centering
    \includegraphics[width=0.82\linewidth]{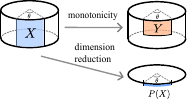} 
    \caption{Regions on the cylinder used in Prop.~\ref{prop:powers} and its proof. \(Y\) is a connected region near one physical edge whose width is $x:= \theta/2\pi$. \(X\) is a full-column strip that has the same width as $Y$, and the dimensional reduction of the region is an interval $P(X)$ on a circle.}
    \label{fig:power-law-region}
\end{figure}

\begin{proof}
The first claim, $\Delta(Y;\mathbb{V}^{\chi}) \sim  x^\gamma$, follows immediately from Thm.~\ref{thm:Delta_as_avg_relent_relent} and the definition of $\gamma$ in Eq.~\eqref{eq:gamma-def}. The second statement follows simply by using Lemma~\ref{lemma:Delta_Srel-general} to rewrite the coherent information loss into three terms, where the two regions ($A$ and $B$ in the lemma) are the interval $P(X)$ and its complement on the circle. 
We apply the dimensional reduction correspondence discussed in Sec.~\ref{subsec:dimension-reduction} to turn the intervals into $X$ and its complement as in Fig.~\ref{fig:power-law-region}. Such a correspondence does not change the relative entropy (Assumption~\ref{asmp:UVfinite_entanglement_map_columns}). Thus, we have
\begin{equation}
    \begin{aligned}
   \Delta(P(X);\mathbb{V}^{\rm CFT}) =& \frac1D\sum_{\mfa\in \calA_{\rm code}}
S \left(\rho_X^\mfa\middle\|\bar\rho_X\right)
\\
&+ \log D 
-\frac1D\sum_{\mfa\in \calA_{\rm code}}
S \left(\rho_{\bar{X}}^\mfa\middle\|\bar\rho_{\bar{X}}\right). 
\end{aligned}
\end{equation}
Then by the definition of $\alpha$ and $\beta$, we have the RHS is $c_1 x^\alpha + c_2 x^\beta \sim x^{ \min\{\alpha,\beta\}}$ at small $x$. 
This completes the proof of the second claim.

The proof of the last claim ($\gamma \ge \alpha$) follows from the monotonicity of relative entropy. Identify $Y$ as a subset of $X$, namely $Y\subset X$. Therefore,
\begin{equation}
    \frac{1}{D}\sum_{\mfa \in \calA_{\rm code}} S(\rho^\mfa_X || \bar{\rho}_X)\ge \frac{1}{D}\sum_{\mfa \in \calA_{\rm code}} S(\rho^\mfa_Y || \bar{\rho}_Y),
\end{equation}
by the monotonicity of relative entropy, where the inequality holds for any $x$. Now, applying to the context of small $x$, we must have 
\begin{equation}
    c_1 x^\alpha \ge c_3 x^\gamma , \quad \forall x \ll 1,
\end{equation}
for positive constants $c_1$ and $c_3$. For this to be true, we must have $\gamma \ge \alpha$. This completes the proof. 
\end{proof}

\subsection{Effective code distance}\label{subsec:theoretical-summary}

Let us unpack the physical meaning of Prop.~\ref{prop:powers}, which tells us about the power law dependence of coherent information loss at small region size $x$.
\begin{equation}
        \Delta(Y;\mathbb{V}^{\chi}) \approx c_3 x^\gamma.  
\end{equation}
Here $\gamma>0$ and $c_3$ is a positive constant, and thus the coherent information loss increases as the region's angular width grows. This is an obvious remark.

More importantly, the larger the power $\gamma$, the more robust the memory. This can be argued by considering an \emph{effective code distance} $d^*$ such that, for erasure angular length $x \le d^*$, the coherent information loss is at most $\delta$, namely
\begin{equation}
    \Delta(Y;\mathbb{V}^{\chi}) \le \delta.
\end{equation}

In other words, $d^*$ as a function of $\delta$ behaves like
\begin{equation}
    d^*(\delta)  = \left(\frac{\delta}{c_3} \right)^{1/\gamma}.
\end{equation}
Suppose one is interested in small $\delta$, then the larger $\gamma$ is, the better the behavior of $d^*$ at small $\delta$. 

This idea of considering effective code distance is naturally motivated by a sequence of recent works on AQEC codes. In particular,~\cite{Yi_2024,2025arXiv251004453Y,Bentsen_2024} considered closely related notions of effective code distances, using a variety of related local quantities as the errors. In fact, the idea of local erasure noise and the destruction of memory dates back to early works on stabilizer codes~\cite{Grassl1997}.

From this it is also easy to understand why the dimensional reduction of the chiral edge code has a smaller power law exponent $\min\{\alpha,\beta\} \le \alpha \le \gamma$, according to
 \begin{equation}
        \Delta(P(X);\mathbb{V}^{\rm CFT}) \sim  x^{\min\{\alpha,\beta\}}.
 \end{equation}
 The intuitive reason is that if we squash the 2D cylinder into a 1D circle, more extended regions of the cylinder (such as $X$ in Fig.~\ref{fig:power-law-region}) will be treated as a local region. Decoherence on such regions can give larger coherent information loss. The effective code distance, for this 1D CFT code obtained by the dimensional reduction, is
 \begin{equation}
     d^*(\delta)  = \left(\frac{\delta}{c'} \right)^{1/\min\{\alpha,\beta\}},
 \end{equation}
 with a certain constant $c'$.
 Whenever $\min\{\alpha,\beta\}$ is strictly smaller than $\gamma$, the 2D chiral edge code has enhanced robustness over the 1D CFT code obtained by the dimensional reduction.  

These exponents will be computed explicitly in explicit models of chiral edge codes; see Sec.~\ref{sec:power_law_numerics}. In particular, in many (but not all) examples we compute, $\min\{\alpha,\beta\} = \beta$ and it is strictly smaller than $\gamma$. For such examples, there is a strictly better information robustness of the chiral edge code over its \(1\)D dimensional reduction.

\section{Power law exponents in examples}
\label{sec:power_law_numerics}

In this section, we give explicit examples of the important power law exponents $\alpha,\beta$ and $\gamma$ defined in the previous section (Sec.~\ref{sec:coherent_loss_vs_relent}).
For $\alpha$ and $\beta$ we provide a numerical computation and contrast it with a theoretical prediction based on Appendix~\ref{app:CFT_relative_entropy} and Ref.~\cite{Lashkari2026}. The numerical data are collected on critical spin chain lattice models, which are known to be dimensional reductions of 2D chiral topological orders. For the exponent $\gamma$, which is currently not available by analytical methods, we provide numerical data for the chiral semion example through a Bosonic Laughlin wavefunction with $\nu=\frac{1}{2}$. For the Ising example, we only claim the general lower bound $\gamma\geq \alpha$.

\begin{figure*}[t]
\centering

\begin{minipage}[c]{0.25\textwidth}
\centering
\resizebox{0.95\linewidth}{!}{%
\includegraphics[width=0.2\textwidth]{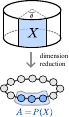}
}
\end{minipage}
\hfill
\begin{minipage}[c]{0.7\textwidth}
\centering

\begin{overpic}[width=0.9\linewidth]{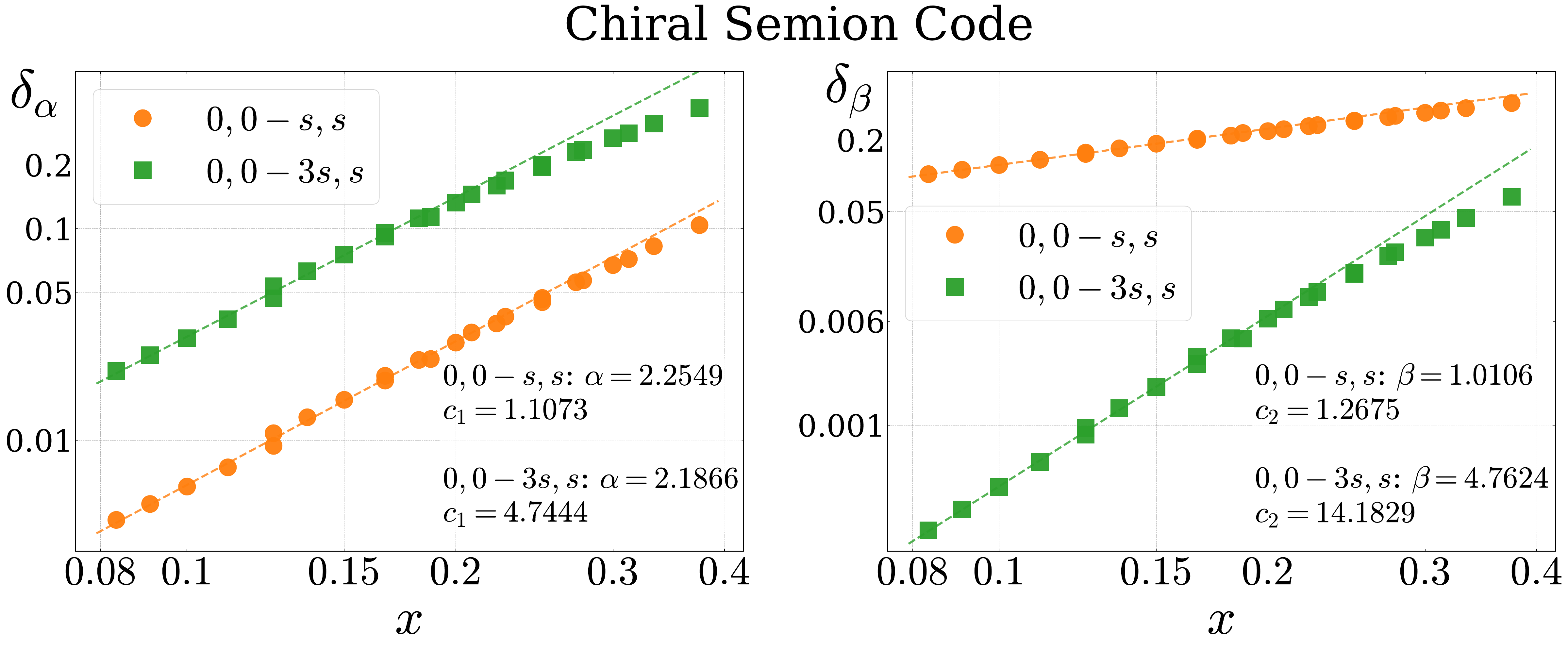}
  \put(2,40){\large\bfseries (a)}
\end{overpic}

\vspace{0.8em}

\begin{overpic}[width=0.9\linewidth]{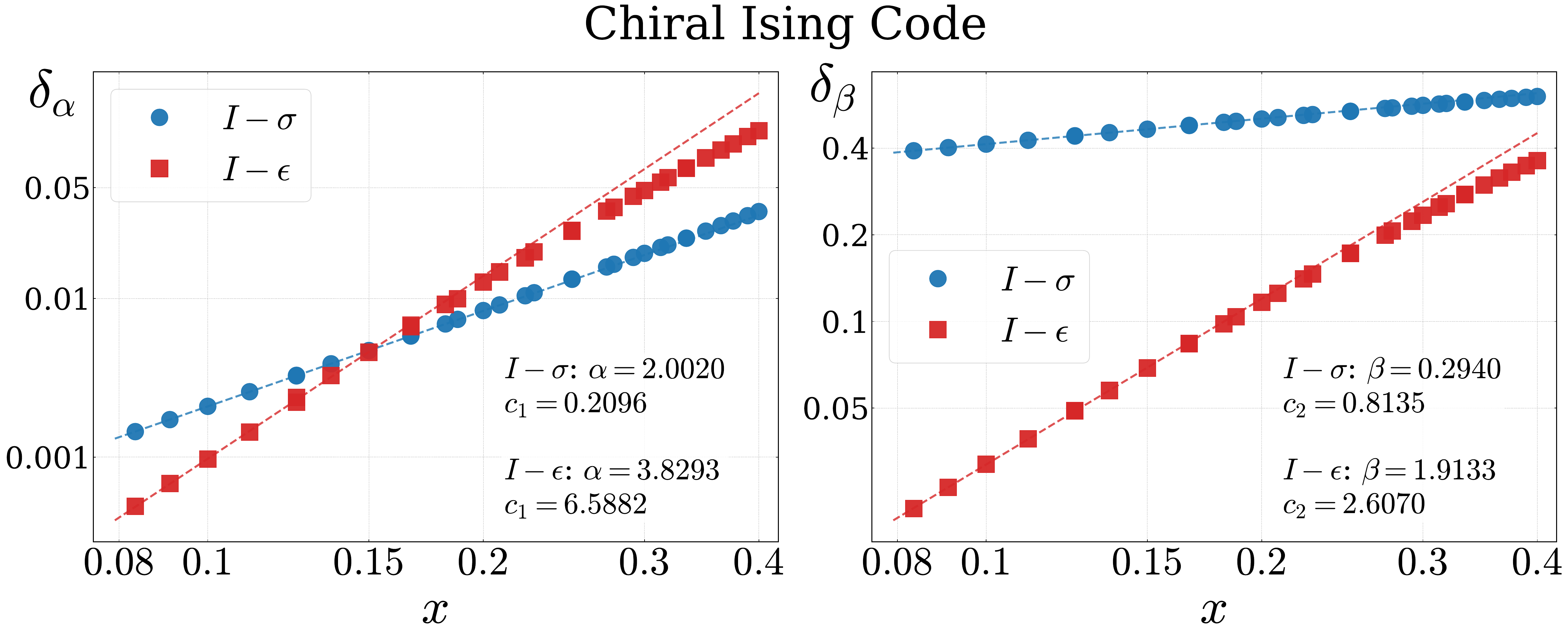}
  \put(2,40){\large\bfseries (b)}
\end{overpic}

\end{minipage}

\vspace{0.8em}

\noindent
\begin{minipage}[t]{0.495\linewidth}
  \vspace{0pt}
  \centering
  \begin{overpic}[width=\linewidth]{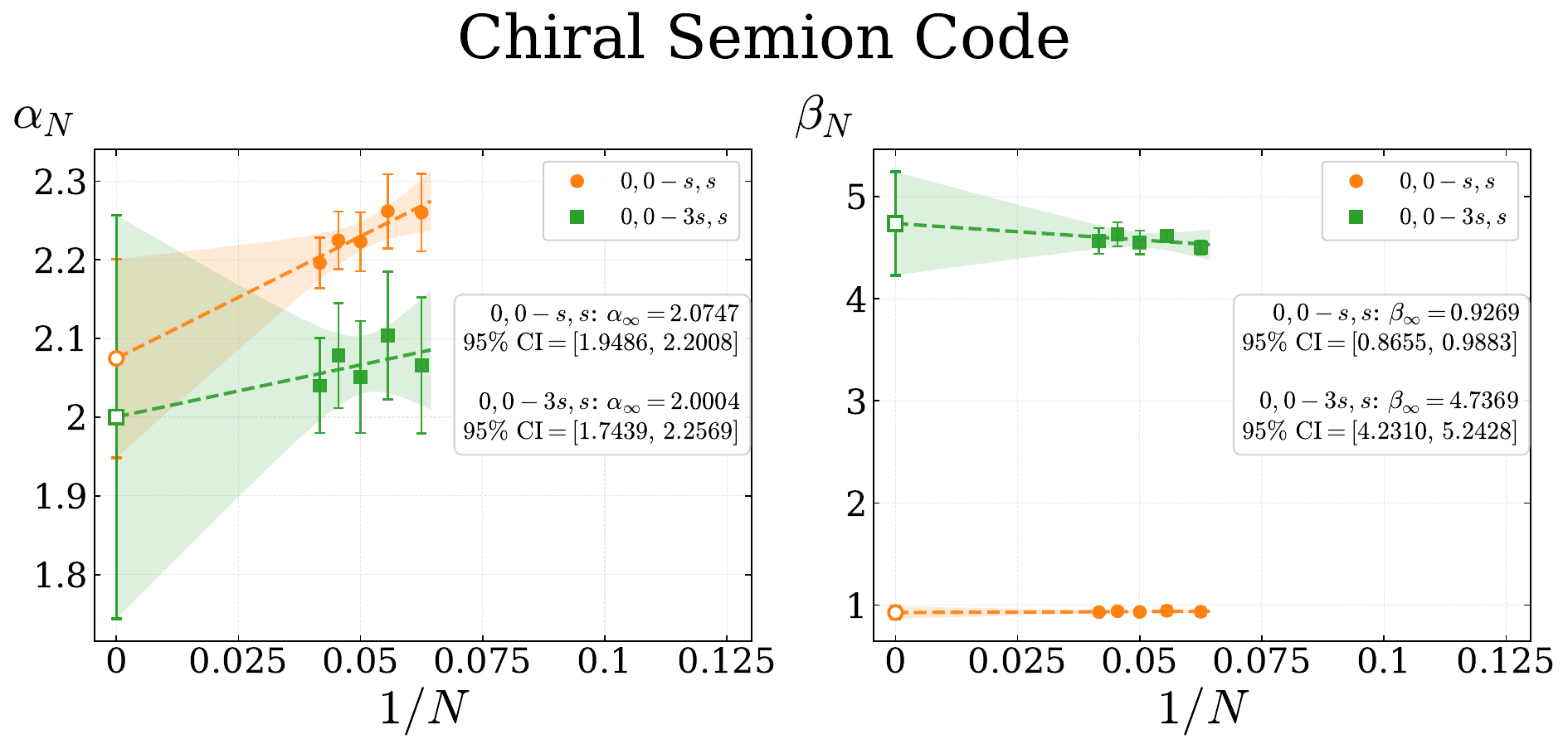}
    \put(1,45){\large\bfseries (c)}
  \end{overpic}
\end{minipage}\hfill%
\begin{minipage}[t]{0.495\linewidth}
  \vspace{0pt}
  \centering
  \begin{overpic}[width=\linewidth]{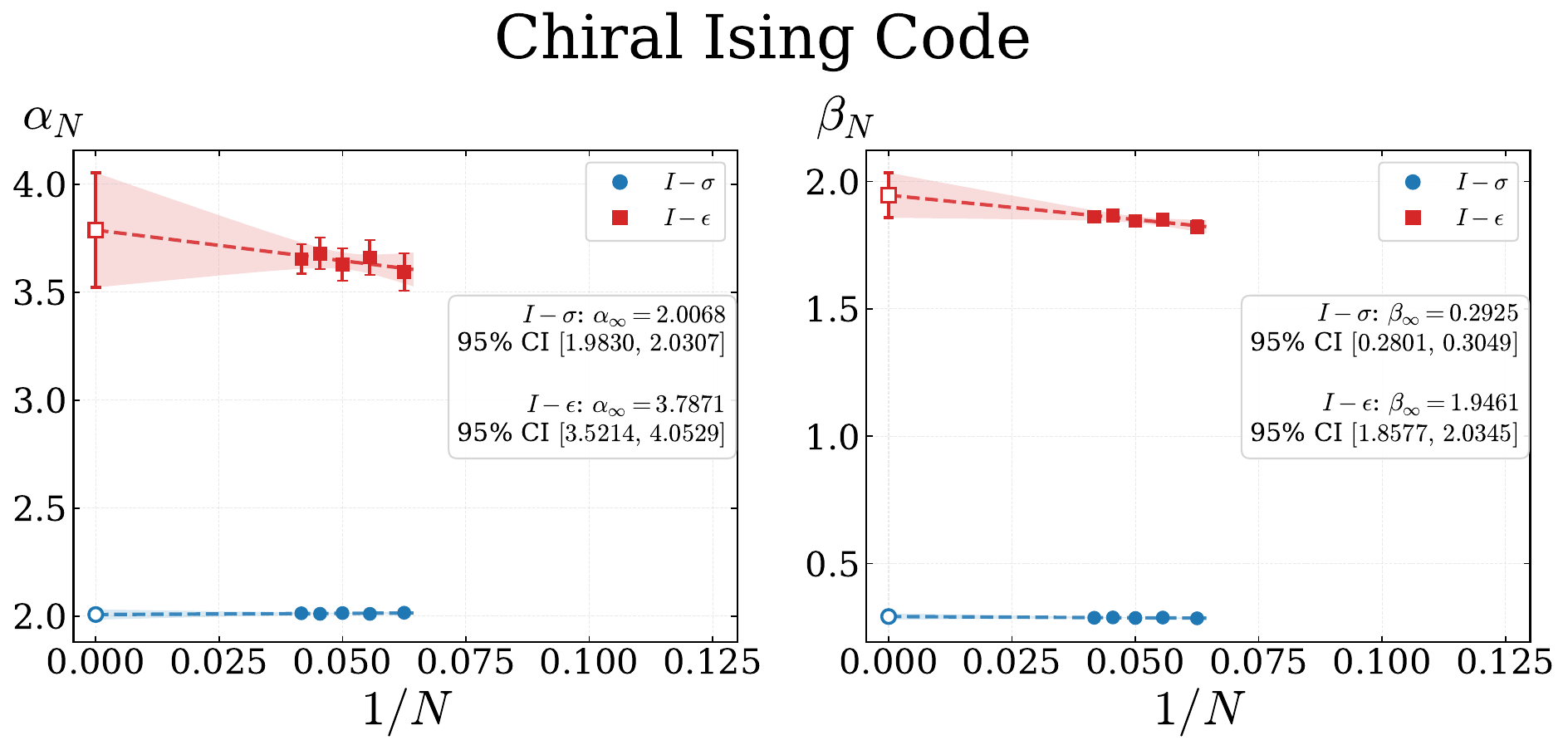}
    \put(1,45){\large\bfseries (d)}
  \end{overpic}
\end{minipage}

\caption{{\bf Fitting $\alpha,\beta$ on a finite 1D circle.}
Numerical scaling relative entropies quantities $\delta_\alpha :=\frac{1}{D} \sum_\mfa S(\rho^\mfa_A||\bar{\rho}_A) \approx c_1 x^\alpha $ and $\delta_\beta :=\log D - \frac{1}{D} \sum_\mfa S(\rho^\mfa_{\bar{A}}||\bar{\rho}_{\bar{A}}) \approx c_2 x^\beta$. Here $A$ is an interval, with width $x$, where $A= P(X)$ in the notion of Fig.~\ref{fig:power-law-region}. The system sizes here are $N=16,18,20,22,24$. 
(a) Verifying the powers of the chiral semion code, with code subspace
\(\mathbb{V}^{\chi}\bigl(\cylindericon[1.2]\bigr)\) as \(\mathrm{span}\{\ket{\psi_1^{0,0}},\ket{\psi_s^{\frac{1}{2},\frac{1}{2}}}\}\) and \(\mathrm{span}\{\ket{\psi_1^{0,0}},\ket{\psi_s^{\frac{3}{2},\frac{1}{2}}}\}\)
explained in Eq.~\eqref{eq:CFT-code-subspace} using a 1D free boson CFT. 
(b) Verifying the powers of chiral Ising code, with two choices of code subspace 
\(\mathbb{V}^{\chi}\bigl(\cylindericon[1.2]\bigr)\) as
\(\mathrm{span}\{\ket{\psi^1},\ket{\psi^\sigma}\}\) and
\(\mathrm{span}\{\ket{\psi^1},\ket{\psi^\epsilon}\}\). We use the dimensional reduction codes states (as 1D Ising CFT primary states) realized by the 1D transverse field Ising chain at the critical point. All the power-law fits are based on the data with $x\leq 0.15$. 
(c) Finite-size extrapolation for the chiral semion codes.
For each size $N$, the data with $x\leq0.25$ are independently fitted
to $\delta_\alpha(N,x)=c_{1,N}x^{\alpha_N},\,\,\delta_\beta(N,x)=c_{2,N}x^{\beta_N}$. The resulting finite-size exponents are extrapolated linearly in $1/N$, using $\alpha_N=\alpha_\infty+\frac{a_\alpha}{N},\,\,
\beta_N=\beta_\infty+\frac{a_\beta}{N}$. The intercepts at $1/N=0$ determine the thermodynamic-limit exponents $\alpha_\infty$ and $\beta_\infty$. 
(d) The same fixed-$N$ fitting and thermodynamic-limit extrapolation
for the chiral Ising codes.  
}
\label{fig:alpha_beta}
\end{figure*}

\subsection{The 2D chiral semion code}
\label{subsec:numerics_compact_boson_beta_relent}

For the 2D chiral semion code, we pick two code subspaces spanned by two states on the cylinder:  
\begin{equation}\label{eq:CFT-code-subspace}
\begin{aligned}
&\mathbb{V}^{\chi} \bigl(\cylindericon[1.2]\bigr)=\mathrm{span}\{\ket{\psi_1^{0,0}},\ket{\psi_s^{\frac{1}{2},\frac{1}{2}}}\}, \quad \text{and} \\
& \mathbb{V}^{\chi} \bigl(\cylindericon[1.2]\bigr)=\mathrm{span}\{\ket{\psi_1^{0,0}},\ket{\psi_s^{\frac{3}{2},\frac{1}{2}}}\}.
\end{aligned}
\end{equation} 
Here the notation of states $\ket{\psi_1^{0,0}}$, $\ket{\psi_s^{\frac{1}{2},\frac{1}{2}}}$ and $\ket{\psi_1^{0,0}},\ket{\psi_s^{\frac{3}{2},\frac{1}{2}}}$ are as explained in 
Example~\ref{exmp:chiral_Semion}. According to Example~\ref{exmp:1D-semion}, the dimension reduction of the chiral semion code is the 1D compact boson RCFT with central charge $c=1$ and compactification radius \(R=\sqrt2\). The three states correspond to the vacuum and primary states with scaling dimension $(1/4, 1/4)$, $(9/4, 1/4)$ of the CFT. Such states can be realized by the lattice wave functions of Ref.~\cite{Nielsen2012,Nielson2014} on a uniform circle.  
Further details are available in Appendix~\ref{app:FQH_wavefunction}.

By dimensional reduction, the powers $\alpha$ and $\beta$ can be computed on the 1D spin chain model. We numerically compute $\alpha$ and $\beta$ on the 1D spin chain with number of qubits $N = 16,18,20,22,24 $ as in Fig.~\ref{fig:alpha_beta}, and find
\begin{equation}
    (\alpha,\beta) \approx 
\begin{cases} 
(2.255,1.011), & \operatorname{span}\{\ket{\psi_1^{0,0}},\ket{\psi_s^{\frac{1}{2},\frac{1}{2}}}\} \\ 
(2.187,4.762), & \operatorname{span}\{\ket{\psi_1^{0,0}},\ket{\psi_s^{\frac{3}{2},\frac{1}{2}}}\}.
\end{cases}
\end{equation}   
Also, we made finite-size extrapolation of the exponents $\alpha$ and $\beta$ as shown in Fig.~\ref{fig:alpha_beta} (c)~\footnote{In panels (c) and (d), the error bars on the filled finite-$N$ symbols represent one standard error of the exponent obtained from the corresponding fixed-$N$ log-log regression. These standard errors quantify the uncertainty of the power-law fits but are not used as weights in the subsequent $1/N$ regression. The open symbols at $1/N=0$ denote the fitted thermodynamic-limit intercepts, and their error bars are two-sided 95\% confidence intervals constructed using Student's $t$ distribution. The shaded regions show the pointwise 95\% confidence bands for the fitted mean linear extrapolations.},
\begin{equation}
    (\alpha_{\infty},\beta_{\infty}) \approx 
\begin{cases} 
(2.075,0.927), & \operatorname{span}\{\ket{\psi_1^{0,0}},\ket{\psi_s^{\frac{1}{2},\frac{1}{2}}}\} \\ 
(2.000,4.737), & \operatorname{span}\{\ket{\psi_1^{0,0}},\ket{\psi_s^{\frac{3}{2},\frac{1}{2}}}\}.
\end{cases}
\end{equation} 

Our numerical finding can also be verified with an analytical replica-trick calculation of relative entropy for CFT: 
\begin{equation}
    (\alpha_{\text{th}},\beta_{\text{th}}) =
\begin{cases} 
(2,1), & \operatorname{span}\{\ket{\psi_1^{0,0}},\ket{\psi_s^{\frac{1}{2},\frac{1}{2}}}\} \\ 
(2,5), & \operatorname{span}\{\ket{\psi_1^{0,0}},\ket{\psi_s^{\frac{3}{2},\frac{1}{2}}}\} ;
\end{cases}
\end{equation}
see Appendix~\ref{app:alpha-beta} and~\cite{Lashkari2026}.
We could see the coherent-information loss on the small interval of 1D CFT is governed by $\min\{\alpha,\beta\} =\beta$ for code subspace $\operatorname{span}\{\ket{\psi_1^{0,0}},\ket{\psi_s^{\frac{1}{2},\frac{1}{2}}}\}$, and $\min\{\alpha,\beta\} =\alpha$ for $\operatorname{span}\{\ket{\psi_1^{0,0}},\ket{\psi_s^{\frac{3}{2},\frac{1}{2}}}\}$.

We further test the exponent \(\gamma\) related to the decoherence of a local region near the edge of the chiral edge code. We compute this by preparing the code states on the cylinder in Fig.~\ref{fig:edge_gamma}. We compute $\gamma$ by directly calculating $\delta_\gamma=\frac{1}{D}\sum_\mfa S(\rho_Y^\mfa\|\bar{\rho}_Y)$.

Such chiral semion wave functions, for both the vacuum and the excited states, are based on the \(\nu=\tfrac12\) analytical lattice Laughlin states in~\cite{Nielsen2012,Nielson2014}, for which we review the relevant details in Appendix~\ref{app:FQH_wavefunction}. The result is shown in~Fig.~\ref{fig:edge_gamma}. In particular, the power-law fit gives \(\gamma\approx2.774,2.471\) for the two code subspaces. This agrees with the theoretical lower bound as $\gamma \ge  \alpha$.

\subsection{2D chiral Ising code}
\label{subsec:numerics_ising_beta_relent}

The Ising topological order has three anyon types, $I, \sigma, \epsilon$, and on the cylinder, the lowest-energy states correspond to the three sectors on the cylinder: $|\psi^1\rangle$, $|\psi^\sigma\rangle$, and $|\psi^\epsilon\rangle$ respectively. We choose two code subspaces  $\mathbb{V}^{\chi} \bigl(\cylindericon[1.2]\bigr)$ as 
\begin{equation}
\begin{aligned}
   & \mathrm{span}\{\ket{\psi^1},\ket{\psi^\sigma}\},\,\,\, \text{and}\\
    & \mathrm{span}\{\ket{\psi^1},\ket{\psi^\epsilon}\}.\,\,\,
\end{aligned}
\end{equation}
The notation is as in Example~\ref{exmp:chiral_Ising}.
Consider the \(1\)D Ising CFT, realized numerically from the transverse field Ising model at the critical point on a circle with periodic boundary conditions. Let the three primary states be \(\ket{\varphi^1}\), \(\ket{\varphi^\sigma}\), and \(\ket{\varphi^\epsilon}\). The dimensional reduction of the above code subspaces becomes CFT codes with $\mathbb{V}^{\rm CFT}$:
\begin{equation}
\mathrm{span}\{\ket{\varphi^1},\ket{\varphi^\sigma}\},\,\,\,
\mathrm{span}\{\ket{\varphi^1},\ket{\varphi^\epsilon}\},\,\,\,
\end{equation}
each of which is two-dimensional. Our method of finding these primary states of the Ising chain is based on the periodic uniform matrix product state (puMPS)~\cite{Zou2017}, as explained in Appendix~\ref{app:ising_pumps_method}. The numerical computation of $\alpha$, $\beta$ for the two code subspaces is in Fig.~\ref{fig:alpha_beta}. We arrive at the data of $\alpha,\beta$, $\gamma$ ($\gamma$ is only numerically computed in Table~\ref{tab:power_law_and_stability_compact_boson_relent}) summarized in Table~\ref{tab:power_law_and_stability_Ising_relent}. The analytical computation of $\alpha$ can be found in Appendix~\ref{app:alpha-beta}.
The dominating exponent associated with the 1D CFT code robustness is $\min\{\alpha,\beta\} =\beta$, for both cases. 

\begin{figure}[h]
    \centering
    \begin{tikzpicture}[
        inner sep=0pt,
        outer sep=0pt
    ]
        \node[anchor=north west] (leftfig) at (0,-0.15)
        {
            \includegraphics[
                width=0.34\linewidth
            ]{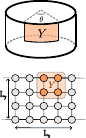}
        };

        \node[anchor=north west] (rightfig) at (0.34\linewidth,0)
        {
            \includegraphics[
                width=0.60\linewidth
            ]{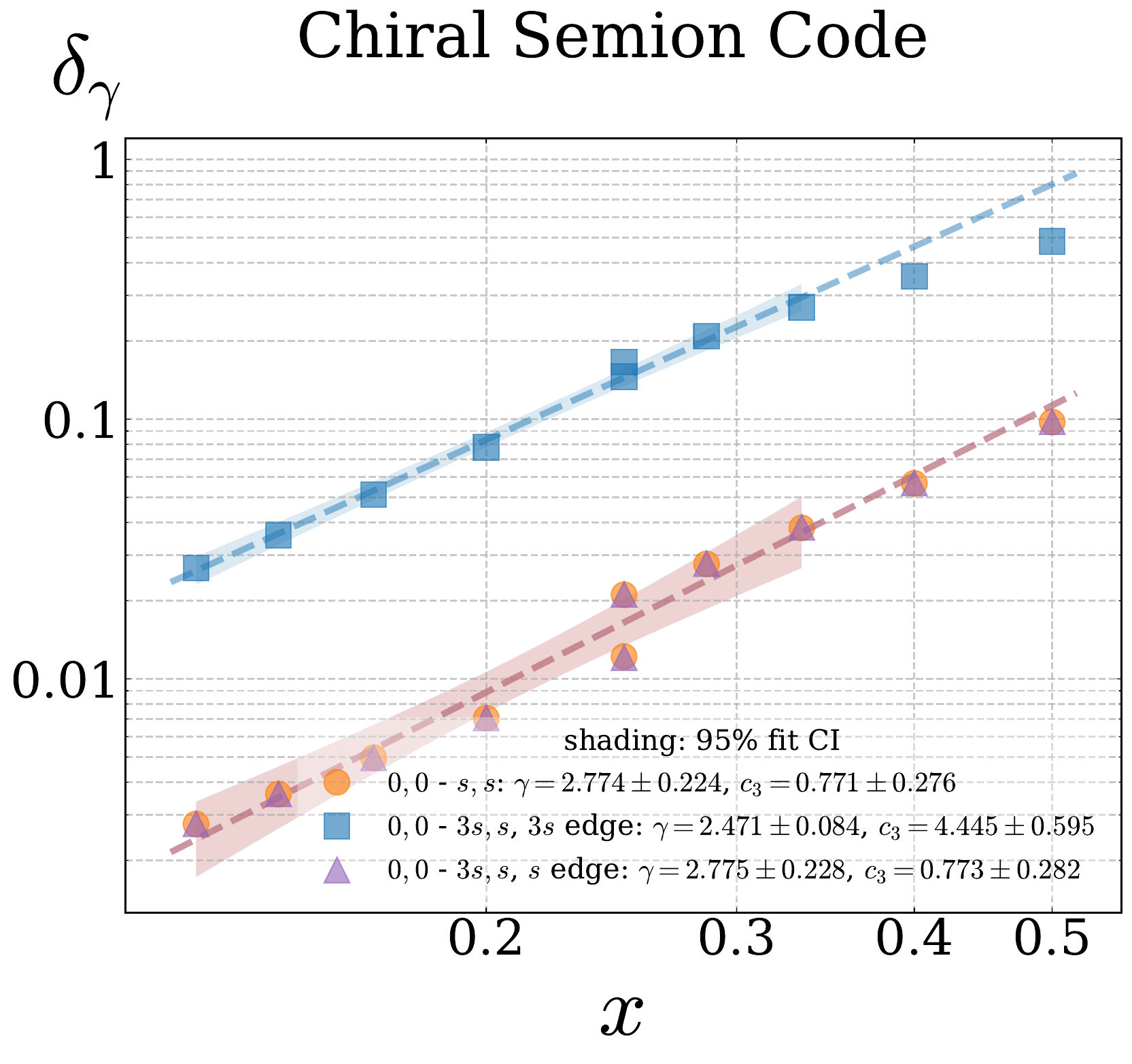}
        };
    \end{tikzpicture}
    \includegraphics[width=0.98\linewidth]{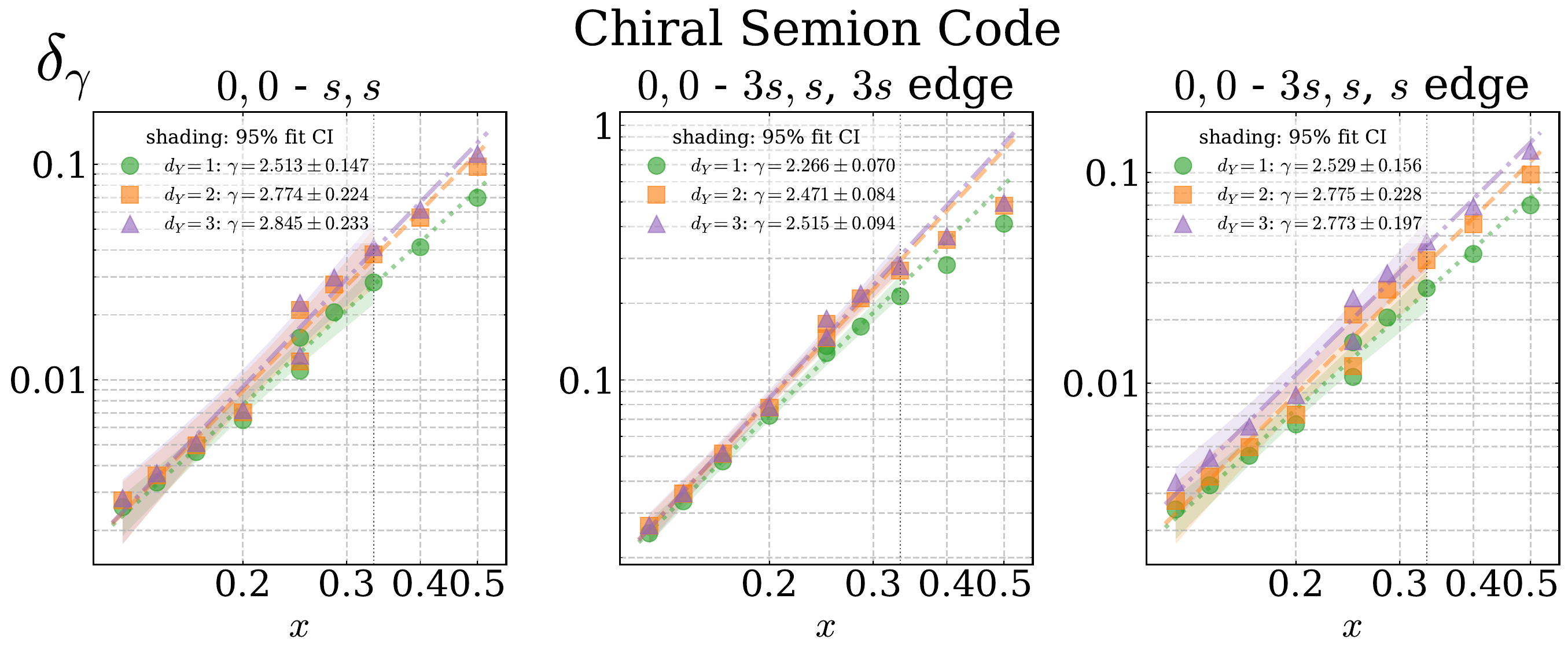}

    \caption{{\bf Fitting $\gamma$ on a finite cylinder.} We consider a square lattice of size $L_x \times L_y$ on cylinder (see Appendix~\ref{app:square_lattice}). For the numerical calculation, we fixed \(L_y=4\), with \(L_x=4,5,6,7,8\). The region \(Y\) is a connected subsystem adjacent to a single physical edge of thickness $d_Y=2$. Power-law fit \(\delta_\gamma=\frac{1}{D}\sum_\mfa S(\rho_Y^\mfa\|\bar{\rho}_Y)\approx c_3x^{\gamma}\) is shown with original data. Here, $\delta_\gamma$ is the coherent information loss $(\Delta(Y;\mathbb{V}^{\chi} \bigl(\cylindericon[1.2]\bigr))$, when the orthogonality error is negligible. The code subspace is \(\mathrm{span}\{\ket{\psi_1^{0,0}},\ket{\psi_s^{\frac{1}{2},\frac{1}{2}}}\}\) and \(\mathrm{span}\{\ket{\psi_1^{0,0}},\ket{\psi_s^{\frac{3}{2},\frac{1}{2}}}\}\). The power-law fit gives \(\gamma\approx2.774\) and \(\gamma\approx2.471\) (for the edge with charge $3/2$), for the fitting of which, only data with $x \le 1/3$ are used. Numerical results for $d_Y=1,2,3$ with $95\%$-confidence interval are also available, as shown in the lower-half. It is worthy of notice that, the behavior of $\delta_\gamma$ in $\mathrm{span}\{\ket{\psi_1^{0,0}},\ket{\psi_s^{\frac{1}{2},\frac{1}{2}}}\}$ is similar to that on the edge with charge $3/2$ in $\mathrm{span}\{\ket{\psi_1^{0,0}},\ket{\psi_s^{\frac{3}{2},\frac{1}{2}}}\}$.}  
    \label{fig:edge_gamma}
\end{figure}

\subsection{Summary of power law exponents}
\label{subsec:numerical_power_law_summary_relent}

We summarize the extracted power-law exponents $\alpha,\beta$ and $\gamma$ in Tables~\ref{tab:power_law_and_stability_compact_boson_relent} and \ref{tab:power_law_and_stability_Ising_relent}. $\alpha_{\rm th}$ and $\beta_{\rm th}$ refer to the theoretical values of $\alpha$ and $\beta$ computed in Appendix~\ref{app:CFT_relative_entropy}, which match our numerical computation up to errors we attribute to finite sizes.

\begin{table}[H]
     \centering
     \setlength{\extrarowheight}{4pt}
     \begin{tabular}{|c|c|c|c|c|c|c|}
     \hline
        \(\mathbb{V}^{\chi}\bigl(\cylindericon[1.2]\bigr)\) & \makecell{\(\alpha\)\\\(\alpha_{\infty}\)} & \(\alpha_{\text{th}}\) & \makecell{\(\beta\)\\\(\beta_{\infty}\)} & \(\beta_{\text{th}}\) &\( \gamma_{\rm 1D}\) & \(\gamma\) \\
        \hline
         \(\mathrm{span}\{\ket{\psi_1^{0,0}},\ket{\psi_s^{\frac{1}{2},\frac{1}{2}}}\}\) & \makecell{\(2.255\)\\\(2.075\)} & \(2\) & \makecell{\(1.011\)\\\(0.927\)} & \(1\) & \( 1.011 \)  & \( 2.774\) \\
         \hline
         \(\mathrm{span}\{\ket{\psi_1^{0,0}},\ket{\psi_s^{\frac{3}{2},\frac{1}{2}}}\}\) & \makecell{\(2.187\)\\\(2.000\)} & \(2\) & \makecell{\(4.762\)\\\(4.737\)} & \(5\) & \(2.187 \) & \( 2.471\) \\
        \hline
     \end{tabular}
     \caption{{\bf Power-law exponents for the chiral semion code.} These exponents $\alpha, \beta$ and $\gamma$ are relevant to the robustness of the chiral edge code and its 1D CFT code obtained by dimensional reduction. $\alpha_{\rm th}$ and $\beta_{\rm th}$ refer to the theoretical values of $\alpha$ and $\beta$. $\gamma_{\rm 1D}:= \min\{\alpha,\beta\}$.}
     \label{tab:power_law_and_stability_compact_boson_relent}
\end{table}

The results illustrate the mechanism of Sec.~\ref{sec:coherent_loss_vs_relent}. In the \(2\)D chiral edge code, the erasure of a connected region is controlled by a power law exponent $\gamma$ 
that is greater or equal to the exponent governing the 1D CFT code robustness $\min\{\alpha,\beta\}$. This is by
\begin{equation}
    \gamma \ge  \alpha \ge \min\{\alpha,\beta\}=: \gamma_{\rm 1D}.
\end{equation}
In three of the four examples, we find clear evidence that $\gamma$ is strictly larger than $\min\{\alpha,\beta\}$, and $\beta < \alpha$ for those examples. 

The remaining case is the second chiral-semion code in Table~\ref{tab:power_law_and_stability_compact_boson_relent}. There, the leading exponent does not show a clear enhancement: $\beta > \alpha$, so the 1D CFT code is controlled by $\alpha$, and the finite-size estimate gives $\gamma \approx \alpha$. Thus, for this particular code subspace, the 2D chiral edge code appears comparable to its 1D dimensional reduction under the local-erasure exponent.
Even in this case, however, the 2D chiral edge realization retains a physical advantage: the conformal degrees of freedom arise as topologically enforced edge modes of a gapped topological phase, rather than from a one-dimensional Hamiltonian tuned to criticality. The exponent comparison captures local distinguishability, while this additional stability reflects the many-body origin of the chiral edge code.

\begin{table}[H]
     \centering
     \setlength{\extrarowheight}{4pt}
     \begin{tabular}{|c|c|c|c|c|c|c|}
     \hline
        \(\mathbb{V}^{\chi}\bigl(\cylindericon[1.2]\bigr)\) & \makecell{\(\alpha\)\\\(\alpha_{\infty}\)} & \(\alpha_{\text{th}}\) & \makecell{\(\beta\)\\\(\beta_{\infty}\)} & \(\beta_{\text{th}}\) & \( \gamma_{\rm 1D}\)  & \(\gamma\geq\)  \\
        \hline
        \(\mathrm{span}\{\ket{\psi^1},\ket{\psi^\sigma}\}\) & \makecell{\(2.002\)\\\(2.007\)} & \(2\) & \makecell{\(0.294\)\\\(0.293\)} &\(\frac{1}{4}\)  & \( 0.294 \) & \( 2\) \\
        \hline
        \(\mathrm{span}\{\ket{\psi^1},\ket{\psi^\epsilon}\}\) & \makecell{\(3.829\)\\\(3.787\)} & \(4\) & \makecell{\(1.913\)\\\(1.946\)}& \(2\) &  \(1.913\) & \( 4\) \\
        \hline
     \end{tabular}
     \caption{{\bf Power-law exponents for the chiral Ising code.} It is presented in parallel to Table~\ref{tab:power_law_and_stability_compact_boson_relent}. The only difference is that $\gamma$ here is bounded using Prop.~\ref{prop:powers} instead of numerical computation on finite sizes.}
     \label{tab:power_law_and_stability_Ising_relent}
\end{table}

\section{Power-law-range recovery channel}\label{sec:quasi-local-R}

In this section, we construct a power-law-range recovery channel that approximately recovers the original state of the chiral edge code for any local noise. By local noise, we mean a noise supported on a local disk. The interesting case is a disk near an edge, e.g., $A$ of $O(1)$ size in Fig.~\ref{fig:quasi-local-recovery}(a); this is because the decoherence of the bulk disk can be recovered perfectly by a local recovery channel, as in the TQFT code Eq.~\eqref{eq:local-recovery}. By a power-law-range channel, we mean a quantum channel whose support is a vanishing fraction of the length $L_x$ of the cylinder in the neighborhood of $A$, as the system size grows, e.g., the region $AB$ of Fig.~\ref{fig:quasi-local-recovery}(a). Importantly, the smallness of coherent information loss for erasure noise on $A$ implies only a recovery channel; it does not imply the geometrical locality. We provide a general theorem (Thm.~\ref{thm:quasi-local-recovery}) on power-law-range recovery.

\begin{figure}[h]
    \centering
    \includegraphics[width=0.85\linewidth]{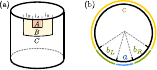}
    \caption{Regions related to the identification of the power-law-range recovery channel. (a) The cylinder of chiral topological order, with $A$ being an $O(1)$ sized disk near an edge where the decoherence happens. $B$ is a buffer that scales sub-extensively with $L_x$. (b) Regions on a circle which are used in defining $\mu^*$. In the context of Thm.~\ref{thm:quasi-local-recovery}, the sizes of $a, b_L, b_R$ are determined by the intervals resulting from $A$ and $B$ touching the upper edge.}
    \label{fig:quasi-local-recovery}
\end{figure}

In order to make a precise statement about the error of recovery, we introduce two information-theoretic quantities that are computable given the knowledge of the code subspace. (Note that the following quantities are different from $\alpha,$ $\beta$ and $\gamma$ considered before.) We define
\begin{equation}
    \gamma^* := 
    \min_{\substack{\mfa,\mfb \in \calA_{\rm code}, \\ \mfa\ne \mfb}}
    \gamma_{\mfa\mfb}.
    \label{eq:gamma-a-def}
\end{equation}
where $\gamma_{\mfa\mfb}$ is computed from relative entropy $S(\rho^\mfa_Y || \rho^\mfb_{Y}) \sim x^{\gamma_{\mfa \mfb}}$ for $Y$ in Fig.~\ref{fig:edge_gamma}. Moreover, we define
\begin{equation}
    \mu^* := \min_{\mfa\in \calA_{\rm code}} \mu_\mfa,
\end{equation}
where $\mu_\mfa$ is defined relating to the disjoint interval mutual information $I(a:c)_{|\varphi^\mfa\rangle}$ through an empirical relation
\begin{equation}\label{eq:mu-a-def}
    I(a:c)_{|\varphi^\mfa\rangle} \sim \eta^{\mu_{\mfa}},
\end{equation}
 at small $\eta$, where $|\varphi^\mfa\rangle$ is the dimensional reduction of a code word and $\eta$ is the cross-ratio. We refer to Appendix~\ref{app:mu-computation} for the precise definition of $\eta$.

\begin{tcolorbox}[breakable, enhanced, colback=yellow!5!white,colframe=orange!35!white]
\begin{theorem}[Power-law-range Recovery map]\label{thm:quasi-local-recovery}
Consider a chiral edge code on a cylinder of circumference $L_x$ and width $L_y$, with code subspace $\mathbb{V}^{\chi} \bigl(\cylindericon[1.2]\bigr)
:=\mathrm{span} \left\{\ket{\psi_Q^\mfa}:\mfa\in\mathcal A_{\rm code}\right\}$, where each code state is associated with some Abelian anyon sector. 
Then, for any quantum channel $\calN_A$ supported on a local $O(1)$-sized disk $A$ adjacent to the edge as illustrated in Fig.~\ref{fig:quasi-local-recovery}, with length of $A$ projected to the $x$ direction $l_A=O(1)$, there exists an approximate recovery channel $\calR_{AB}$ supported on $AB$ such that the size of $B$ scales sub-extensively with $l_B  \sim L_x^{\lambda^*}$, where
\begin{equation}
    \lambda^* = \frac{\gamma^{*}}{\gamma^{*} + \mu^{*}},
\end{equation}
and the recovery infidelity is upper bounded by
\begin{equation}
    1-F( \rho, \calR_{AB} \circ \calN_A (\rho))  \le  O(1) L_x^{-\frac{\mu^* \, \gamma^*}{\mu^* + \gamma^*}},
\end{equation}
for any density matrix $\rho$ in the code subspace, in which $F(\rho,\sigma)=\|\sqrt{\rho}\sqrt{\sigma}\|_1^2$. The recovery map can be chosen to be independent of the code state and the channel $\calN_A$ as $\calR_{AB}:= \calE_{B\to BA}\circ \Tr_{A}$, where $\calE_{B\to BA}$ is the universal twirled Petz map~\cite{JRSWW-universal-recovery} constructed from $\bar{\rho}$, the maximally mixed state in the code subspace.
\end{theorem}
\end{tcolorbox}

See Appendix~\ref{app:quasi_local_recovery_channel} for the proof of the theorem, which needs the dimensional reduction assumption (Assumption~\ref{asmp:UVfinite_entanglement_map_columns}). The Abelian restriction enters through the Markov property used to control the conditional mutual information in the proof, see Appendix~\ref{app:gamma-computation}. We also remark that, the sub-extensive length scale $l_B\sim L_x^{\lambda^*}$ only needs to occur near the edge. It is enough to have $O(1)$ bulk correlation length thickness of $B$ away from the edge.

\begin{exmp}[Recovery channel for the chiral Ising code]
For the code subspace spanned by $|\psi^1\rangle$ and $|\psi^\sigma\rangle$, the anyon sector $\sigma$ is not Abelian, so we do not consider it here. 
For the Abelian code subspace spanned by $|\psi^1\rangle $ and $|\psi^\epsilon\rangle$, we have
\begin{equation}
    \gamma^* = \min\{\gamma_{I \epsilon}, \gamma_{\epsilon I} \} =4.  
\end{equation}
This is explained in Appendix~\ref{app:gamma-computation}.
The exponent $\mu^*$ is
\begin{equation}
    \mu^*  = \min \{ \mu_1, \mu_\epsilon \} \approx \min \{ 0.297,2.00 \} =  0.297
\end{equation}
according to the finite-size estimation shown in Appendix~\ref{app:mu-computation}.
From this we can estimate
\begin{equation}
    \lambda^* \approx 0.931
\end{equation}
giving recovery infidelity scaling as $\sim L_x^{-0.276}$. These results can be compared with $\lambda^*_{\text{th}}=\frac{16}{17}\approx 0.941$ and recovery infidelity scaling $\sim L_x^{-\frac{4}{17}}\approx L_x^{-0.235}$, based on the theoretical value $\mu_{1\,\text{th}}=\frac{1}{4}$ as shown in~\cite{2009JSMTE..11..001C,2011JSMTE..01..021C}.

    
\end{exmp}

\begin{exmp}[Recovery channel for the chiral semion code]

For code subspace $\operatorname{span}\{|\psi_1^{0,0}\rangle,|\psi_s^{\frac{1}{2},\frac{1}{2}}\rangle\}$,  
\begin{equation} 
    \gamma^* = \min\{ \gamma_{1s}, \gamma_{s1} \} = 2, 
\end{equation}
which is derived 
in Appendix~\ref{app:gamma-computation}, thanks to the Abelian nature of semion $s$. We note the dependence of the above quantities $\gamma_{1s}, \gamma_{s1}$ on the code subspace. The exponent $\mu^*$ is
\begin{equation}
    \mu^* = \min \{ \mu_{1}, \mu_{s}\} \approx 1.071,
\end{equation}
from $\mu_{1} \approx 1.071, \mu_{s} \approx 1.102$ as identified in Appendix~\ref{app:mu-computation}. 
From this we can estimate 
    \begin{equation}
    \lambda^* \approx 0.651.
\end{equation}
giving recovery infidelity scaling as $\sim L_x^{-0.697}$. These results can be compared with $\lambda^*_{\text{th}}=\frac{2}{3}$ and recovery infidelity scaling $\sim L_x^{-\frac{2}{3}}$, based on the theoretical value $\mu_{1\,\text{th}}=1$ as shown in~\cite{2009JSMTE..11..001C,2011JSMTE..01..021C}.

For code subspace $\operatorname{span}\{|\psi_1^{0,0}\rangle,|\psi_s^{\frac{3}{2},\frac{1}{2}}\rangle\}$, 
\begin{equation}
    \gamma^* =\min\{ \gamma_{1s}, \gamma_{s1} \} = 2
\end{equation}
and
\begin{equation}
    \mu^{*}=\min\{\mu_{{1}},\mu_{s}\}\approx 1.071.
\end{equation}
Thus, we have
\begin{equation}
    \lambda^*\approx 0.651.
\end{equation}
giving recovery infidelity scaling as $\sim L_x^{-0.697}$. Similarly, we have $\lambda^*_{\text{th}}=\frac{2}{3}$ and recovery infidelity scaling $\sim L_x^{-\frac{2}{3}}$, based on $\mu_{1\,\text{th}}=1$.
\end{exmp}  
In this regard, the chiral semion code has better locality in terms of the power-law-range recovery channel we identify, for the Abelian code subspaces we consider.

\section{Discussion}\label{sec:discussion}

In this work, we introduced \emph{chiral edge codes}, a family of approximate quantum error-correcting codes whose codewords are cylinder primary states associated with distinct bulk anyon sectors. The construction offers two advantages over a standalone one-dimensional CFT code. First, its gapless degrees of freedom arise at the boundary of a stable gapped phase, rather than from tuning a microscopic Hamiltonian to criticality. Second, the two-dimensional geometry imposes a stronger notion of locality: a geometrically local error near one edge does not simultaneously probe the opposite edge, whereas dimensional reduction treats a full column connecting the two edges as a local interval.

We quantified this geometric advantage using coherent-information loss under local erasure. When the complement of the erased region contains a noncontractible bulk annulus, the coherent-information loss reduces to an average relative entropy on the erased region. In the small-region limit, the resulting exponents obey $\gamma\geq\alpha\geq\min\{\alpha,\beta\}$, showing that the two-dimensional encoding is never worse at the level of the small-region exponent, with strict enhancement established or numerically supported in three of the four code subspaces considered. For Abelian code subspaces, we further constructed a power-law-range recovery channel supported on the erased edge region together with a subextensive buffer.

These results reveal a hybrid protection mechanism. The gapped bulk provides a nonlocal decomposition into anyon sectors and spatially separates the two physical edges, while the edge theory and bulk--edge geometry control the residual algebraic information leakage near the boundary. This allows CFT relative entropy and disjoint-interval mutual information to diagnose approximate quantum error correction, while entanglement-bootstrap methods identify the encoded sectors and constrain the spatial support of recovery. Chiral edge codes therefore combine the robust realization of gapless boundary degrees of freedom with a geometric local-erasure advantage unavailable in the one-dimensional description. Our work suggests four main future directions.

\begin{itemize}[leftmargin=13pt]
    \item {\bf Microscopic foundations.} A central problem is to derive the dimensional-reduction correspondence and full-boundary entanglement conditions for microscopic chiral topological phases, with controlled finite-size and finite-correlation-length errors. This would clarify which leakage and recovery exponents are universal and which depend on the boundary realization or code subspace. Since the present power-law-range recovery theorem relies on the full-boundary Markov property for Abelian sectors, extending spatially local recovery to non-Abelian sectors may require a formulation that incorporates their fusion-space structure.
    \item {\bf Geometry-dependent erasures.} The main text considers a connected disk-like erasure adjacent to one physical edge. Disconnected regions, regions touching both edges, and regions winding nontrivially around the cylinder can probe different aspects of the bulk--edge encoding; for some such geometries, the complement no longer contains a noncontractible annulus, so the present relative-entropy reduction need not apply. Their coherent-information loss and recovery range may therefore depend on topology and edge connectivity, not only on size. Accordingly, the effective distance $d^*(\delta)$ introduced in Sec.~\ref{subsec:theoretical-summary} is a geometry-resolved diagnostic rather than a worst-case code distance. A natural next step is to define effective distances indexed by erasure geometry and determine how their scaling constrains local recovery. This viewpoint may also inform information-theoretic diagnostics of memories in mixed-state phases~\cite{PRXQuantum.5.020343,hlfh-86yz,PhysRevA.111.032402,2025arXiv251222121V,Sang2025Markov-length,Yang2025mixed}; punctured coherent information provides a related approach~\cite{Negari2026}.

    \item {\bf Logical operations and decoding.} The universal twirled Petz map establishes recovery with controlled spatial support, but does not provide an efficient microscopic decoder. An important question is whether the Petz map can be approximated by local circuits, tensor-network algorithms, or experimentally accessible measurements. For related work on decoding one-dimensional CFT codes, see Ref.~\cite{Zhang2025}. 
    The bulk--edge geometry may also support logical operations through deformations of the edge~\cite{You2015,Zhu2018}, spatial rotations~\cite{Wang2024}, or adiabatic motion of edges and interfaces~\cite{Cong2017}. Geometries with three or more edges may allow additional couplings between encoded sectors. Constructing explicit decoders and logical-gate protocols is therefore necessary to assess the computational utility of chiral edge codes.  
    \item {\bf Stochastic noise and thresholds.} A central question is whether the advantage established for local erasure persists under spatially extensive noise. Appendix~\ref{app:IID_noise} provides finite-size evidence that, for selected Pauli channels and one chiral-semion code subspace, the two-dimensional realization exhibits a more favorable weak-noise flow of coherent information than its dimensionally reduced CFT counterpart. These data neither establish a threshold nor identify its controlling mechanism. It remains to determine whether chiral edge codes are recoverable under sufficiently weak stochastic noise that applies to the entire system and whether their stability is governed by the edge-local exponent $\gamma$, other universal data, or microscopic details. 
    This question is distinct from that for fractional quantum Hall memories based on bulk topological-sector encodings on a torus or non-Abelian fusion-space encodings~\cite{2025arXiv251008490W}, because the physical edge degrees of freedom participate directly in the encoding considered here. 
    The role of this bulk--edge structure can be isolated particularly sharply by comparing the chiral Ising code with its dimensionally reduced Ising CFT code. Sufficiently weak independent and identically distributed (IID) Pauli noise is not correctable in the latter~\cite{Sang2024AQECC}, whereas the additional bulk-edge geometry may alter the stability of the former. Whether the chiral Ising code remains recoverable under a general class of sufficiently weak IID noise---and, if so, whether the resulting threshold is controlled by universal edge data---is an important open problem.
     
\end{itemize}
Several further extensions may support these main directions. The construction may extend to selected descendant states, ungappable edges with counterpropagating modes~\cite{PhysRevX.3.021009,Kaidi2021}, and configurations in which defect sectors replace ordinary anyon sectors in the annulus. The analysis also motivates CFT calculations of relative entropy involving mixtures of primary states and disjoint-interval mutual information in excited states, which determine the leakage and recovery exponents appearing here~\cite{Lashkari2015,Sarosi2016,Sarosi2017,Ugajin2017}. Although the relevant error-correction quantities are von Neumann quantities, their R\'enyi analogues may provide a useful analytical and numerical route through replica continuation and stabilized extrapolation to $n\to1$~\cite{Vijay2025}. More broadly, it would be valuable to identify which ingredients of the present construction---a robust sector decomposition, protected gapless degrees of freedom, and algebraically suppressed local distinguishability---can arise in other forms of many-body chirality~\cite{Kim_2022,Zou2022chiral,Vardhan2025,2026arXiv260620472E}.

\section*{Acknowledgments}

We thank Bryan Clark for the discussion of numerical methods, Leonid P. Pryadko and Jinmin Yi for the discussion about how to quantify a good code by the exponents, Yijian Zou for answering questions about the decoding method of 1D CFT code, Dominic Williamson, Xiang Li, Ting-Chun Lin, John McGreevy, Isaac Kim, Akash Vijay, and Yuta Hirasaki for interesting discussions related to chiral topological ordered edges or local decoherence, and Nima Lashkari, Jignesh Mohanty and Tom Faulkner for discussions related to relative entropy computations. This work made use of the Illinois Campus Cluster, a computing resource that is operated by the Illinois Campus Cluster Program (ICCP) in conjunction with the National Center for Supercomputing Applications (NCSA) and which is supported by funds from the University of Illinois at Urbana-Champaign. BS and JYL are supported by the IQUIST fellowship, faculty startup grant at the University of Illinois, Urbana-Champaign, and IBM-Illinois Discovery Accelerator Institute. 
BS gratefully acknowledges the hospitality of the Isaac Newton Institute, Tsinghua University, and the Perimeter Institute, where parts of this work were carried out during research visits.  


\appendix

\section{Relative entropy in RCFT}
\label{app:CFT_relative_entropy}

In this appendix, we discuss the relative entropy between two primary states, as well as the relative entropy between a primary state and a mixture of primary states in 1D RCFT. This will explain the origin of the power-law exponents $\alpha$ and $\beta$ defined in Sec.~\ref{sec:coherent_loss_vs_relent} and Fig.~\ref{fig:Srel}, when $x\ll 1$:
\begin{align}
    x^\alpha &\sim \frac{1}{D}\sum_{\mfa \in \calA_{\rm code}} S(\rho^\mfa_X|| \bar{\rho}_X),  \\
   x^\beta &\sim \log D - \frac{1}{D}\sum_{\mfa \in \calA_{\rm code}} S(\rho^\mfa_{\bar{X}}|| \bar{\rho}_{\bar{X}}),  
\end{align}
as well as the exponents $\alpha_{\mfa \mfb}$
\begin{equation}
    x^{\alpha_{\mfa \mfb}}\sim S(\rho_X^\mfa\|\rho_X^\mfb),
\end{equation}
which we use in Appendix~\ref{app:gamma-computation}. We remind the reader that, upon dimensional reduction, the above formulas are equivalent to the relative entropy formulas for a single interval $A\subset S^1$ in 1D CFT, and again $x:=\operatorname{Arc}(A)/2\pi$, therefore $0<x<1$.
 
We will first collect various relative entropy formulas, in particular formulas for compact free boson CFT and Ising CFT. Unfortunately, analytic formulas for general $0<x<1$ are only available in a small number of special examples \cite{Lashkari2014,Lashkari2015,Ruggiero2017,Ugajin2017}, and only for relative entropy between two primaries. However, as we will see, there are universal formulas for relative entropy of $A$ and $\bar A$ in the limit $x\ll 1$, as a result of the operator product expansion (OPE) used in computing the replica correlators. For relative entropy between two primaries, such formulas already exist for the interval $A$ when $x \ll 1$ \cite{Sarosi2016,Sarosi2017,Lashkari2026}. For relative entropy between a primary and a mixture of primaries, the universal formulas are available for \textit{both} the interval $A$ and its complement $\bar A$, the details of which will be presented in future work~\cite{Lashkari2026}; here we summarize the formulas needed for the present work. Finally, we give a brief review of the replica method used to obtain such results.

\subsection{Relative entropy between different primaries}\label{app:Rel-a-b}
For relative entropy between two different primaries, there exist analytic formulas for any $0<x<1$ in free boson CFT \cite{Lashkari2014,Lashkari2015,Ruggiero2017} and Ising CFT \cite{Ugajin2017}, as we will discuss in the examples later. Beyond these examples, exact all-\(x\) results are rare. However, the leading small-interval ($x\ll 1$) behavior of the relative entropy is known. 

For two spinless (i.e. the conformal weights $h=\bar h$) primary states \(a,b\) with different conformal weights $(h_a,\bar h_a)\neq(h_b,\bar h_b)$, 
the leading order
relative entropy is \cite{Sarosi2016,Sarosi2017}
 \begin{equation} \label{eq:rel_ent_cft}
 \begin{aligned}
     &S(\rho_A^a\|\rho_A^b)
\\&=
\frac{\sqrt{\pi}\Gamma(\Delta+1)}
{4\Gamma \left(\Delta+\frac32\right)}
\sum_{\mathcal{O}_p\in\mathcal L}
\left(
C_{p a^* a}
-
C_{p b^* b}
\right)^2
(\pi x)^{2\Delta}
+\cdots ,
 \end{aligned}
 \end{equation}
where $C_{p a^* a}$ and $C_{p b^* b}$ are OPE coefficients, \(\mathcal L\) is the set of lightest operators $\mathcal{O}_p$ with $C_{p a^* a}-C_{p b^* b}\neq 0$, and \(\Delta=h+\bar h\) is their scaling dimension. Note that the above formula holds only when such $\mathcal{O}_p$ have $\Delta\leq 2$. When $\Delta>2$, the leading order formula is 
 \begin{equation}
     S\left(\rho^a_{A} \| \rho^b_{A}\right)=\frac{16}{15} \frac{1}{c}\left(h_a-h_b\right)^2(\pi x)^4 +\cdots 
     \label{eq:rel_ent_cft_2} 
 \end{equation}
where $c$ is the CFT central charge. 

When $a$ and $b$ have the same conformal weights, the relative entropy is always given by Eq. $\eqref{eq:rel_ent_cft}$. Similar formulas can be derived for fields with spin, i.e., $h\neq \bar h$ \cite{Lashkari2026}.

From Eq.~\eqref{eq:rel_ent_cft} and~\eqref{eq:rel_ent_cft_2}, we notice that the relative entropy between primaries $a$ and $b$ has an exchange symmetry $S(\rho_A^a\|\rho_A^b)=S(\rho_A^b\|\rho_A^a)$ to leading order in $x\ll 1$.

From these formulas one can easily find the exponent $\alpha_{\mfa\mfb}$ used in Appendix~\ref{app:gamma-computation}.
Below we show the relative entropy formulas in some explicit examples.    

\subsubsection{1D compact free boson CFT ($c=1$)}
\label{app:CFT_relative_ent_compact_boson}

In compact free boson CFT, we are interested in the primary states that correspond to vertex operators. Given a compactification radius $R$, the vertex operators are labeled by a pair of integers $(m,n)$ (see \cite{DiFrancesco:1997nk,Thorngren2021} for nice reviews on this)
\begin{equation}
\label{eq:vertex_op}
   V(z,\bar z)
   =
   :\exp \left(
   m\frac{i}{\sqrt R}\phi(z)
   +
   n\frac{i}{\sqrt R}\bar \phi(\bar z)
   \right): ,
\end{equation}
where $\phi(z)$ and $\bar \phi(\bar z)$ are the holomorphic and anti-holomorphic parts of the scalar field.  

For any $0<x<1$, the relative entropy between two primaries with integer labels $(m_a,n_a)$ and $(m_b,n_b)$ (as in Eq. \eqref{eq:vertex_op}) is derived by Lashkari \cite{Lashkari2014,Lashkari2015},
\begin{equation}
\label{eq:free_boson_Srel_exact}
\begin{aligned}
    &S(\rho^a_A||\rho^b_A) \\
    =& \frac{1}{R}\left[(m_a-m_b)^2+(n_a-n_b)^2\right](1 - \pi x \cot(\pi x)).
\end{aligned}
\end{equation}

For $x\ll 1$, we can easily find the relative entropy formulas by expanding Eq. \eqref{eq:free_boson_Srel_exact}. But one can also use a slight generalization of the universal formula $\eqref{eq:rel_ent_cft}$ to find:
\begin{equation}
\begin{aligned}
    &S(\rho^a_A||\rho^b_A) \\
    =&  \frac{1}{3R}\left[(m_a-m_b)^2+(n_a-n_b)^2\right](\pi x)^2 +\cdots 
\end{aligned}
\end{equation}
This formula comes from the fact that the relevant lightest primaries are always the holomorphic and anti-holomorphic $U(1)$ currents $J$ and $\bar J$. They have conformal weights $(h_J,\bar h_J)=(1,0)$ and $(h_{\bar J},\bar h_{\bar J})=(0,1)$, and OPE coefficients
\begin{equation}
    C_{JV^\dagger V}=-\frac{m}{\sqrt{R}},\quad C_{\bar JV^\dagger V}=-\frac{n}{\sqrt{R}}.
\end{equation}

In this paper, we particularly considered the cylinder primary states in chiral semion topological order that correspond to the following vertex operators
\begin{equation}
\begin{aligned}
        |\varphi_1^{0,0}\rangle \leftrightarrow V(z,\bar{z}) &= \mathbf{I} \\
        |\varphi_s^{\frac12,\frac12}\rangle \leftrightarrow V(z,\bar{z})&= :\exp(\frac{i}{\sqrt{2}}\phi(z)+\frac{i}{\sqrt{2}}\bar{\phi}(\bar{z})):\\
        |\varphi_s^{\frac32,\frac12}\rangle \leftrightarrow V(z,\bar{z})&= :\exp(i\frac{3}{\sqrt{2}}\phi(z)+\frac{i}{\sqrt{2}}\bar{\phi}(\bar{z})):\\
        |\varphi_s^{\frac12,-\frac12}\rangle \leftrightarrow V(z,\bar{z})&= :\exp(\frac{i}{\sqrt{2}}\phi(z)-\frac{i}{\sqrt{2}}\bar{\phi}(\bar{z})):.
\end{aligned}
\end{equation}
For later convenience of labeling states, we also use the following notation,
\begin{equation}
\begin{aligned}
    |\varphi_1^{0,0}\rangle 
    & \to |0,0\rangle, \qquad
    \ket{\varphi_s^{\frac{1}{2},\frac{1}{2}}}
     \to |s,s\rangle, \\
    \ket{\varphi_s^{\frac{3}{2},\frac{1}{2}}}
    & \to |3s,s\rangle, \qquad
   \ket{\varphi_s^{\frac{1}{2},-\frac{1}{2}}}
     \to |s,t\rangle.
\end{aligned} 
\end{equation}

Applying the above formulas, we find the following relative entropy to leading order in $x$:
\begin{align}
\label{eq:nu_half_primary_relent}
    S(\rho_A^{s,s}\|\rho_A^I)
    &=
    \frac13(\pi x)^2 +\cdots \\
    S(\rho_A^{3s,s}\|\rho_A^{I})
    &=\frac{5}{3}(\pi x)^2 +\cdots. 
\end{align}
Other cases can be worked out similarly.  

\subsubsection{1D Ising CFT ($c=\tfrac12$)}
\label{app:CFT_relative_ent_Ising}
Next we move on to the Ising CFT.
Full Ising primaries are $I$, $\sigma(z,\bar z)$, $\varepsilon(z,\bar z) $ with
$
(h_\sigma,\bar h_\sigma)=\Big(\tfrac{1}{16},\tfrac{1}{16}\Big),\quad
(h_\varepsilon,\bar h_\varepsilon)=\Big(\tfrac{1}{2},\tfrac{1}{2}\Big)
$
and OPE coefficients
$
C_{I\sigma\sigma}=1, C_{\varepsilon\sigma\sigma}=\tfrac12,
C_{I\varepsilon\varepsilon}=1,C_{\varepsilon\varepsilon\varepsilon}=0.
$

For Ising CFT, there exist analytic relative entropy formulas for all $0<x<1$ between certain pairs of primaries, derived in \cite{Ugajin2017}:
\begin{equation}
\begin{aligned}
    &S(\rho^\sigma_A || \rho^I_A) = S(\rho^I_A || \rho^\sigma_A) = \frac{1}{4}( 1 - \pi x \cot(\pi x)),\\
    & S(\rho_A^{\epsilon} || \rho_A^{I})=2(\log (2 \sin \pi x)+1-\pi x \cot (\pi x)
     \\ & +\psi_0\left(\frac{\csc \pi x}{2}\right)+\sin \pi x),\\
     &S(\rho_A^{\epsilon}||\rho_A^{\sigma})= S(\rho_A^{\epsilon} || \rho_A^{I})+ S(\rho_A^{I} || \rho_A^{\sigma}),
\end{aligned}
\end{equation}
where $\psi_0(x)$ is the digamma function. Note that we have an exchange symmetry $S(\rho^\sigma_A || \rho^I_A)=S(\rho^I_A || \rho^\sigma_A)$~\cite{Ugajin2017}, but it is not clear that this symmetry exists for $I,\epsilon$ and $\sigma,\epsilon$. 

In the case $x\ll 1$, we can use the universal formula \eqref{eq:rel_ent_cft} and \eqref{eq:rel_ent_cft_2} to find the relative entropy to leading order in $x$.
For $\sigma$ versus $I$, the lightest operator distinguishing them is
\(\varepsilon\), with scaling dimension \(\Delta_\varepsilon=1\). Thus
\begin{equation}
\label{eq:ising_sigma_I_relent}
    S(\rho_A^\sigma\|\rho_A^I)
    =
    \frac13
    \left(C_{\varepsilon\sigma\sigma}-C_{\varepsilon II}\right)^2
    (\pi x)^2
    =
    \frac{1}{12}(\pi x)^2 .
\end{equation}
Similarly,
\begin{equation}
\label{eq:ising_sigma_eps_relent}
    S(\rho_A^\sigma\|\rho_A^\varepsilon)
    =
    \frac13
    \left(C_{\varepsilon\sigma\sigma}
    -
    C_{\varepsilon\varepsilon\varepsilon}\right)^2
    (\pi x)^2
    =
    \frac{1}{12}(\pi x)^2 .
\end{equation}
For \(I\) versus \(\varepsilon\), the \(\varepsilon\)-exchange contribution
vanishes, and the leading term is the universal stress-tensor contribution:
\begin{equation}
\label{eq:ising_I_eps_relent}
    S(\rho_A^I\|\rho_A^\varepsilon)
    =
    \frac{16}{15}\frac{1}{c}
    \left(h_\varepsilon-h_I\right)^2
    (\pi x)^4
    =
    \frac{8}{15}(\pi x)^4
\end{equation}
The other cases can be obtained from the above by the $a\leftrightarrow b$ exchange symmetry in the universal formula.

\subsection{Relative entropy with a probabilistic mixture}\label{app:alpha-beta}

Next, we present the relative entropy formulas between a primary and a mixture of primaries, which can be used to determine the exponents $\alpha$ and $\beta$.

\subsubsection{Analytic results of $\alpha$}\label{app:A-alpha}
Before introducing the explicit results, for the small interval case, one may use joint convexity of relative entropy~\cite{1974CMaPh..39..111L,carlen2010trace}
\begin{equation}
    S(\sum_i \lambda_i \rho_i||\sum_i \lambda_i \sigma_i) \leq \sum_i \lambda_i S(\rho_i||\sigma_i)
    \label{eq:joint_convexity}
\end{equation}
to find an upper bound 
\begin{equation}
    \frac{1}{D}\sum_{a=1}^DS(\rho_A^a\|\bar{\rho}_A)\leq \frac{1}{D^2}\sum_{a,b=1}^DS(\rho_A^a\|\rho_A^b), 
\end{equation}
where $\bar{\rho}_A=\frac{1}{D}\sum_{a=1}^D\rho_A^a$. 

Through a replica-trick calculation, we can derive the relative entropy between a primary state and the equal mixture when $D=2$. The answer is surprisingly simple~\cite{Lashkari2026}: to leading order in $x$, the answer compared with Eq.~\eqref{eq:rel_ent_cft} is
\begin{equation}
\label{eq:rel_ent_cft_alpha}    S(\rho_A^a\|\frac{\rho_A^a+\rho_A^b}{2})=S(\rho_A^b\|\frac{\rho_A^a+\rho_A^b}{2})=\frac{1}{4}  S(\rho_A^a\|\rho_A^b).
\end{equation}

We therefore find for a small interval $A$
\begin{equation}
    \frac{1}{2}\sum_{a=1}^2S(\rho_A^a\|\bar{\rho}_A)=\frac{1}{2}\left[\frac{1}{4}\sum_{a,b=1}^2S(\rho_A^a\|\rho_A^b)\right]
\end{equation}
which means that, to leading order in $x$, the average of relative entropy between the primary and equal mixture is \emph{exactly half} of its upper bound derived by joint convexity. Generally, for any integer $D$, we can prove the following relation to leading order in $x$ \cite{Lashkari2026}
\begin{equation}
    \frac{1}{D}\sum_{a=1}^DS(\rho_A^a\|\bar{\rho}_A)=\frac{1}{2} \left[\frac{1}{D^2}\sum_{a,b=1}^DS(\rho_A^a\|\rho_A^b)\right].
\end{equation}

It follows that the power $\alpha$ computed from the replica trick is identical to the one obtained from a more naive joint convexity bound, with only the coefficient differing by a factor of $1/2$. 
Therefore, the powers $\alpha$ can actually be simply obtained from the relative entropy between two primary states in Appendix~\ref{app:Rel-a-b}. The relevant $\alpha$ exponent results for chiral semion code and chiral Ising code can be checked in~\cref{tab:power_law_and_stability_compact_boson_relent} and~\cref{tab:power_law_and_stability_Ising_relent}, respectively.

\subsubsection{Analytic results of $\beta$}\label{app:A-beta}
Next we look at the relative entropy for the complement region $\bar A$, whose angular size $1-x$ becomes large when $x\ll 1$. In this case, the joint convexity bound becomes trivial, since when $x\to 0$, we have $S(\rho_{\bar A}^a||\rho_{\bar A}^b)\to +\infty$ while $S(\rho_{\bar A}^a||\bar \rho_{\bar A})\to \log D$. So the answer from replica calculation will provide genuinely new information about the exponent $\beta$.

In the case $D=2$, we have the following results for the relative entropy in $\bar A$, to leading order in $x$, again assuming $a$ and $b$ are spinless~\cite{Lashkari2026}: 
 \begin{equation} \label{eq:rel_ent_cft_beta}
 \begin{aligned}
     \log 2 - S(\rho_{\bar A}^a\|(\rho_{\bar A}^a+\rho_{\bar A}^b)/2)=
\sum_{\mathcal{O}_q \in\mathcal L'}
c_q
(\pi x)^{2\Delta'}
+\cdots ,
 \end{aligned}
 \end{equation}
where $\mathcal{L}'$ is the set of lightest primary operators in the OPE of $\mathcal{O}_a \mathcal{O}_b^\dagger$, with scaling dimension $\Delta'$, and $c_q$ are some positive numbers, whose explicit forms will be presented in \cite{Lashkari2026}.

For compact free boson CFTs, the lightest primary operator appearing in the OPE of $\mathbf{I}$ and $V$ is $V$. Therefore we find
\begin{align}
    \log 2- S(\rho_A^{s,s}\|(\rho_A^I+\rho_A^{s,s})/2) \sim x^{1},\\
    \log 2- S(\rho_A^{I}\|(\rho_A^I+\rho_A^{s,s})/2) \sim x^{1},\\
    \log 2- S(\rho_A^{3s,s}\|(\rho_A^I+\rho_A^{3s,s})/2) \sim x^5,\\
    \log 2- S(\rho_A^{I}\|(\rho_A^I+\rho_A^{3s,s})/2) \sim x^5.
\end{align}

For Ising CFT, we find the explicit results are as follows. 
\begin{align}
    \log 2- S(\rho^I_{\bar A}||(\rho^I_{\bar A}+\rho^\sigma_{\bar A})/2) \sim x^{1/4} \\
    \log 2- S(\rho^\sigma_{\bar A}||(\rho^I_{\bar A}+\rho^\sigma_{\bar A})/2) \sim x^{1/4} \\
    \log 2- S(\rho^I_{\bar A}||(\rho^I_{\bar A}+\rho^\epsilon_{\bar A})/2) \sim x^{2} \\
    \log 2- S(\rho^\epsilon_{\bar A}||(\rho^I_{\bar A}+\rho^\epsilon_{\bar A})/2) \sim x^{2} \\
    \log 2- S(\rho^\sigma_{\bar A}||(\rho^\sigma_{\bar A}+\rho^\epsilon_{\bar A})/2) \sim x^{1/4}\\
    \log 2- S(\rho^\epsilon_{\bar A}||(\rho^\sigma_{\bar A}+\rho^\epsilon_{\bar A})/2) \sim x^{1/4}.
\end{align}
From these equations, the values of $\beta$ can be easily inferred.

\subsection{Review of the replica method}
Here we give a brief review of the replica method used to compute the relative entropy in (1+1)D CFT, as first discussed in \cite{Lashkari2014,Lashkari2015}. We then discuss the simplification of the calculation for the interval $A$ and $\bar A$ using OPE when $x\ll 1$ \cite{Sarosi2016,Sarosi2017,Lashkari2026}. The full detailed calculation will be presented in \cite{Lashkari2026}.

To compute the relative entropy $S(\rho||\sigma)$ using the replica trick, we first write it as
\begin{equation}
\label{eq:replica_trick}
    S(\rho||\sigma)= \lim_{n\to 1}\partial_n \log \frac{\tr \tilde \rho^n}{\tr \tilde \rho \, \tilde{\sigma}^{n-1}}. 
\end{equation}
where we use tildes for normalized density matrices, for example, $\tilde \rho \equiv \frac{\rho}{\tr\rho}$. It is important to make the normalization explicit, since generally $\tr \rho\neq 1$, $\tr \sigma\neq 1$ in a path-integral calculation.

We are interested in two cases: (1) $\tilde \rho = \tilde \rho_A^a$, $\tilde \sigma = \tilde \rho_A^b$, and (2) $\tilde \rho = \tilde \rho_A^a$, $\tilde \sigma = (\tilde \rho_A^a+ \tilde \rho_A^b)/2$. 

In either case, traces of (unnormalized) density matrices can be represented as path integrals on a $n$-sheeted replica manifold $\mathcal{M}_n$, as shown in Fig. \ref{fig:replica_trick}. For case (1), this is straightforward; for case (2), one needs to first expand $\tilde \rho \tilde{\sigma}^{n-1}$ into words made out of $\tilde \rho_A^a$ and $\tilde \rho_A^b$, and compute the trace of each term individually.

\begin{figure}
    \centering
    \includegraphics[width=\linewidth]{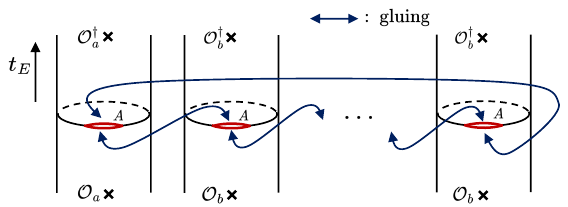}
    \caption{The Euclidean path integral for computing $\tr [\rho^a_A (\rho^b_A)^{n-1}]$. Given region $A$, the replica manifold $\mathcal{M}_n$ is obtained by gluing $n$ copies of Euclidean manifolds $\mathbb{R}\times S^1$ as indicated by the blue arrows. Using the state-operator correspondence, the path integral can be computed by a $2n$-point function with operators inserted at Euclidean time $t_E=\pm \infty$ on each sheet.}
    \label{fig:replica_trick}
\end{figure}

Using the state-operator correspondence, those path integrals can be turned into correlation functions with the corresponding operators inserted at Euclidean time $t_E=\pm \infty$ on each sheet, as also shown in Fig. \ref{fig:replica_trick}. For example,
\begin{equation}
   \tr [(\tilde \rho_A^a)^n] = \frac{\langle \Psi_{a,1}\cdots \Psi_{a,n}\rangle_{\mathcal{M}_n}}{(\langle  \Psi_{a,1}\rangle_{\mathcal{M}_1})^n},
\end{equation}
\begin{equation}
\begin{aligned}
    &\tr [\tilde{\rho}_A^a(\tilde{\rho}_A^b)^{m}(\tilde{\rho}_A^a)^{n-m-1}] \\
    =& \frac{\langle \Psi_{a,1} \Psi_{b,2} \cdots \Psi_{b,m+1} \Psi_{a,m+2} \cdots  \Psi_{a,n}\rangle_{\mathcal{M}_n}}{(\langle \Psi_{a,1}\rangle_{\mathcal{M}_1})^{n-m}(\langle \Psi_{b,1}\rangle_{\mathcal{M}_1})^{m}}.
\end{aligned}
\end{equation}
where $\Psi_{{s_k},k}=\mathcal{O}_{s_k}^\dagger(t_E=\infty_k)\mathcal{O}_{s_k}(t_E=-\infty_k)$, and $s_k=a,b$ is a pair of $a$ or $b$ operators insertion on the $k$-th sheet.

Although possible in principle, it is usually practically hard to compute those correlators on $\mathcal{M}_n$. In CFT, \cite{Lashkari2014,Lashkari2015} used a trick to conformally map $\mathcal{M}_n$ to the complex plane. The conformal map has the following feature: for each $k$, the $\mathcal{O}_{s_k}^\dagger$ and $\mathcal{O}_{s_k}$ in $\Psi_{s_k,k}$ are mapped to $z=\exp(i\pi k/n \pm  i \pi x/n)$ on the complex plane. The branch cut at the interval $A$ is mapped to $z=\exp(i\pi k/n) $, $k=1,...,n$. As shown in \cite{Lashkari2014,Lashkari2015}, the contribution from such a conformal map vanishes in the end as we take $n\to 1$. Therefore, the above correlators can be equivalently computed on the complex plane. 

This conformal transformation makes the calculation more tractable. However, even after the conformal transformation, the relative entropy can only be computed explicitly in very special cases, for example, between two primaries in free boson CFT \cite{Lashkari2014,Lashkari2015,Ruggiero2017}, or between certain primaries in Ising CFT \cite{Ruggiero2017}.

While the full correlators are generally hard to compute, we can nevertheless consider the limits $x\ll 1$ and $1-x\ll 1$, and compute the correlators perturbatively using the operator product expansion \cite{Sarosi2016,Sarosi2017,Lashkari2026}. The case $1-x\ll 1$ can be related to considering the relative entropy for the complement region $\bar A$ when $x\ll 1$. 

This simplification happens because the OPE organizes the product of two operators that are separated by a small distance in terms of a power series expansion of other operators, and the first few terms of such an expansion are often simple and universal. 

When $x\ll 1$, on the complex plane, the pair of operators within the same $\Psi_{s,k}$ approach each other, and therefore $\mathcal{O}_a^\dagger \mathcal{O}_a$, $\mathcal{O}_b^\dagger \mathcal{O}_b$ OPEs are relevant. The $\mathcal{O}_a^\dagger \mathcal{O}_a$ OPE is
\begin{equation}
\begin{aligned}
&\mathcal{O}_a(z)^\dagger \mathcal{O}_a(0) \\
=& z^{-2\Delta_a}\left(  \mathbf{I}+ \sum_{\mathcal{O}_p\in \mathcal{L}} z^{\Delta} C_{p a^* a} \mathcal{O}_p(0)+\cdots\right), 
\end{aligned}
\end{equation}
and similarly for $\mathcal{O}_b^\dagger \mathcal{O}_b$. Here the leading term is always the identity operator, and it is followed by $\mathcal{L}$, the set of lightest primary operators with non-zero OPE coefficient $C_{p a^* a}$, where $\Delta$ is their conformal dimension. When $\Delta>2$, however, the stress tensor $T$ (with $\Delta_T=2$) will be in place of the set of operators $\mathcal{O}_p$.  

 When $1-x\ll 1$, for any $k$ (mod $n$), the operator $\mathcal{O}_{s_k}$ or from $\Psi_{s_k,k}$ approaches the operator $\mathcal{O}_{s_{k+1}}^\dagger$ from $\Psi_{s_{k+1},k+1}$. Since $s_k$ can differ from $s_{k+1}$, $\mathcal{O}_a^\dagger \mathcal{O}_b$, $\mathcal{O}_b^\dagger \mathcal{O}_a$ OPE also become relevant. The $\mathcal{O}_a^\dagger \mathcal{O}_b$ OPE is
\begin{equation}
\begin{aligned}
&\mathcal{O}_a(z)^\dagger \mathcal{O}_b(0) \\
=& z^{-(\Delta_a+\Delta_b)}\left(\sum_{\mathcal{O}_p\in \mathcal{L}'} z^{\Delta'} C_{p a^* b} \mathcal{O}_p(0)+\cdots\right), 
\end{aligned}
\end{equation}
and similarly for $\mathcal{O}_b^\dagger \mathcal{O}_a$. Here the leading term involves $\mathcal{L}'$, the set of lightest primary operators with non-zero OPE coefficient $C_{p a^* b}$, where $\Delta'$ is their conformal dimension. 

The above OPEs will turn $2n$-point functions into sums over $n$-point functions, which turn out to be further largely simplified and computed analytically due to the appearance of the identity operator in certain OPEs. With those $2n$-point functions in hand, we can then assemble them into Eq. $\eqref{eq:replica_trick}$ (which can include the sum over words), and then analytically continue to $n\to 1$. Finally, this leads us to the results presented before, i.e., Eqs.~\eqref{eq:rel_ent_cft}, \eqref{eq:rel_ent_cft_2}, \eqref{eq:rel_ent_cft_alpha}, \eqref{eq:rel_ent_cft_beta}. The detailed calculation, in particular the calculation for the case $1-x \ll 1$ and the case involving the probabilistic mixture, will be presented in \cite{Lashkari2026}.

\section{Chiral semion model by lattice bosonic Laughlin states}
\label{app:FQH_wavefunction} 

\subsection{Analytical wave functions}
\label{app:FQH_analytical_wavefunction}

The chiral semion topological order can be realized by the bosonic fractional quantum Hall (FQH) with $\nu = 1/2$. Ref.~\cite{Nielson2014,Nielsen2012} provides a wider class of FQH states, with $\nu = 1/q$ in which $q$ is an even integer. (Similar models have been generalized to non-Abelian topological orders, see e.g.~\cite{Manna2018,Liu2025}.) Here, we discuss the wave functions that correspond to cylinder primary. For general $q$,
let the ``charge number'' of the two edges to be $p_1$ and $p_2$, where $p_j \in \{ 1,2,\cdots, q-1\}$ (module $q$). Physically, the charges of the edges are $(p_1/q , p_2/q)$.
The wave function of any even $q$ is given by 
\begin{equation}
\begin{aligned}\label{eq:wave-with-w12-p12-q}
    \Psi^{[p,w,z]}_q(n_1, \cdots, n_N) = \frac{1}{C} \, \delta_{\mathbf{n}}^{[p,q]} \cdot \prod_{i < j} (z_i - z_j)^{q n_i n_j} \\\cdot \prod_{i\ne j} (z_i - z_j)^{- n_i} 
    \cdot \prod_{i,j} (w_i - z_j)^{p_i n_j}  
\end{aligned}
\end{equation}
where $n_j \in \{ 0,1 \}$ and the delta function is 
\begin{equation}\label{eq:delta-2}
    \delta_{\mathbf{n}}^{[p,q]} = \left\{ \begin{array}{cc}
        1, & \text{for } \sum_j n_j = \frac{ N - p_1 - p_2}{ q}  \\
         0, & \text{otherwise.} 
    \end{array} \right.
\end{equation}
The values $\{ z_j\}_{j=1}^N$ are complex numbers chosen depending on the geometric shape of the system, and $w_1, w_2$ are two complex numbers which reflect the locations of topological excitations~\cite{Nielson2014}. 
Importantly, the system size $N$ (i.e., the total number of qubits) must be such that $N\mod q =0$. The sum of two changes should be $p_1 + p_2 = k q$ for the wave function to represent a primary state. 
 
\subsection{Cylinder with a square lattice}
\label{app:square_lattice}

For the cylinder used in Fig.~\ref{fig:edge_gamma}, we choose the lattice of \(N=L_xL_y\) sites as follows. For integers \(L_x\) and \(L_y\), sites are labeled by a row index \(j_2=0,\ldots,L_y-1\) and a periodic coordinate \(j_1=0,\ldots,L_x-1\), with \(j_1\sim j_1+L_x\). The corresponding row-major site label is \(j=1+j_2L_x+j_1\). The complex plane coordinates $\{z_j\}$ are generated from the cylinder coordinate  
\begin{equation}
    u_1(j_2)=\frac{2\pi}{L_x}j_2,\quad
    u_2(j_1)=\frac{2\pi}{L_x}j_1 .
\end{equation}
Thus 
\begin{equation}
\begin{aligned}
 &z_j:=z_{j_1,j_2}
    =\exp \left(u_1(j_2)+i u_2(j_1)\right)
\\& =\exp \left(\frac{2\pi}{L_x}j_2
    +i\frac{2\pi}{L_x}j_1\right).  
\end{aligned}
\end{equation}
For each fixed row \(j_2\), the \(L_x\) sites are equally spaced around the periodic direction, while increasing \(j_2\) moves outward along the cylinder before the exponential map.

The values of $(p_1, p_2)$ associated with the $\nu = 1/2$ bosonic Laughlin cylinder primary states we use are: 
\begin{equation}\label{eq:cylinder-primary-boson}
    \begin{aligned}
       \ket {\psi_1^{0,0}} & \to (0,0), \\
        \ket{\psi_s^{\frac12,\frac12}} & \to (1,1), \\
       \ket{\psi_s^{\frac12,-\frac12}} & \to (1,-1), \\
       \ket{\psi_s^{\frac32,\frac12}} & \to (3,1) ,
    \end{aligned}
\end{equation}
written in the notation in Example~\ref{exmp:chiral_Semion}. We set \(w_1=0\) and \(w_2=10^8\).

\subsection{Spin chains}
\label{app:semion_chain}

We can use this analytical wave function to simulate the 1D compact free boson RCFT with compactification radius $R=\sqrt{2}$ and central charge $c=1$. To do so, we choose $q=2$. For the primary state $\ket{\varphi_s^{\frac{1}{2},\frac{1}{2}}}$ which corresponds to the vertex operator 
\begin{equation}
    V(z,\bar{z})=:\exp(\frac{i}{\sqrt{2}}\phi(z)+\frac{i}{\sqrt{2}}\bar{\phi}(\bar{z})):,
\end{equation}
we set $w_1=0$ and $w_2=10^8$. The qubits are arranged evenly on the circle with unit radius. In other words, we let 
\begin{equation}
    z_j=\exp(\frac{2\pi i}{N}j),\quad  \text{for}\quad  j=1,2,\cdots N.
\end{equation}
Again, we set \(w_1=0\) and \(w_2=10^8\).
The primary states of 1D RCFT corresponding to those listed in Eq.~\eqref{eq:cylinder-primary-boson} under dimensional reduction, are
\begin{equation}
    \ket {\varphi_1^{0,0}} , \quad 
    \ket {\varphi_s^{\frac12,\frac12}} , \quad 
    \ket {\varphi_s^{\frac12,-\frac12}} , \quad 
    \ket {\varphi_s^{\frac32,\frac12}}.
\end{equation}
They share the same assignment of $(p_1, p_2)$.

\section{Ising CFT primary: puMPS method}
\label{app:ising_pumps_method}

This appendix explains the periodic uniform MPS (puMPS) method~\cite{Zou2017} used for the one-dimensional
Ising CFT data in the numerical section~\ref{subsec:numerics_ising_beta_relent}. The goal is to obtain finite-size
lattice representatives of the three low-lying periodic-sector primaries
\(\ket{I}\), \(\ket{\sigma}\), and \(\ket{\epsilon}\), and then to compute reduced
density matrices on a small interval \(M=\{1,\ldots,\ell\}\). The notations used here are mainly based on the ones in~\cite{Zou2017}.

\subsection{Hamiltonian and periodic uniform MPS ansatz}
\label{app:ising_pumps_ham_ansatz}

The critical transverse-field Ising chain used in the simulation is
\begin{equation}
\label{eq:ising_hamiltonian_pumps_app}
    H_N=-\sum_{j=1}^{N}X_jX_{j+1}-\sum_{j=1}^{N}Z_j,
\end{equation}
The ground-state variational ansatz is the periodic uniform matrix product state
(puMPS). A puMPS with physical dimension \(d=2\) and bond dimension
\(D\) is specified by one tensor \(A^s\in\mathbb{C}^{D\times D}\), repeated at
every site:
\begin{equation}
\label{eq:pumps_ground_ansatz_app}
    \ket{\Psi(A)}
    =
    \sum_{s_1,\ldots,s_N=1}^{d}
    \operatorname{Tr} \big[A^{s_1}A^{s_2}\cdots A^{s_N}\big]
    \ket{s_1s_2\cdots s_N} .
\end{equation}
The trace enforces the periodic boundary condition. The state is invariant
under the gauge transformation \(A^s\mapsto G^{-1}A^sG\), with \(G\) invertible,
and the implementation fixes this freedom by repeatedly bringing the tensor to a
left-canonical gauge, which requires  
\begin{equation}
\label{eq:left_canonical_app}
    \sum_s A_L^{s\dagger}A_L^s=\mathbf{1},
    \qquad
    \sum_s A_L^s\lambda^2 A_L^{s\dagger}=\lambda^2,
\end{equation}
where \(\lambda\) is the diagonal matrix of Schmidt coefficients of the
corresponding infinite uniform MPS. The center tensor is
\begin{equation}
\label{eq:Ac_def_app}
    A_C^s=A_L^s\lambda .
\end{equation}

The ground state is obtained by minimizing
\begin{equation}
\label{eq:pumps_ground_energy_app}
    E(A,\bar A)=
    \frac{\bra{\Psi(A)}H_N\ket{\Psi(A)}}{\braket{\Psi(A)|\Psi(A)}}.
\end{equation}
In practice, the program first applies a small number of VUMPS (variational uniform matrix product state)~\cite{2018PhRvB..97d5145Z} iterations and then
uses local energy minimization. The local minimization uses the deformed puMPS
with only the first tensor varied,
\begin{equation}
\label{eq:local_deformed_pumps_app}
    \ket{\Psi_{A_L}(A_C)}
    =
    \sum_{\mathbf{s}}
    \operatorname{Tr} \big[(A_C^{s_1}\lambda^{-1})A_L^{s_2}\cdots A_L^{s_N}\big]
    \ket{\mathbf{s}},
\end{equation}
and the induced local norm matrix
\begin{equation}
\label{eq:local_metric_app}
    \braket{\Psi_{A_L}(A_C)|\Psi_{A_L}(A_C)}
    =
    \bar A_C^{\mu}g_{\mu\nu}A_C^{\nu},
    \quad \mu=(s,a,b).
\end{equation}
The physical gradient direction is therefore 
\begin{equation}
\label{eq:pumps_gradient_direction_app}
    \Delta A_C^{\mu}
    =
    -g^{\mu\nu}
    \frac{\partial E_{A_L}(A_C,\bar A_C)}{\partial \bar A_C^{\nu}},
    \quad g^{\mu\nu}g_{\nu\rho}=\delta^\mu_{\rho},
\end{equation}
where \(E_{A_L}\) is the auxiliary energy functional obtained by replacing
\(\ket{\Psi(A)}\) in Eq.~\eqref{eq:pumps_ground_energy_app} with
\(\ket{\Psi_{A_L}(A_C)}\). The uniform tensor is updated by
\begin{equation}
\label{eq:pumps_gradient_update_app}
    A_L^s \longleftarrow A_L^s+\tau\,\Delta A_C^s\lambda^{-1},
\end{equation}
where \(\tau\) is chosen by a line search, followed by canonicalization and
normalization. This gives the optimized tensor \(A\) used below. The resulting
normalized state is identified with the finite-size representative of the Ising
CFT vacuum,
\begin{equation}
\label{eq:I_state_app}
    \ket{I}_N\equiv \frac{\ket{\Psi(A)}}{\sqrt{Z_I}},
    \quad
    Z_I=\braket{\Psi(A)|\Psi(A)}.
\end{equation}

\subsection{Construction of excited states \texorpdfstring{\(\ket{\sigma}\)}{sigma} and \texorpdfstring{\(\ket{\epsilon}\)}{epsilon}}
\label{app:ising_pumps_tangent_states}

Low-energy excited states are represented by Bloch-state tangent vectors built
on the optimized ground-state tensor \(A\). For momentum \(p=2\pi k/N\), define
\begin{equation}
\label{eq:pumps_tangent_ansatz_app}
\begin{aligned}
&\ket{\Phi_p(B;A)}
    \\&=
    \sum_{n=1}^{N} e^{ipn}
    \sum_{\mathbf{s}}
    \operatorname{Tr} \left[
        A^{s_1}\cdots A^{s_{n-1}}B^{s_n}
        A^{s_{n+1}}\cdots A^{s_N}
    \right]
    \ket{\mathbf{s}} . 
\end{aligned}
\end{equation}
The sign of \(p\) is a convention fixed by the translation operator; all states
we used here have \(k=0\). The tensor \(B\) has the same index structure as \(A\). It is useful
to solve the variational problem in center gauge,
\begin{equation}
\label{eq:BC_center_gauge_app}
    B^s=B_C^s\lambda^{-1}.
\end{equation}
Let \(B^\mu_C\) be the vectorized variational parameters. Projecting the
Hamiltonian into the tangent subspace gives the generalized eigenvalue problem
\begin{equation}
\label{eq:tangent_generalized_evp_app}
    H^{\mathrm{eff}}_{\mu\nu}(p)B_C^\nu
    =
    E\,G_{\mu\nu}(p)B_C^\nu,
\end{equation}
with
\begin{align}
\label{eq:tangent_metric_hamiltonian_app}
    G_{\mu\nu}(p)
    &=
    \left\langle
      \frac{\partial \Phi_p(\bar B_C;\bar A)}{\partial \bar B_C^\mu}
      \middle|
      \frac{\partial \Phi_p(B_C;A)}{\partial B_C^\nu}
    \right\rangle,\\
    H^{\mathrm{eff}}_{\mu\nu}(p)
    &=
    \left\langle
      \frac{\partial \Phi_p(\bar B_C;\bar A)}{\partial \bar B_C^\mu}
      \middle|H_N\middle|
      \frac{\partial \Phi_p(B_C;A)}{\partial B_C^\nu}
    \right\rangle .
\end{align}
The metric \(G(p)\) is positive semidefinite, not strictly positive definite,
because the tangent representation has gauge redundancies. Numerically one
uses the pseudoinverse \(\widetilde G(p)\) and solves
\begin{equation}
\label{eq:tangent_pinv_evp_app}
    \widetilde G^{\rho\mu}(p)H^{\mathrm{eff}}_{\mu\nu}(p)B_C^\nu
    =
    E B_C^\rho,
\end{equation}
normalizing the eigenvectors by
\begin{equation}
\label{eq:tangent_norm_app}
    \bar B_C^\mu G_{\mu\nu}(p)B_C^\nu=1.
\end{equation}
After the eigenvalue problem is solved, one converts back to \(B^s=B_C^s\lambda^{-1}\)
and stores the excited state as a pair \((A,B)\).

For the PBC Ising data in this paper, we solve only the momentum-zero sector, ($k=0,\,\, p=0$), and derive the three lowest tangent-space eigenvectors. The optimized puMPS
\(\ket{\Psi(A)}\) itself is used as the vacuum representative. The two lowest
nontrivial tangent-space states in the same sector are then assigned as
\begin{equation}
\label{eq:sigma_epsilon_tensors_app}
    \ket{\sigma}_N
    \equiv
    \frac{\ket{\Phi_0(B_\sigma;A)}}{\sqrt{Z_\sigma}},
    \qquad
    \ket{\epsilon}_N
    \equiv
    \frac{\ket{\Phi_0(B_\epsilon;A)}}{\sqrt{Z_\epsilon}} .
\end{equation} 

\begin{algorithm}[t]
\DontPrintSemicolon
\caption{PBC Ising primary-state construction by puMPS}
\label{alg:pumps_ising_primary_app}
\KwIn{System size \(N\), bond dimension \(D\).}
Construct the local Ising MPO for Eq.~\eqref{eq:ising_hamiltonian_pumps_app} and the split PBC MPO \(H_{\rm OBC}+H_{\partial}\).\;
Initialize a random translation-invariant tensor \(A\in\mathbb{C}^{D\times d\times D}\).\;
Optimize \(A\) by VUMPS preconditioning followed by local gradient minimization of Eq.~\eqref{eq:pumps_ground_energy_app}.\;
Canonicalize \(A\) and compute \(\lambda\).\;
Build \(G(0)\) and \(H^{\rm eff}(0)\) from Eq.~\eqref{eq:tangent_metric_hamiltonian_app}.\;
Solve the pseudoinverse eigenproblem Eq.~\eqref{eq:tangent_pinv_evp_app} for the lowest tangent vectors.\;
Use \(\ket{\Psi(A)}\) as \(\ket{I}\); use the two lowest nontrivial tangent states as \(\ket{\sigma}\) and \(\ket{\epsilon}\).\; 
\end{algorithm}

Here, we choose the system size: $N=16,18,20,22,24$ with bond dimension $D = 14,16,16,18,18 $. 

\subsection{Reduced density matrices from puMPS contractions}
\label{app:pumps_rdm_contraction}

The reduced density matrix of subsystem $M$ can be obtained directly from puMPS transfer matrices, avoiding
the dense vector of size \(d^N\). For two local tensors \(X^s,Y^s\in\mathbb{C}^{D\times D}\), define the
single-site double-layer transfer matrices
\begin{equation}
\label{eq:double_layer_def_app}
    E_{X,Y}^{s,t}:=X^s\otimes \overline{Y^t},
    \qquad
    E_{X,Y}:=\sum_{u=1}^{d}E_{X,Y}^{u,u} .
\end{equation}
For the ground state, \(X=Y=A\), the norm is
\begin{equation}
\label{eq:ground_norm_transfer_app}
    Z_I=\operatorname{Tr}((E_{A,A})^{N}).
\end{equation}
For \(M=\{1,\ldots,\ell\}\) and multi-indices
\(\mathbf{s}=(s_1,\ldots,s_\ell)\), \(\mathbf{t}=(t_1,\ldots,t_\ell)\), the reduced
density matrix of the ground state is
\begin{equation}
\label{eq:ground_rdm_transfer_app}
    \big(\rho_{M}^{I}\big)_{\mathbf{s},\mathbf{t}}
    =
    \frac{1}{Z_I}
    \operatorname{Tr} \left[
        E_{A,A}^{s_1,t_1}E_{A,A}^{s_2,t_2}\cdots
        E_{A,A}^{s_\ell,t_\ell}
        (E_{A,A})^{N-\ell}
    \right].
\end{equation}
The transfer matrix's power
\((E_{A,A})^{N-\ell}\) contracts the complement \(\bar M\), while the
\(\ell\) block tensors keep the physical indices \(\mathbf{s},\mathbf{t}\) open.

The tangent-state reduced density matrices are obtained by the same rule, but
one must sum over the insertion positions of the ket tensor \(B\) and the bra
tensor \(\bar B\). It is convenient to write one formula that covers ground and
tangent states. Let \(\alpha\in\{I,\sigma,\epsilon\}\). For \(\alpha=I\) there is a
single insertion configuration \(r=\varnothing\), with
\(X_j^{(I,r)}=A\) for all \(j\) and phase \(\omega_I(r)=1\). For a tangent state
\(\alpha\in\{\sigma,\epsilon\}\), the configurations are
\(r=n\in\{1,\ldots,N\}\), with
\begin{equation}
\label{eq:tangent_config_tensor_app}
    X_j^{(\alpha,n)}=
    \begin{cases}
    B_\alpha, & j=n,\\
    A, & j\neq n,
    \end{cases}
    \qquad
    \omega_\alpha(n)=e^{ip_\alpha n} .
\end{equation}
For the states used in this work \(p_\sigma=p_\epsilon=0\). Define
\begin{equation}
\label{eq:generic_state_norm_app}
    Z_\alpha
    =
    \sum_{r,r'\in\mathcal{C}_\alpha}
    \omega_\alpha(r)\overline{\omega_\alpha(r')}
    \operatorname{Tr} \left[
      \prod_{j=1}^{N}
      E_{X_j^{(\alpha,r)},X_j^{(\alpha,r')}}
    \right],
\end{equation}
where \(\mathcal{C}_I=\{\varnothing\}\) and
\(\mathcal{C}_{\sigma}=\mathcal{C}_{\epsilon}=\{1,\ldots,N\}\). Then the mixed
reduced operator
\begin{equation}
\label{eq:mixed_rdm_def_app}
    \rho_{M}^{\alpha\beta}
    :=
    \operatorname{Tr}_{\bar M}\ket{\psi_\alpha}\bra{\psi_\beta}
\end{equation}
is
\begin{align}
\label{eq:generic_mixed_rdm_transfer_app}
    \big(\rho_{M}^{\alpha\beta}\big)_{\mathbf{s},\mathbf{t}}
    & =
    \frac{1}{\sqrt{Z_\alpha Z_\beta}}
    \sum_{r\in\mathcal{C}_\alpha}
    \sum_{r'\in\mathcal{C}_\beta}
    \omega_\alpha(r)\overline{\omega_\beta(r')}
    \nonumber\\
    &\quad \times
    \operatorname{Tr} \left[
      \prod_{j=1}^{\ell}
      E_{X_j^{(\alpha,r)},X_j^{(\beta,r')}}^{s_j,t_j}
      \prod_{j=\ell+1}^{N}
      E_{X_j^{(\alpha,r)},X_j^{(\beta,r')}}
    \right].
\end{align}
Setting \(\alpha=\beta\) gives the ordinary reduced density matrices
\(\rho_{M}^{I}\), \(\rho_{M}^{\sigma}\), and \(\rho_{M}^{\epsilon}\). Setting
\(\alpha\neq\beta\) gives the off-diagonal reduced operators needed for a direct
puMPS contraction of the reference-system mutual information. Equation
\eqref{eq:generic_mixed_rdm_transfer_app} is also the precise tensor-network
meaning of tracing out the complement of a tangent-state puMPS.

For the computation of the ground-state reduced density matrix on the left-hand side of Eq.~\eqref{eq:ground_rdm_transfer_app}, the RAM needed is \(d^{2\ell}\) instead of \(d^N\), which is needed if we start with the pure state on the entire system. For tangent states, the expression above has an explicit double sum over the two insertion positions. A naive implementation scales as \(O(N^2 d^{2\ell})\) transfer contractions, but the products can be reused with prefix and suffix transfer matrices. For fixed small \(\ell\), this avoids the dense-vector bottleneck and is the contraction strategy needed to push the Ising scan to larger system
sizes.

\section{Power-law-range recovery channel details}
\label{app:quasi_local_recovery_channel}

\subsection{Tools for the proof}

We present a set of tools and necessary background knowledge for the proof of~\cref{thm:quasi-local-recovery}.

\subsubsection{A universal recovery theorem}
There is a universal recovery theorem following from Eq.~(15) in Ref.~\cite{JRSWW-universal-recovery}. 
For a pair of states $\rho, \sigma$ such that \(\operatorname{supp}(\rho)\subseteq\operatorname{supp}(\sigma)\), the relative entropy under a quantum channel $\mathcal{N}$ satisfies
\begin{equation}
\begin{aligned}
&S(\rho||\sigma) - S(\mathcal N(\rho)||\mathcal N(\sigma))
\\&\ge-\log {F} \left(\rho,\mathcal R\circ \mathcal N(\rho)
\right),  
\end{aligned}
\end{equation} 
where $\mathcal{R}$ denotes a universal recovery channel that only depends on $\mathcal{N}$ and $\sigma$. 
Here, we use the Uhlmann fidelity $F(\omega,\tau):=\|\sqrt{\omega}\sqrt{\tau}\|_1^2$. Furthermore, the recovery map $\calR$ has an explicit form as worked out in Ref.~\cite{JRSWW-universal-recovery}.  

We shall need the special context of tripartite states $\rho_{ABC}$ and $\sigma_{ABC}$ living on a tensor product of three finite-dimensional Hilbert spaces $\calH_A \otimes \calH_B \otimes \calH_C$ with a channel $\calN_A$ acting on $A$. Furthermore, we require a special form of $\sigma$ as
\begin{equation}
    \sigma_{ABC}=\sigma_{AB}\otimes\sigma_C,
\end{equation}
and with $\rho$ still satisfies \(\operatorname{supp}(\rho)\subseteq\operatorname{supp}(\sigma)\). In this context, the universal recovery theorem implies the following inequality: 
\begin{equation}
\begin{aligned}
&S(\rho_{ABC}\|\sigma_{ABC})
-
S(\calN_A(\rho_{ABC})\| \calN_A(\sigma_{ABC}))
\\
&\ge
-\log F\left(
\rho_{ABC},
(\mathcal R_{AB}\otimes\mathrm{id}_C)\circ\calN_A(\rho_{ABC})
\right),
\end{aligned}
\label{eq:recovery_theorem-ABC}
\end{equation}
where the recovery channel $\calR_{AB}$ depends only on $\sigma_{AB}$ and $\calN_A$. The verification of the support of $\calR$ on $AB$ is not completely trivial, and it follows from the explicit form of the recovery channel in Ref.~\cite{JRSWW-universal-recovery}.

\subsubsection{Full boundary axioms from entanglement bootstrap}
 
We present two conditions for partitions covering an edge. They are useful in understanding cylinder primary states and especially Abelian ones. These conditions are analogs of entanglement bootstrap axioms {\bf A0} and {\bf A1} of the topologically ordered bulk~\cite{Shi2019fusion}, and are considered in Ref.~\cite{Chiral-vira2024}.
\begin{itemize}[leftmargin=13pt]
    \item Full boundary {\bf A0} refers to
    \begin{equation}
        S_{C} + S_{BCD} - S_{BD}  \approx 0, 
    \end{equation}
    for partition $B,C,D$ on cylinder that is topologically the same as Fig.~\ref{fig:full-bdy-A1}. Such a condition is expected to be satisfied for any cylinder primary state $|\psi^\mfa\rangle$. However, it may not hold for superpositions of different cylinder primary states.
     \item Full boundary {\bf A1} refers to
    \begin{equation}
       S_{BC} + S_{CD} - S_{B} - S_D  \approx 0,  
    \end{equation}
    for partition $B,C,D$ on cylinder that is topologically the same as Fig.~\ref{fig:full-bdy-A1}. Such a statement is expected to hold only for Abelian primary states, namely it holds for $|\psi^\mfa\rangle$ with an Abelian $\mfa \in \calA$. 
\end{itemize}
 Suppose that the full boundary axioms are satisfied on an edge. Then collapsing the edge to a point will result in a point satisfying the bulk axioms. In other words, the edge is invisible unless we cut it into pieces. We should expect non-vanishing errors for systems with finite onsite Hilbert spaces for chiral states~\cite{Li2025strict,Ranard2024strict}, though such errors typically decay fast towards zero as the sizes of the regions increase. We shall neglect the errors for the applications in this appendix.
 
\begin{figure}[h]
    \centering
    \includegraphics[width=0.45\linewidth]{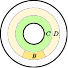}
    \caption{{\bf Full boundary Axioms.} Partition $BCD$ of an annulus surrounding an edge, which we use to define the full boundary entanglement bootstrap axioms.} 
    \label{fig:full-bdy-A1}
\end{figure}

\subsubsection{Conditional mutual information of chiral edge states}

For an Abelian cylinder primary, we apply full boundary $\mathbf{A1}$ to relate the conditional mutual information of a certain partition of the cylinder to the mutual information between disjoint intervals $A'',C''$ that can be calculated in the corresponding 1D CFT. 
\begin{Proposition}
\label{prop:mutual_on_chiral}
Consider a code word $\rho^\mfa$ of a chiral edge code, and the partitions in Fig.~\ref{fig:proof}. Suppose $\mfa$ is Abelian, we have
\begin{equation}
I(A:C|B)_{\rho^\mfa}=I(A:C')_{\rho^\mfa} \le I(A'':C'')_{\rho^\mfa}.
\end{equation}

\end{Proposition}
\begin{proof}
The equality $I(A:C|B)_{\rho^\mfa}=I(A:C')_{\rho^\mfa}$ follows from the full boundary version of entanglement bootstrap axiom $\mathbf{A1}$ which we explained above. The detailed computation, with regions shown Fig.~\ref{fig:proof} (b) and (c), is 
    \begin{equation}
\begin{aligned}
  & I(A:C_1C_2C')_{\rho^{\mfa}}-I(A:C')_{\rho^\mfa}\\
  &=(S_{C_1C_2C'}-S_{AC_1C_2C'}-S_{C'} + S_{AC'})_{\rho^\mfa}\\
  &=I(C_1C_2:A|C')_{\rho^\mfa}\\
  & = (S_{C_1C_2C'}-S_{B}-S_{C'} + S_{B C_1 C_2})_{\rho^\mfa}\\
&:=\Delta(B,C_1C_2,C')_{\rho^\mfa}.
\end{aligned}
\end{equation} 
Then, we apply the full boundary {\bf A1} and notice that shrinking the regions in the 1st and 3rd slots of $\Delta$ cannot decrease the value, namely $\Delta(XX',Y,ZZ') \le \Delta(X,Y,Z)$. Thus,
\begin{equation}\label{eq:A1-bound}
    \Delta(B,C_1C_2,C')_{\rho^\mfa} \le \Delta(\tilde{B},C_1C_2,\tilde{C}')_{\rho^\mfa}
\end{equation}
where $\tilde{B}$ and $\tilde{C}'$ not shown in the figure, are regions with $B$ and $C'$ minus a thin layer adjacent to the edge. Importantly, $(\tilde{B},C_1C_2,\tilde{C}')$ is now topologically identical to the partition $(B,C,D)$ in Fig.~\ref{fig:full-bdy-A1}, thus the right-hand side of Eq.~\eqref{eq:A1-bound} is zero by the full boundary {\bf A1}, so is the left-hand side. 
Therefore,
\begin{equation}
I(A:C|B)_{\rho^\mfa}\;=\;I(A:C')_{\rho^\mfa}, 
\end{equation}
for partitions in Fig.~\ref{fig:proof}.
Note that we prove this only for Abelian states because the full boundary {\bf A1} should be violated by non-Abelian anyonic charges.

The inequality that $I(A:C')_{\rho^\mfa} \le I(A'':C'')_{\rho^\mfa}$ follows from the fact that $A \subset A''$ and $C'\subset C''$ and the monotonicity of mutual information under partial trace. This completes the proof. 
\end{proof}

\begin{figure}[h]
    \centering
    \includegraphics[width=0.85\linewidth]{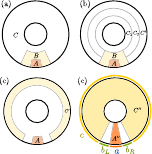}
    \caption{Partitions of the cylinder viewed as a flat disk, useful for the proof of theorem~\ref{thm:quasi-local-recovery}. (a) A local disk $A$ adjacent to an edge. $B$ separates $A$ from the rest. (b) Further partition $C$ into $C_1 C_2$ and $C'$. (c) $A'$ and $C'$ are both disks adjacent to the same edge. (d) Columns $A'' \supset A$ and $C'' \supset C'$ are introduced for dimensional reduction purposes. Each of $a,b_L, b_R$ and $c$ is the 1D correspondence of a certain column.   
    }
    \label{fig:proof}
\end{figure}

\subsection{Proof of the theorem~\ref{thm:quasi-local-recovery}}

Let $A$ be the erased local disk near the edge and let $B$ be a buffer region surrounding $A$. Define $C := Q \setminus (AB)$; see Fig.~\ref{fig:proof}(a). We denote the channel on $A$ as $\calN_A$. 
Our first trick is to apply the variant of the universal recovery theorem from~\eqref{eq:recovery_theorem-ABC}, choosing  
\begin{equation}
\sigma_{ABC}:=\pi_{AB}\otimes \pi_C,
\end{equation}
where $\pi_{X}
:=
\frac1D
\sum_{\mfa\in A_{\mathrm{code}}}
\rho^\mfa_{X}$
is the maximally mixed code state supported on a subsystem $X$. We obtain
\begin{equation} 
\begin{aligned}
&  S(\rho_{ABC}||\sigma_{ABC}) - S(\mathcal \Tr_A(\rho_{ABC})||\mathcal \Tr_A(\sigma_{ABC})) \\ 
&\ge-\log F\left(\rho_{ABC},\mathcal \calE_{B\to AB}\circ \mathcal \Tr_A(\rho_{ABC})
\right), 
\end{aligned}
\label{eq:recovery_theorem}
\end{equation}
where $\calE_{AB}$ is a universal twirled Petz map [Eq.(21) of \cite{JRSWW-universal-recovery}], which depends only on $\pi_{AB}$.

Therefore, it suffices to upper bound the relative-entropy difference 
$S(\rho_{ABC}||\sigma_{ABC}) - S(\mathcal \rho_{BC}||\mathcal \sigma_{BC})$.

Recall that we consider an arbitrary density matrix in the code subspace, which is a convex combination of pure states in the code subspace of the form $ |\Psi\rangle=\sum_\mfa c_\mfa |\psi^\mfa\rangle$. Since $C$ contains a noncontractible annulus, the anyon charge can be measured from $C$. Consequently, local reduced density matrices on $AB$ are classical mixtures of states from different anyon sectors because of the orthogonality:
$ \rho_{AB}
=\sum_\mfa p_\mfa\, \rho^\mfa_{AB}$, in which $p_\mfa :=|c_\mfa|^2$ for $|\Psi\rangle$. For a mixed state in the code subspace formed by a convex combination, the probabilities $p_\mfa$ are determined similarly by an average of probabilities from each pure state.  

Now, we calculate the relative-entropy difference  
\begin{equation}
\begin{aligned}
&
S(\rho_{ABC}||\pi_{AB}\otimes\pi_C)
-
S(\rho_{BC}||\pi_B\otimes\pi_C)
\\
&=
I(A:C|B)_\rho
+
S(\rho_{AB}||\pi_{AB})
-S(\rho_B||\pi_B)\\
& \le 
I(A:C|B)_\rho
+
S(\rho_{AB}||\pi_{AB}).
\label{eq:DeltaDecomp}
\end{aligned}
\end{equation}
The second line follows from simple algebra, and the third line follows from the non-negativity of relative entropy.

We now bound the two terms separately.
First, by joint convexity of relative entropy, we have
\begin{equation} 
\begin{aligned}
S(\rho_{AB}\|\pi_{AB})
&=
S\left(
\left.\sum_\mfa p_\mfa \rho^\mfa_{AB}
\right\|
\frac1D\sum_\mfb \rho^\mfb_{AB}
\right)
\\
&\le
\sum_\mfa p_\mfa
S\left(
\left.\rho^\mfa_{AB} \right\| \frac1D\sum_\mfb \rho^\mfb_{AB}
\right)
\\
&\le
\frac1D
\sum_{\mfa,\mfb}
p_\mfa
S(\rho^\mfa_{AB}||\rho^\mfb_{AB}).
\label{eq:jointconv}
\end{aligned}
\end{equation}
The relative entropies $S(\rho^\mfa_{AB}||\rho^\mfb_{AB})\sim x^{\gamma_{\mfa \mfb}}$, where $x$ is the angular size of the arc associated with $AB$. 
Plugging this in Eq.~\eqref{eq:jointconv}, we obtain
\begin{equation}
S(\rho_{AB}\|\pi_{AB})\leq c_5 x^{\gamma^*},
\end{equation}
at small $x$, where $c_5$ is some constant and
\begin{equation}
    \gamma^* := 
    \min_{\substack{\mfa,\mfb \in \calA_{\rm code}, \\ \mfa\ne \mfb}}
    \gamma_{\mfa \mfb}.
\end{equation}
Taking $l_A=O(1)$ and $l_B\sim L_x^\lambda$ with
$0<\lambda<1$, we have
$l_{AB}\sim L_x^\lambda$,
and therefore
\begin{equation}\label{eq:power-A}
    S(\rho_{AB}||\pi_{AB})
 \le O(1)\
L_x^{-(1-\lambda)\gamma^*}.
\end{equation}

Next we estimate the conditional mutual information $I(A:C|B)_\rho$.
It is convenient to use the Holevo leakage
\begin{equation}
\begin{aligned}
\chi_Y(p)
&:=
S\left(
\sum_\mfa p_\mfa \rho^\mfa_Y
\right)
-
 \sum_\mfa p_\mfa S(\rho^\mfa_Y)
\\&=\sum_\mfa
p_\mfa
S(\rho^\mfa_Y||\bar\rho^p_Y),
\end{aligned}
\end{equation}
where $\bar\rho^p_Y =
\sum_\mfa p_\mfa \rho^\mfa_Y. $
After direct algebra, we have
\begin{equation}
    \begin{aligned}
I(A:C|B)_\rho \le & (S_{A} + S_{AB} - S_B)_\rho \\
= & (S_{A} + S_{AB} - S_B)_{\bar{\rho}^p} \\
= &\sum_\mfa p_\mfa
(S_{A} + S_{AB} - S_B)_{\rho^\mfa} \\
&+
\chi_A(p)
+
\chi_{AB}(p)
-\chi_B(p). \\
= &\sum_\mfa p_\mfa
I(A:C|B)_{\rho^\mfa} \\
&+
\chi_A(p)
+
\chi_{AB}(p)
-\chi_B(p) \\
\le
&\sum_\mfa p_\mfa
I(A:C|B)_{\rho^\mfa}\\
&+
\chi_A(p)
+
\chi_{AB}(p).
\end{aligned}
\end{equation}
The first line follows from the strong subadditivity. The second line follows from the orthogonality between different anyon sectors on $C$. The third line follows from the definition of Holevo leakage applied to regions $A, B$ and $AB$ respectively. The third equality follows from the purity of code state $\rho^\mfa$. The last inequality follows by 
dropping the nonnegative term $\chi_B(p)$. 

Next we bound each term on the right-hand side. The computation of $I(A:C|B)_{\rho^\mfa}$ is boiled down to a disjoint interval mutual information on 1D CFT, as implied by
Prop.~\ref{prop:mutual_on_chiral}. In detail,
\begin{equation}
\begin{aligned}
    I(A:C|B)_{\rho^\mfa}
\le & I(A'':C'')_{\rho^\mfa}  \\
=& I(a:c)_{|\varphi^\mfa\rangle} \\
\sim & \, \eta_{a,b_L, c}^{\mu_\mfa},
\end{aligned}
\end{equation}
with regions shown in Fig.~\ref{fig:proof}. When $l_a=O(1)$, $l_{ab_L},l_{ab_R}\sim L_x^\lambda$, we obtain  $\eta_{a,b_L, c} = O(1) L_x^{-\lambda}$.
Thus,
\begin{equation}
I(A:C|B)_{\rho^\mfa}
= O(1)
L_x^{-\lambda \mu_\mfa}
.
\end{equation}
Hence
\begin{equation}
\sum_\mfa p_\mfa I(A:C|B)_{\rho^\mfa}
=O(1) 
L_x^{-\lambda \mu^*},
\end{equation}
where
\begin{equation}
\mu^{*}
:=
\min_\mfa \mu_\mfa.
\end{equation}

The Holevo leakage terms $\chi_A(p)$ and $
\chi_{AB}(p)$ are sums of relative entropies between anyon sectors and the mixture $\bar{\rho}^p$. Repeating the joint convexity argument (similar to the steps from Eq.~\eqref{eq:jointconv} to Eq.~\eqref{eq:power-A}), we derive 
\begin{equation}
\chi_A(p)+\chi_{AB}(p)
=O(1)
L_x^{-(1-\lambda)\gamma^*}.
\end{equation}
This follows from the definition of $\gamma^*$.

Combining all estimates for the deficit of relative entropy, we get
\begin{equation}
\begin{aligned}
&S(\rho_{ABC}\|\sigma_{ABC})
-
S(\rho_{BC}\|\sigma_{BC})
\\&\le
m_1
L_x^{-(1-\lambda)\gamma^*}
+
m_2 L_x^{-\lambda\mu^*},
\end{aligned}
\end{equation}
where $m_1$ and $m_2$ are two positive $O(1)$ constants.

Now consider an arbitrary channel $\calN_A$,
Applying 
\begin{equation}
\begin{aligned}
&-\log
F\left(
\rho_{ABC},
\calR_{AB}\circ \calN_A(\rho_{ABC})
\right)\\
=&-\log
F\left(
\rho_{ABC},
\calE_{B\to AB}(\rho_{BC})
\right)
\\&\le
m_1
L_x^{-(1-\lambda)\gamma^*}
+
m_2 L_x^{-\lambda\mu^*}
\end{aligned}
\end{equation} 
Here $\calR_{AB}:= \calE_{B\to BA}\circ \Tr_A$ as defined in Thm.~\ref{thm:quasi-local-recovery}. The first line follows from the property of the channel on $A$. The second line uses the inequality in Eq.~\eqref{eq:recovery_theorem}. 

Since the RHS is small when $L_x\gg 1$, we have the upper bound of the infidelity
\begin{equation}
\begin{aligned}
  &1-
F\left(
\rho_{ABC},
\calR_{B\to AB}(\rho_{BC})
\right)
\\&\le
m_1
L_x^{-(1-\lambda)\gamma^*}
+
m_2 L_x^{-\lambda\mu^*}.
\end{aligned}
\end{equation}
Further, we can optimize the scaling law over $\lambda$ (to make the infidelity vanish most rapidly with $L_x$) by equating the exponents,
\begin{equation}
(1-\lambda)\gamma^*
=\lambda\mu^*,
\end{equation}
this gives the optimal $\lambda$:
\begin{equation}
\lambda^*
=\frac{\gamma^*}
{\gamma^*+\mu^*}.
\end{equation}
In conclusion, we have
\begin{equation}
1-F\left(
\rho_{ABC},
\calR_{B\to AB}(\rho_{BC})
\right)\le O(1)L_x^{
-\frac{\gamma^*\mu^*}
{\gamma^*+\mu^*}
},
\end{equation}
for all kinds of quantum channels applied on the local disk $A$ near one of the physical edges, with an explicit choice of $\calR_{AB}$ that only depends on the code subspace but not the noise channel.

\subsection{A trick for computing $\gamma_{\mfa \mfb}$}\label{app:gamma-computation}

We present a useful proposition that allows us to compute $\gamma_{\mfa\mfb}$ exactly using the dimensional reduction picture.
Recall that, $\gamma_{\mfa\mfb}$ is defined according to
\begin{equation}
    S(\rho^\mfa_A || \rho^\mfb_{A}) \sim x^{\gamma_{\mfa \mfb}}
\end{equation}
for region $A$ in Fig.~\ref{fig:Markov-Abelian}. We now define $\alpha_{\mfa\mfb}$ according to
\begin{equation}
    S(\rho^\mfa_{ABC} || \rho^\mfb_{ABC}) \sim x^{\alpha_{\mfa \mfb}}
\end{equation}
for regions in the same figure.
It is evident that $\alpha_{\mfa\mfb}$ can be computed by the dimensional reduction picture, as the single-interval relative entropy between primary states of a 1D CFT.

\begin{figure}[h]
    \centering
    \includegraphics[width=0.34\linewidth]{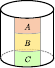}
    \caption{Partitions of a full column $X=ABC$ for the proof of Prop.~\ref{prop:alpha=gamma}. Markov condition $I(A:C|B)\approx 0$ holds for the Abelian code state $|\psi^\mfa\rangle$ of the chiral edge code as a consequence of the full boundary {\bf A1}.}
    \label{fig:Markov-Abelian}
\end{figure}

\begin{Proposition}\label{prop:alpha=gamma}
   Consider a pair of code words in the chiral edge code, $|\psi^\mfa\rangle$ and $|\psi^\mfb\rangle$. Suppose that the values of $\gamma_{\mfa \mfb}$ computed from a local region touching the upper and the lower edges are identical and both $\mfa$ and $\mfb$ are Abelian. Then, 
   \begin{equation}
       \gamma_{\mfa \mfb} = \alpha_{\mfa \mfb}. 
   \end{equation}
   \end{Proposition}

Therefore, the exponents $\gamma_{\mfa \mfb}$, in a large class of codewords, can be analytically computed in the 1D CFT obtained from the dimensional reduction (Sec.~\ref{subsec:dimension-reduction}). Details of such computation for small intervals are presented in Appendix~\ref{app:CFT_relative_entropy}. 

\begin{proof} 
    Suppose both $\mfa$ and $\mfb$ are Abelian, and let the full column $X = ABC$ as in Fig.~\ref{fig:Markov-Abelian}.
    The Markov condition that follows from the full boundary {\bf A1} implies that
    \begin{equation}
        \log{\rho_{ABC}} \approx \log{\rho_{AB}} +\log{\rho_{BC}} - \log{\rho_{B}}.
    \end{equation} 
    Plugging this in for both states, we have
    \begin{equation}
        \begin{aligned}
            & \,\,\,\, S(\rho^\mfa_X|| \rho^\mfb_X) \\
            &= S(\rho^\mfa_{ABC}|| \rho^\mfb_{ABC}) \\
            &= S(\rho^\mfa_{AB}|| \rho^\mfb_{AB}) + S(\rho^\mfa_{BC}|| \rho^\mfb_{BC}) - S(\rho^\mfa_{B}|| \rho^\mfb_{B}) \\
            &= S(\rho^\mfa_{AB}|| \rho^\mfb_{AB}) + S(\rho^\mfa_{BC}|| \rho^\mfb_{BC}) \\
            & \sim x^{\gamma_{\mfa \mfb}}.
        \end{aligned}
    \end{equation}
    The 2nd line follows from $X=ABC$. The 3rd line follows from replacing the entanglement Hamiltonian ($-\log \rho_{ABC}$) with that of the subsystems. The 4th line follows from $\rho^\mfa_B = \rho^\mfb_B$ for bulk disk $B$. The last line follows from the assumption on the identification of $\gamma_{\mfa \mfb}$ computed on the local region on the upper edge ($AB$) and the local region on the lower edge ($BC$). 

  Finally, by $S(\rho^\mfa_X|| \rho^\mfb_X)\sim x^{\alpha_{\mfa \mfb}}$ we derive $\gamma_{\mfa \mfb}=\alpha_{\mfa \mfb}$.
  \end{proof}

\begin{figure*}[t]
    \centering
\begin{minipage}[c]{0.23\textwidth}
        \centering
        \includegraphics[width=\linewidth]
        {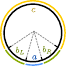}
    \end{minipage}
    \hfill
    \begin{minipage}[c]{0.32\textwidth}
        \centering
        \includegraphics[width=\linewidth]
        {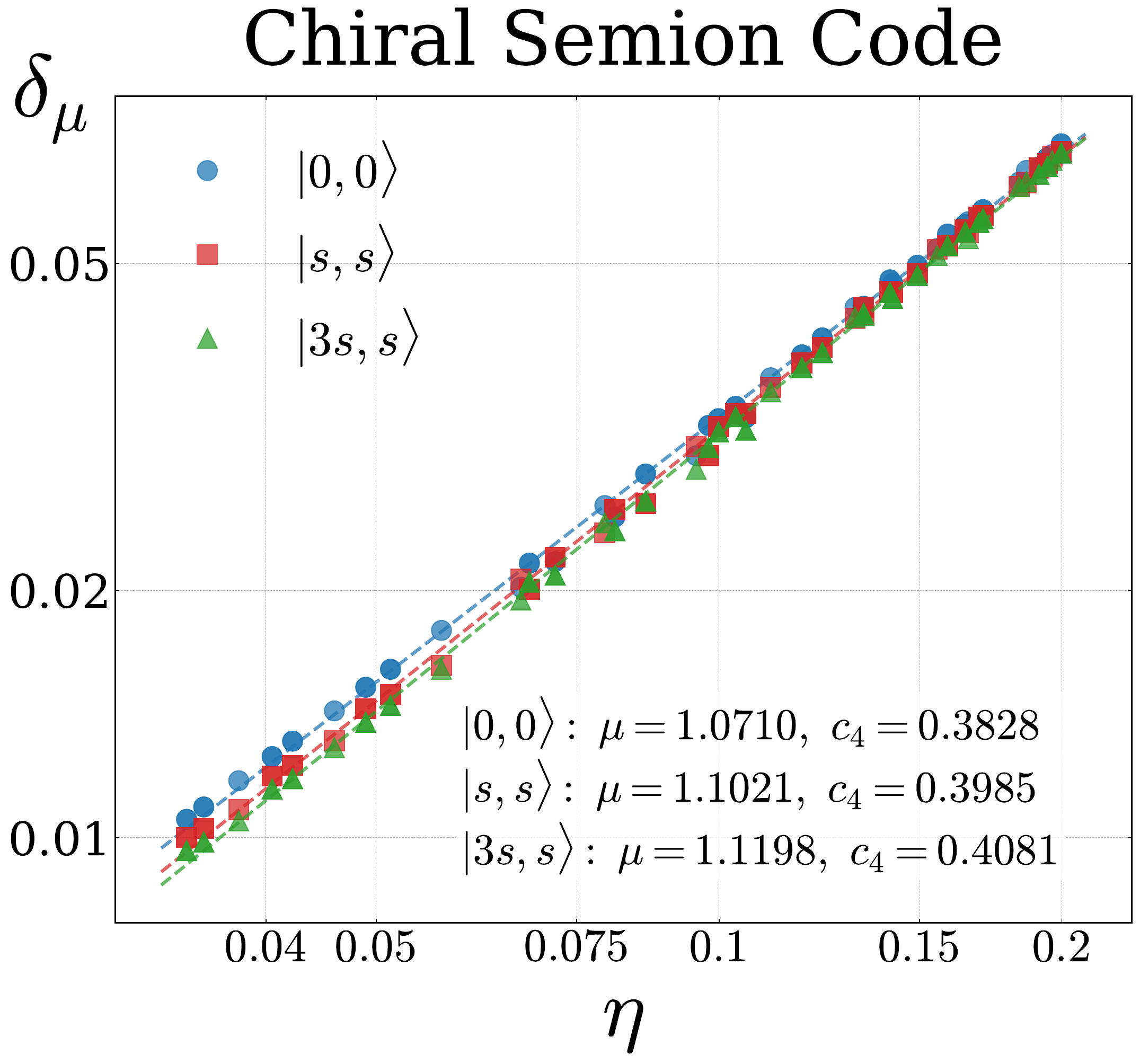}
    \end{minipage}
    \hfill
    \begin{minipage}[c]{0.32\textwidth}
        \centering
        \includegraphics[width=\linewidth]
        {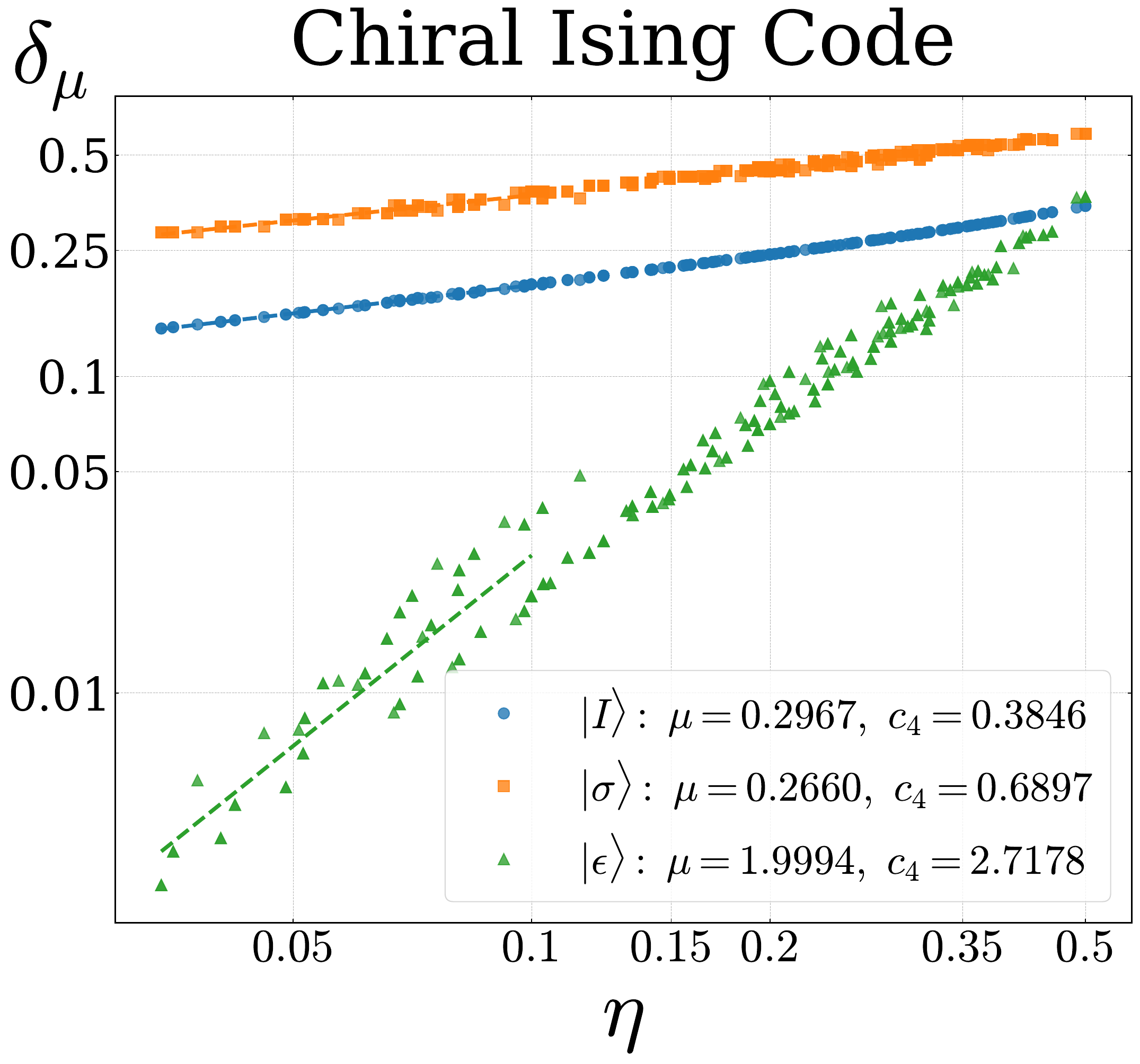}
    \end{minipage}
\caption{{\bf Fitting $\mu$ on a finite circle.} Left: Partitions used in the calculation of $\mu_\mfa$. $\delta_{\mu_\mfa}:=I(a:c)_{\ket{\varphi^\mfa}}\approx c_4\eta^{\mu_\mfa}$ through correspondence in 1D. Middle: For the Chiral semion code, we test the code states $\ket{\varphi_1^{0,0}},\ket{\varphi_{s}^{\frac{1}{2},\frac{1}{2}}},\ket{\varphi_{s}^{\frac{3}{2},\frac{1}{2}}}$, and find $\mu\approx 1.0710, 1.1021, 1.1198$ with $\eta\leq 0.2$. Right: For the chiral Ising code, we test the code states $|\varphi^1\rangle$, $|\varphi^\sigma\rangle$ and $|\varphi^\epsilon\rangle$. We found that the corresponding exponents are $\mu_I\approx 0.2967,\mu_{\sigma}\approx 0.2660,\mu_{\epsilon}\approx 1.9994$ fitting the data with $\eta\leq 0.1$. } 
    \label{fig:edge_mu}
\end{figure*}

\subsection{Computation of relevant exponent $\mu_a$}\label{app:mu-computation}

In this appendix, we compute the exponent $\mu_a$ that is relevant to the power-law-range recovery channel in Sec.~\ref{sec:quasi-local-R}. Recall that $\mu_\mfa$ is defined according to an empirical formula of mutual information of a certain 1D CFT state $|\varphi^\mfa\rangle$ obtained from the dimensional reduction of chiral states on a cylinder. 
\begin{equation}
    I(a:c)_{|\varphi^\mfa\rangle} \sim \eta_{a,b_L,c}^{\mu_\mfa}.
\end{equation}
for low-lying primary states. Note that, the cross-ratio dependence is best justified for the ground state, which has the global conformal symmetry. But we observe numerically that a similar condition holds for low-lying primary states, at least for small $\eta$. Here the cross-ratio for three adjacent intervals $a,b_L,c$ is defined in terms of chord distances as
\begin{equation}
    \eta_{a,b_L,c} = \frac{\sin(\theta_{a}/2)\sin(\theta_{c}/2)} {\sin(\theta_{ab_L}/2)\sin(\theta_{b_L c}/2)}.
\end{equation}
The partitions are shown in Fig.~\ref{fig:edge_mu}.

\subsubsection{Chiral Ising code}

{\bf Computation of $\mu_a$:} For the ground state value $\mu_{1}$, a finite size computation on Ising spin chain gives  
\begin{equation}
    \mu_1 \approx 0.30.
\end{equation}
This is reasonably consistent with the theoretical prediction of Cardy that for small $\eta$, the mutual information $I(a,c)_{\ket{\varphi_1}} \sim \eta^{1/4}$~\cite{2009JSMTE..11..001C,2011JSMTE..01..021C}.

For the result shown in Fig.~\ref{fig:edge_mu}, we use system sizes $N=20,22,24$, with interval sizes $l_a\in\{2,3,4\}$ and $l_c\in\{1,\ldots,6\}$. For each choice of ($N,l_a,l_c$), the two complementary gaps $l_{b_L},l_{b_R}$ are chosen as near-balanced partitions of the remaining length $N-l_a-l_c$ : explicitly, $l_{b_R}=\left\lfloor \frac{N-l_a-l_c}{2}\right\rfloor+s,\quad l_{b_L}=N-l_a-l_c-l_{b_R},$ with $s\in\{-2,-1,0,1,2\}$, retaining only geometries with $l_{b_L},l_{b_R}\ge 4$. For the power-law fitting, we restricted $\eta\leq0.1$.

For the Abelian anyon $\epsilon$ in the Ising anyon theory, we compute $\mu_\epsilon$ as approximately $2.00$, with data shown in Fig.~\ref{fig:edge_mu}.

\subsubsection{Chiral Semion code}
We numerically compute both $\mu_1$ and $\mu_s$ of the chiral semion code, making use of the semion chain of Ref.~\cite{Nielsen2012,Nielson2014} (see also Appendix~\ref{app:semion_chain}). By a finite-size simulation, we fit the power-law exponents as 
\begin{equation}
    \mu \approx 1.07,  1.10, 1.12,
\end{equation}
for primary states $\ket{\varphi_1^{0,0}},\ket{\varphi_s^{\frac{1}{2},\frac{1}{2}}},\ket{\varphi_s^{\frac{3}{2},\frac{1}{2}}}$, respectively. The result of ground state is reasonably consistent with the theoretical prediction $\mu_1=1$~\cite{2009JSMTE..11..001C,2011JSMTE..01..021C}. For the result shown in Fig.~\ref{fig:edge_mu}, we use system size $N=20,22,24$, with interval sizes $l_a=2$ and
$l_c\in\{1,\ldots,6\}$. For each choice of $(N,l_a,l_c)$, the two complementary gaps $l_{b_L},l_{b_R}$ are chosen as near-balanced partitions of the remaining length $N-l_a-l_c$: explicitly,
  \(
  l_{b_R}=\left\lfloor \frac{N-l_a-l_c}{2}\right\rfloor+s,
  l_{b_L}=N-l_a-l_c-l_{b_R},
  \)
  with $s\in\{-2,-1,0,1,2\}$, retaining only geometries with
  $l_{b_L},l_{b_R}\ge 4$. The same set of geometries is used for the
  $|\varphi_1^{0,0}\rangle$, $|\varphi_s^{\frac{1}{2},\frac{1}{2}}\rangle$ and $|\varphi_s^{\frac{3}{2},\frac{1}{2}}\rangle$ sectors. For the power-law fitting and the plotted data, we further restrict to $\eta\le 0.2$.

\section{Effect of IID noise}
\label{app:IID_noise}


The main text studies robustness through coherent-information loss under
geometrically local erasure. In this appendix, we give a complementary
finite-size diagnostic for a different, more extensive, noise model: IID
Pauli dephasing. Here IID means that the same single-site channel is applied
independently to every physical site, and the error locations are not supplied
to the decoder. We only study the weak-noise regime, and we do not attempt to locate a possible mixed-state decoding transition at
larger physical error rates. Instead, we ask whether small IID noise flows
toward or away from the no-noise fixed point as the system size is increased.

For the numerical examples below, we use the bosonic Laughlin state at filling
\(\nu=1/2\) as in Appendix~\ref{app:FQH_wavefunction}. (We note that the decoherence transition of Laughlin states has been studied recently~\cite{2025arXiv251008490W}. Our computation is different in that we consider topological order with edges, and we are mainly interested in weak noise.)
We choose the following code subspace of the chiral semion code  
\begin{equation}
\mathbb{V}^{\chi} \bigl(\cylindericon[1.2]\bigr)
=
\mathrm{span}\{\ket{\psi_1^{0,0}},\ket{\psi_s^{\frac{1}{2},-\frac{1}{2}}}\},
\label{eq:IID_codesubsapce}
\end{equation}
with the corresponding CFT primary states of its dimensionally reduced one-dimensional compact-boson CFT code. The notation \(\ket{\psi_1^{0,0}}\) and \(\ket{\psi_s^{\frac{1}{2},-\frac{1}{2}}}\) is the same as in Example~\ref{exmp:chiral_Semion}. 

We apply the IID Pauli dephasing channel 
\begin{equation}
\mathcal{N}_{p,P}
=
\bigotimes_{j=1}^{N}\mathcal{N}_{p,P}^{[j]},
\quad
\mathcal{N}_{p,P}^{[j]}(\rho)
=
\left(1-\frac{p}{2}\right)\rho
+
\frac{p}{2}\,
P_j\rho P_j ,
\label{eq:IID_Pauli_channel}
\end{equation}
where \(P\in\{X,Z\}\), and \(N\) is the number of qubits of the
corresponding lattice realization. The convention in
Eq.~\eqref{eq:IID_Pauli_channel} is the same as in
Ref.~\cite{Sang2024AQECC}: \(p=1\) gives complete
single-site dephasing in the \(P\) basis.

Let \(\ket{\psi_{RQ}}\) be the maximally entangled state between the reference
system \(R\) and the code subspace, as in Eq.~\eqref{eq:maximal_entangled}.
After the noise channel, we write
\begin{equation}
\rho_{RQ}(p)
=
\mathcal{N}_{p,P}
 \left(
\ket{\psi_{RQ}}\bra{\psi_{RQ}}
\right),
\end{equation}
and compute
\begin{equation}
I_c(p)
=
S\left(\rho_{Q}(p)\right)
-
S\left(\rho_{RQ}(p)\right).
\label{eq:Ic_IID_def}
\end{equation}  

\subsection{Weak-noise scaling ansatz}

Following the coherent-information scaling hypothesis of
Ref.~\cite{Sang2024AQECC}, we assume that in the
weak-noise regime
\begin{equation}
I_c(p)
=
f\left(pN^{\zeta_P}\right),
\quad
p\rightarrow 0,\quad N\rightarrow\infty .
\label{eq:IID_scaling_collapse}
\end{equation}
Here \(\zeta_P\) is the IID scaling-collapse exponent for Pauli type \(P\). (The exponent \(\zeta_P\)
corresponds to the exponent denoted by \(\nu\) in
Ref.~\cite{Sang2024AQECC}.)

The scaling variable in Eq.~\eqref{eq:IID_scaling_collapse} is $x=pN^{\zeta_P}$. At fixed weak physical error rate \(p\), a negative \(\zeta_P\) gives \(x\rightarrow0\) as \(N\rightarrow\infty\). Since
\(f(0)=\log D\), this is consistent with $I_c(p)\longrightarrow\log D$, and hence with weak-noise recoverability in the thermodynamic limit. A more
negative value of \(\zeta_P\) gives a faster flow of \(x\) back to
zero and therefore a stronger finite-size restoration of coherent information.
By contrast, \(\zeta_P=0\) gives no size-improvement at fixed \(p\), and
\(\zeta_P>0\) means that the weak-noise perturbation grows with system size.

\begin{figure}[H]
    \centering
    \includegraphics[width=0.57\linewidth]{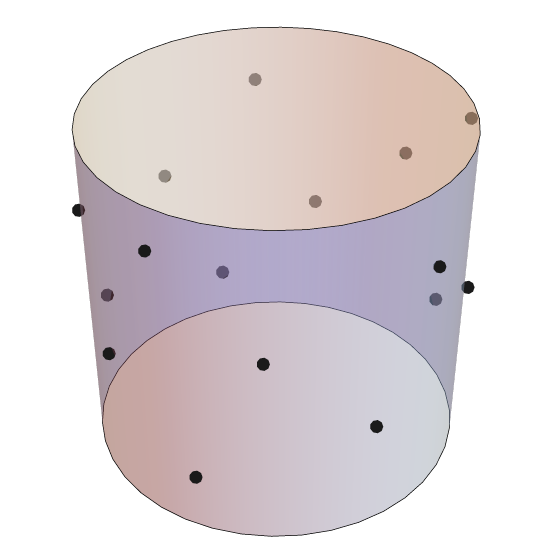}
    \caption{An illustration of the golden-cylinder lattice
    \(\Lambda_b^{[y]}(N)\) with \(b=0.1\). The number of sites is $N=16$ for this figure.}
    \label{fig:golden-cylinder-y}
\end{figure}

\subsection{Finite size scaling with Golden-cylinder coordinate}

We now describe the finite-size scaling method for $\nu=1/2$ fractional quantum Hall states. 
We use the analytical wave function described in Appendix~\ref{app:FQH_analytical_wavefunction}, and we scale up the size of the cylinder $N$.
We map the cylinder coordinate $w=u_1+iu_2$, ($ u_2\sim u_2+2\pi$), to the complex plane by \(z=e^w\). Because the Laughlin-type lattice wave
functions \cite{Nielsen2012,Nielson2014} are invariant under a global dilation of all \(z_j\), a global
translation along \(u_1\) or a global rotation along \(u_2\) does not affect
the universal quantities computed below.

Let \(\theta_{\rm g}\) be the \emph{golden angle}, $\theta_{\rm g}
=\pi(3-\sqrt{5})\approx 137.5^\circ$. For a fixed positive parameter \(b\), we define the golden-cylinder lattice
\(\Lambda_b^{[y]}(N)\) by placing \(N\) physical sites at
\begin{equation}
\begin{aligned}
\bigl(u_1(j),u_2(j)\bigr)
&=
\bigl(jb,(j-1)\theta_{\rm g}\bigr),
 \\
z_j
&=
e^{jb}e^{i(j-1)\theta_{\rm g}},
\label{eq:GammaYcoordinates}
\end{aligned}
\end{equation}
in which $j=1,\ldots,N$. The golden-angle sequence distributes the sites quasi-uniformly around the
periodic direction, while increasing \(N\) extends the lattice along the
cylinder. The usage of the golden angle to make a reasonably uniform lattice is a trick to improve finite-size scaling quality; see Ref.~\cite{modular-commutator,Sharma2026}.
In the data below, we use \(b=0.1\) for the golden cylinder.

\subsection{Numerical comparison: \(1\)D chain versus \(2\)D golden cylinder}
\label{subsec:numerics_IID_chain_vs_cyl}

We compare IID Pauli dephasing for the dimensionally reduced \(1\)D CFT code $\mathbb{V}^{\operatorname{CFT}}=\Gamma(\mathbb{V}^{\chi}\bigl(\cylindericon[1.2]\bigr))$
and the \(2\)D chiral edge code
\(\mathbb{V}^{\chi}\bigl(\cylindericon[1.2]\bigr)\). In both cases the
coherent information is computed using the purification algorithm described in
Appendix~\ref{app:alg_CI}.

\paragraph{\(1\)D CFT code.}
The collapse data for the dimensionally reduced \(1\)D code are shown in
Fig.~\ref{fig:CI_chain}. The best-fit exponents are
\begin{equation}
\zeta_Z^{\rm CFT}\approx -0.2000,
\qquad
\zeta_X^{\rm CFT}\approx 0.0000 .
\label{eq:CFT_IID_exponents}
\end{equation}
Thus \(Z\)-dephasing shows only a mild flow back toward the noise-free fixed
point, while \(X\)-dephasing is approximately size-independent over the
accessible sizes. In the weak-noise scaling interpretation, the \(1\)D code
therefore has weak or absent finite-size enhancement of \(I_c\) under these
IID channels.

\begin{figure}[h]
    \centering
    \includegraphics[width=1.05\linewidth]{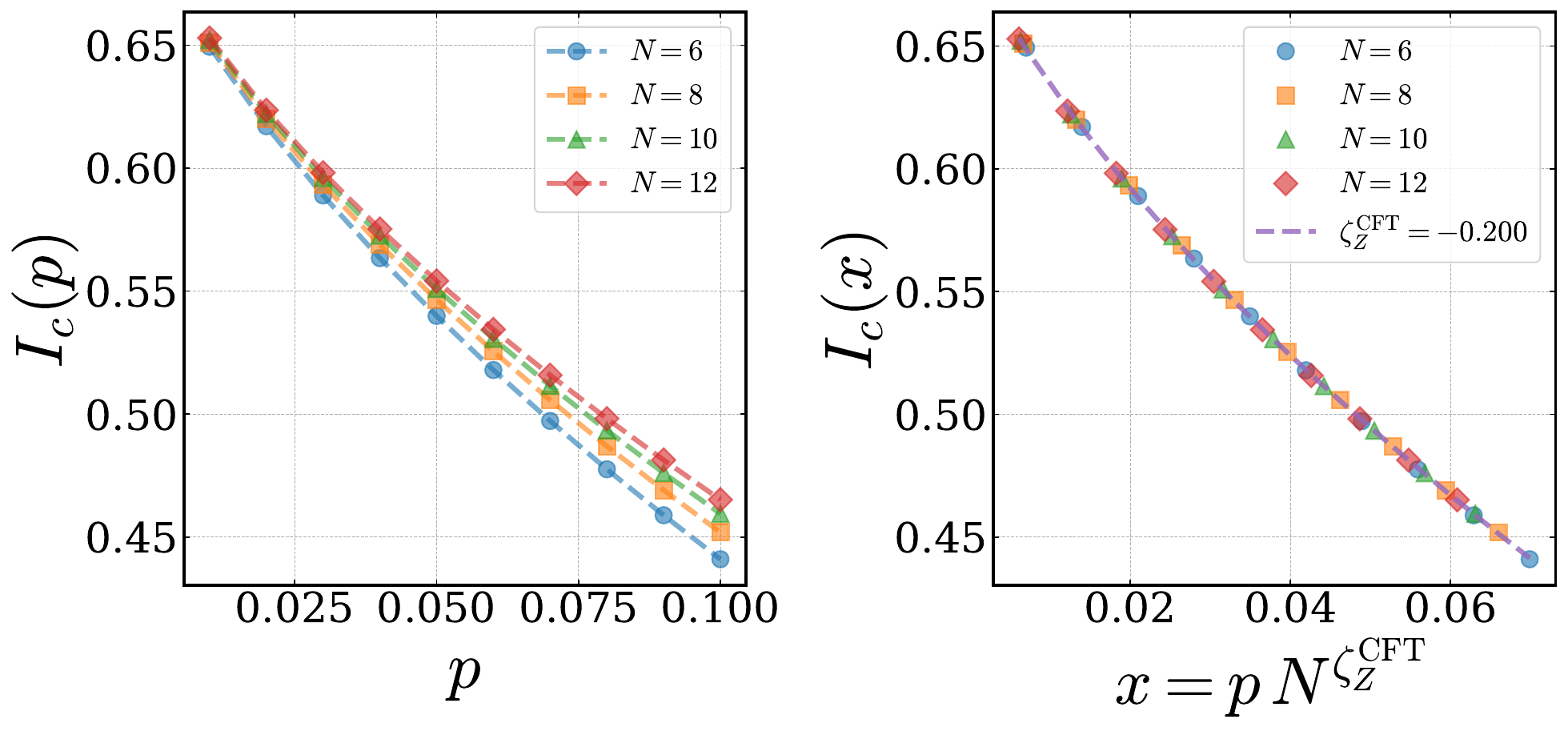}\\[-2pt]
    \includegraphics[width=1.05\linewidth]{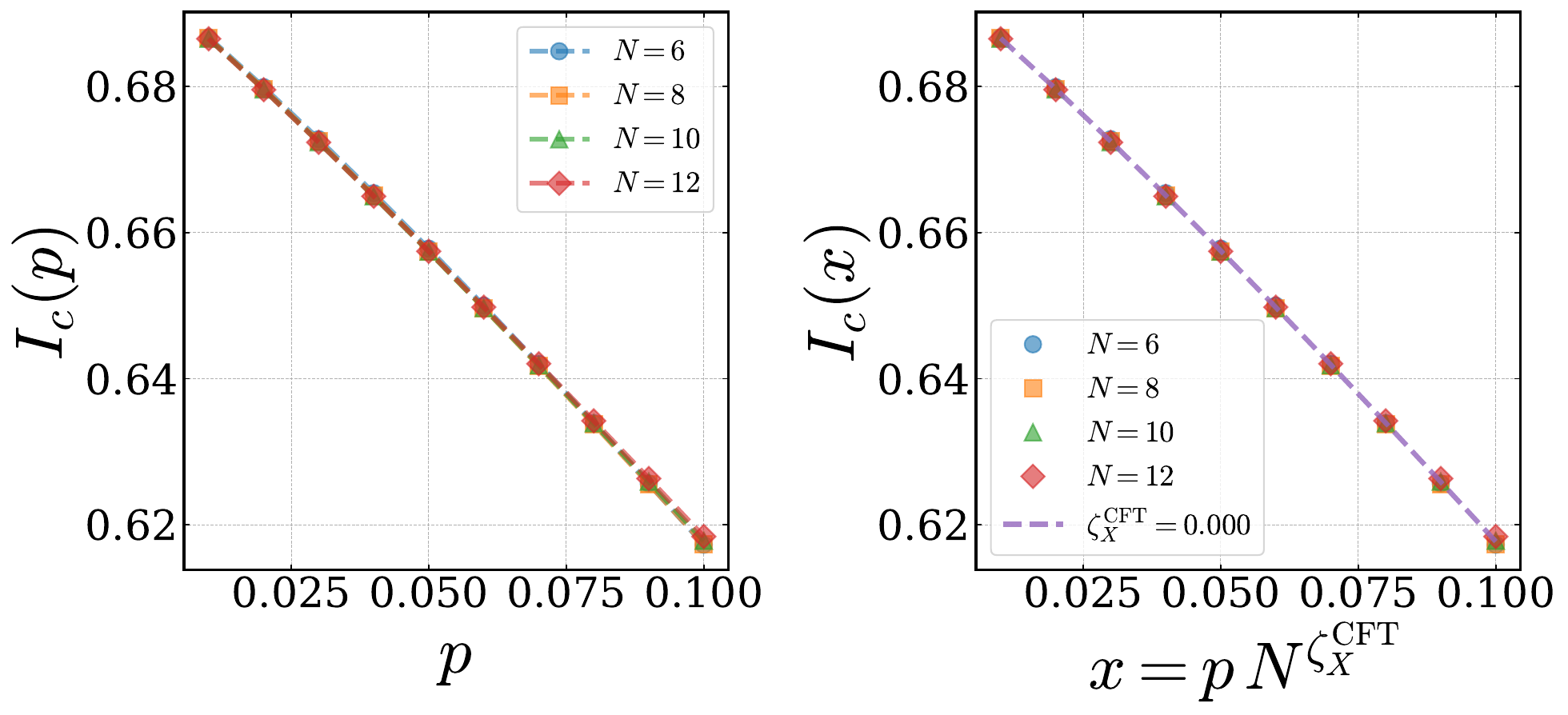}
    \caption{\textbf{Coherent information for the dimensionally reduced
    \(1\)D compact-boson CFT code}
    \(\mathbb{V}^{\operatorname{CFT}}=\Gamma(\mathbb{V}^{\chi}\bigl(\cylindericon[1.2]\bigr))=\mathrm{span}\{\ket{\varphi_1^{0,0}},\ket{\varphi_s^{\frac{1}{2},-\frac{1}{2}}}\}\)
    under IID Pauli dephasing. Each panel shows \(I_c(p)\) versus \(p\)
    for \(N=6,8,10,12\), together with a collapse using
    \(x=pN^{\zeta_P^{\text{CFT}}}\). The best-fit exponents are
    \(\zeta_Z^{\rm CFT}\approx -0.2000\) and
    \(\zeta_X^{\rm CFT}\approx 0.0000\).} 
    \label{fig:CI_chain}
\end{figure}

\paragraph{\(2\)D chiral edge code.}
The \(2\)D golden-cylinder realization gives a markedly different finite-size
flow. As shown in Fig.~\ref{fig:CI_cyl}, the best-fit exponents are
\begin{equation}
\zeta_Z^{\chi}\approx -0.6533,
\qquad
\zeta_X^{\chi}\approx -0.7600 .
\label{eq:chi_IID_exponents}
\end{equation}
Both exponents are negative and substantially smaller than the corresponding
\(1\)D values. Consequently, for fixed weak \(p\), the scaling variable
\(pN^{\zeta_P^\chi}\) decreases rapidly with \(N\), and the data are
consistent with
\begin{equation}
I_c(p)\rightarrow \log D
\end{equation}
in the $p\rightarrow0,\,\, N\rightarrow\infty$ limit.

\begin{figure}[h]
    \centering
    \includegraphics[width=1.00\linewidth]{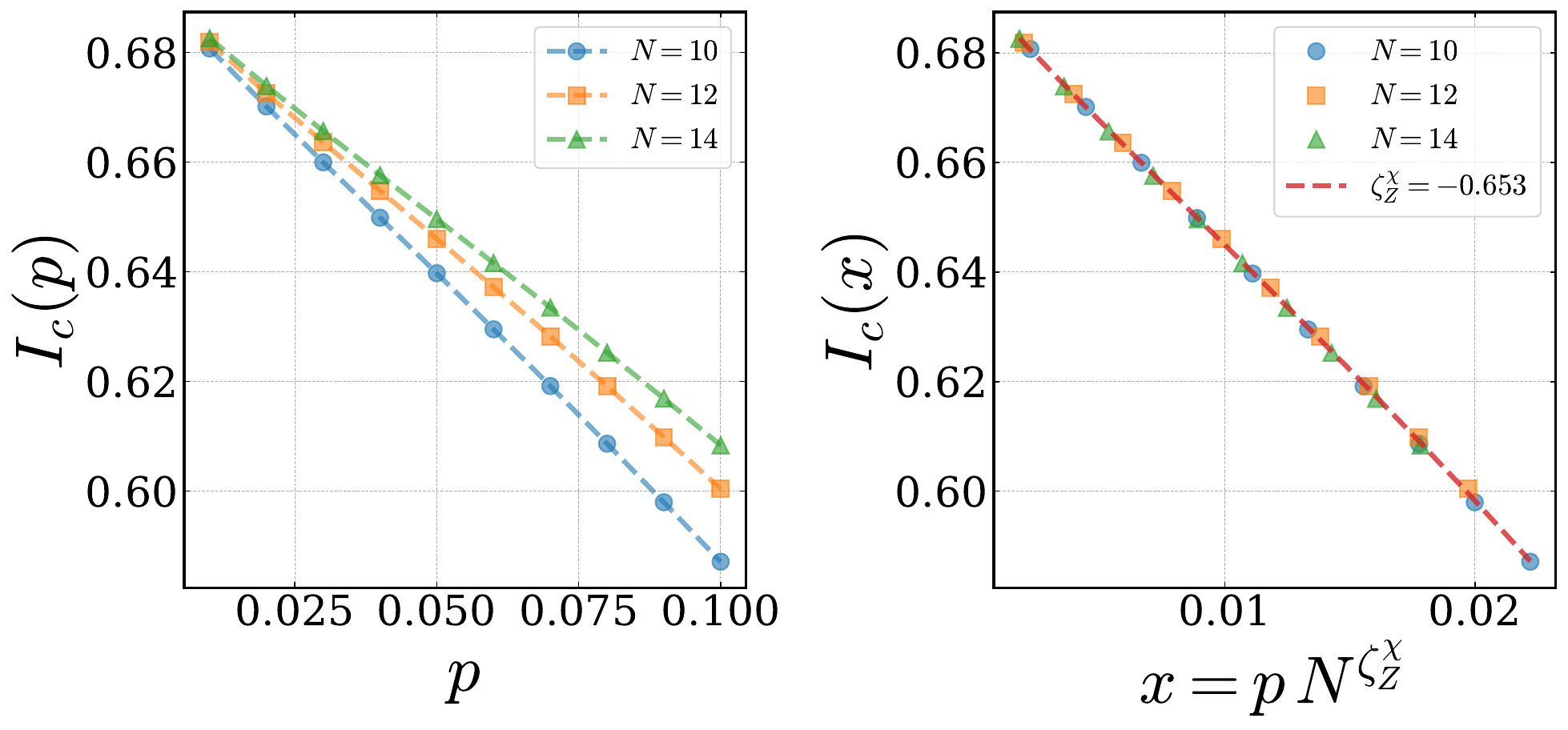}\\[-2pt]
    \includegraphics[width=1.00\linewidth]{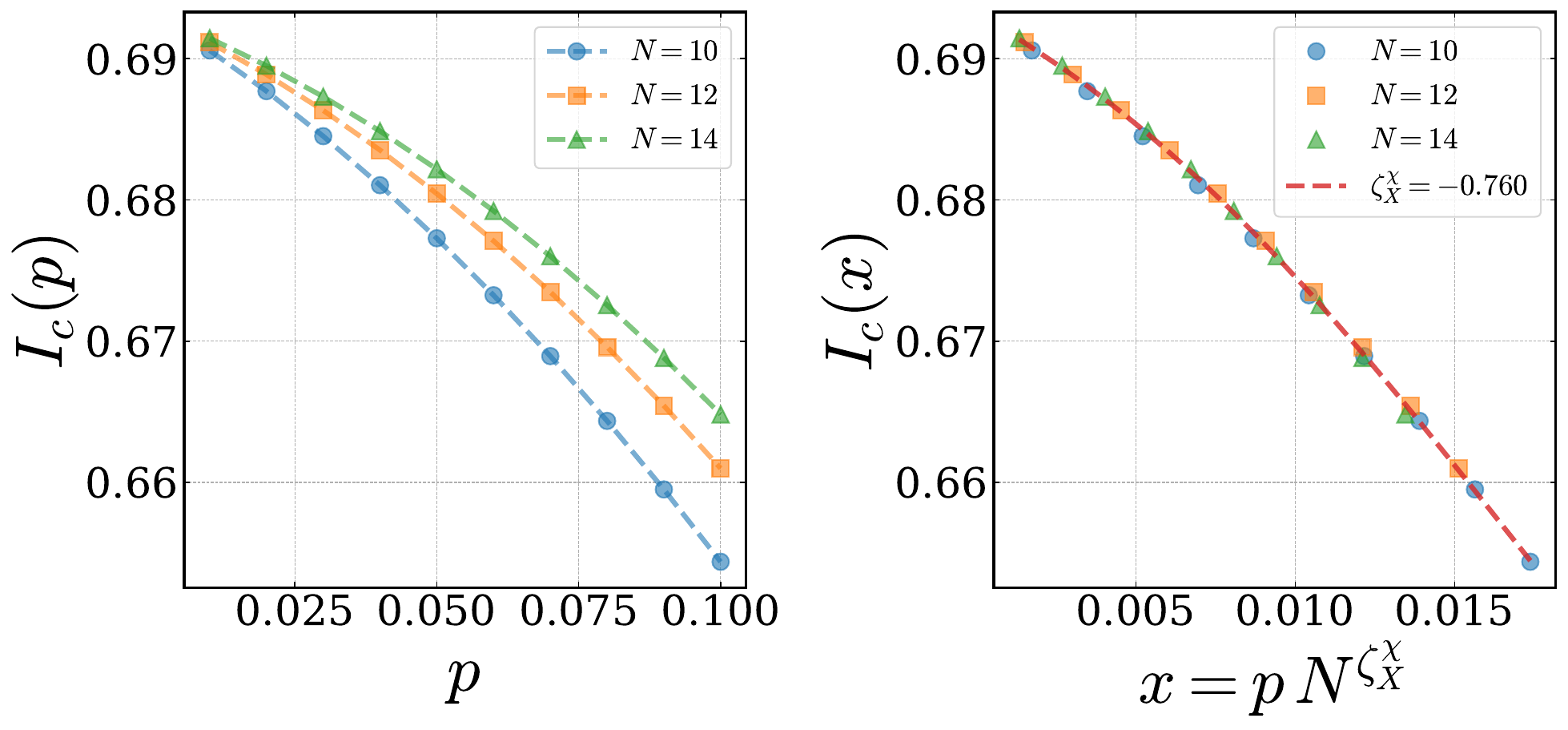}
    \caption{\textbf{Coherent information for the \(2\)D golden-cylinder
    realization} of the chiral semion edge code
    \(\mathbb{V}^{\chi}\bigl(\cylindericon[1.2]\bigr)
    =\mathrm{span}\{\ket{\psi_1^{0,0}},\ket{\psi_s^{\frac{1}{2},-\frac{1}{2}}}\}\)
    under IID Pauli dephasing. We plot \(I_c(p)\) for
    \(N=10,12,14\) and perform the collapse using
    \(x=pN^{\zeta_P^{\chi}}\). The best-fit exponents are
    \(\zeta_Z^{\chi}\approx -0.6533\) and
    \(\zeta_X^{\chi}\approx -0.7600\), indicating a flow toward
    \(I_c=\log D\) at fixed weak \(p\).  }
    \label{fig:CI_cyl}
\end{figure}

The comparison is summarized in Table~\ref{tab:IID_exponent_comparison}.
For both Pauli channels, the \(2\)D chiral edge code has the smaller
collapse exponent:
\begin{equation}
\zeta_Z^\chi < \zeta_Z^{\rm CFT},
\qquad
\zeta_X^\chi < \zeta_X^{\rm CFT}.
\end{equation} 
This is the IID-noise analog of the main-text robustness comparison under
local erasure. Under local erasure, the \(2\)D chiral edge code is more robust
because the coherent-information loss is controlled by a larger positive exponent. Under weak IID noise, the \(2\)D chiral edge code is more robust because
the coherent information flows back toward its maximal value with a more
negative scaling-collapse exponent.

\begin{table}[h]
\centering
\setlength{\extrarowheight}{5pt}
\begin{tabular}{|c|c|}
\hline 
 \(Z\)-dephasing & \(X\)-dephasing  \\
 \hline
\(\zeta_Z^{\rm CFT}\approx -0.2000\) 
&
\(\zeta_X^{\rm CFT}\approx 0.0000\) 
\\
\hline
\(\zeta_Z^{\chi}\approx -0.6533\) 
&
\(\zeta_X^{\chi}\approx -0.7600\)
\\
\hline
\end{tabular}
\caption{Effective weak-noise scaling-collapse exponents for IID Pauli
dephasing with code subspace $\mathbb{V}^{\chi}\bigl(\cylindericon[1.2]\bigr)
=
\mathrm{span}\{\ket{\psi_1^{0,0}},\ket{\psi_s^{\frac{1}{2},-\frac{1}{2}}}\}$. Smaller, i.e. more negative, \(\zeta_P\) means that the scaling
variable \(pN^{\zeta_P}\) flows more rapidly toward zero at fixed weak
physical error rate \(p\), and hence that the coherent information remains
closer to \(\log D\) as the system size grows.}
\label{tab:IID_exponent_comparison}
\end{table}

The comparison in Table~\ref{tab:IID_exponent_comparison} should be interpreted as a finite-size weak-noise observation, rather than as a derivation of an IID-noise stability mechanism. It suggests an analog of the local-erasure robustness hierarchy studied in the main text: for the code subspace and Pauli channels considered here, the scaling variable
\(pN^{\zeta_P}\) flows more rapidly toward the no-noise fixed point in the two-dimensional chiral edge realization than in the dimensionally reduced one-dimensional CFT code. This observation is naturally compared with the CFT-code analysis of IID dephasing in Ref.~\cite{Sang2024AQECC}, and with recent work on noisy topological and fractional quantum Hall mixed states~\cite{PRXQuantum.5.020343, hlfh-86yz, PhysRevA.111.032402, 2025arXiv251222121V, 2025arXiv251008490W,Sang2025Markov-length,Negari2026}.

At present, however, we do not have a first-principles explanation of the
collapse exponents \(\zeta_P^\chi\) in Table~\ref{tab:IID_exponent_comparison}, nor do we know whether they are fixed by the edge-local exponent \(\gamma\), by another universal edge datum, or by more microscopic features of the lattice realization and noise channel. Also, we only fixed $b=0.1$ in the numerical simulation, which means increasing $N$ is only making the cylinder longer rather than extending the system to the thermodynamic limit. A
plausible interpretation of the enhanced robustness in chiral edge code is that the two-dimensional geometry leaves the two physical edges separated by a gapped topological bulk, while dimensional reduction removes this separation by treating an entire column of the
cylinder as a local degree of freedom of the one-dimensional CFT code. This geometric distinction may suppress the ability of weak independent local errors to build up an effective process that distinguishes the encoded anyon sector. We regard this as a guiding picture rather than a demonstrated mechanism.

We emphasize that the above data provide only finite-size evidence in the weak-noise regime. They do not determine a threshold \(p_c\), nor do they rule out a mixed-state decoding transition at larger \(p\). Moreover, the numerics are restricted to one chiral semion code subspace, two Pauli dephasing channels, and the accessible golden-cylinder system sizes. The
main conclusion is therefore modest: near \(p=0\), these data are consistent with a more favorable finite-size flow of coherent information for the two-dimensional chiral edge code than for its one-dimensional dimensional reduction. Whether this behavior persists for other code subspaces, other local noise channels, and larger system sizes is an important open question.

\section{Algorithm for calculating coherent information}
\label{app:alg_CI}

We explain a numerical method that efficiently computes the coherent information at arbitrary $p$ for a certain error type, such as the $X$ or $Z$ error on a single site. The key idea is that we can purify such single-qubit channels with one ancilla qubit.

Let $|\psi\rangle\in \otimes_{j=1}^N\mathbb{C}^{2}$ be a wave function for $N$ qubits, and $P\in\{X,Y,Z\}$ be the noisy gate acting with probability $p/2$ on a single qubit $n$, thus defining a single-qubit noisy channel: 
\begin{equation}
\rho'=\mathcal{N}_{p,P}(\rho)=(1-\frac{p}{2})\rho+\frac{p}{2} P\rho P^\dagger    
\end{equation}
When $\rho=|\psi\rangle \langle\psi|$, we can construct a purified state for $\rho'$ as 
\begin{equation}
|\psi'\rangle=\sqrt{1-\frac{p}{2}}|\psi\rangle\otimes |0_E\rangle+\sqrt{\frac{p}{2}}(P|\psi\rangle)\otimes |1_E\rangle
\label{eq:pufication}
\end{equation}
by introducing an ancilla qubit $E$ for this single-qubit noise channel, such that, by tracing out the ancilla, we get $\rho'=\Tr_{E}|\psi'\rangle \langle\psi'|$. See Alg.~\ref{alg:purification} for the purification algorithm.
Consequently, the number of added ancilla qubits is equal to the number of applied noisy channels $M$, and the dimension of the purified state becomes $2^{N+M}$. With the purified state, we can trace out the ancilla qubits or the reference qubit to get $\rho_{{Q}}$ and $\rho_{{QR}}$, and then $S_{{Q}}$ and $S_{{QR}}$ for the calculation of coherent information. Note that the algorithm is especially efficient when we apply the noise channels to a subsystem, i.e., when $M\le N$.

\begin{algorithm}
\caption{Purification of noise channel by adding ancilla qubit}\label{alg:purification}
\KwData{state vector $|\psi\rangle$ of shape $\underbrace{(2,2,\dots,2)}_{N}$, noise gate $P$, qubit $n$, error probability $p$.}
\KwResult{purified state vector $|\psi'\rangle$ of shape $\underbrace{(2,2,\dots,2)}_{N+1}$}
Initialize $|\psi'\rangle$ with shape $\underbrace{(2,2,\dots,2)}_{N+1}$, $|\psi_1\rangle$ with shape $\underbrace{(2,2,\dots,2)}_{N}$, and $|\psi_2\rangle$ with shape $\underbrace{(2,2,\dots,2)}_{N}$\;
$|\psi_1\rangle \gets \text{swapaxis}(|\psi\rangle, 1, n) $\;
\If{$P=X$}{
$|\psi_2\rangle[0,\dots]=|\psi_1\rangle[1,\dots]$\;
$|\psi_2\rangle[1,\dots]=|\psi_1\rangle[0,\dots]$\;
}
\If{$P=Y$}{
$|\psi_2\rangle[0,\dots]=-i|\psi_1\rangle[1,\dots]$\;
$|\psi_2\rangle[1,\dots]=i|\psi_1\rangle[0,\dots]$\;
}
\If{$P=Z$}{
$|\psi_2\rangle[0,\dots]=|\psi_1\rangle[0,\dots]$\;
$|\psi_2\rangle[1,\dots]=-|\psi_1\rangle[1,\dots]$\;
}
$|\psi'\rangle[\dots,0]=\sqrt{1-p/2}|\psi_1\rangle$\;
$|\psi'\rangle[\dots,1]=\sqrt{p/2}|\psi_2\rangle$\;
$|\psi'\rangle\gets\text{swapaxis}(|\psi'\rangle, 1, n) $\;
\end{algorithm}

\bibliographystyle{ucsd}
\bibliography{ref}  
 
\end{document}